\documentclass[journal]{IEEEtran}

\usepackage{color,array,amsthm}
\usepackage{graphicx}
\usepackage{amsmath,amssymb,amsfonts}
\usepackage{algorithm}
\usepackage{algorithmic}
\usepackage{placeins} 
\usepackage[dvipsnames]{xcolor}
\usepackage[caption=false,font=normalsize,labelfont=sf,textfont=sf]{subfig}
\usepackage[hidelinks]{hyperref}
\usepackage{amsthm}
\usepackage{textcomp}
\usepackage{stfloats}
\usepackage{url}
\usepackage{verbatim}
\usepackage{graphicx}
\usepackage{graphicx}
\usepackage{booktabs}
\usepackage{adjustbox}
\usepackage{times}
\usepackage[noadjust]{cite}
\usepackage{tikz}
\usetikzlibrary{arrows.meta, positioning, fit, calc, backgrounds, matrix}

\newtheorem{lemma}{Lemma}
\newtheorem*{remark}{Remark}
\newtheorem{definition}{Definition}
\newtheorem{corollary}{Corollary} 
\newtheorem{theorem}{Theorem}

\begin{document}

\title{Critical Infrastructure Defense Against Aerial Swarms Under Sensing Uncertainty: Online Allocation With Finite-Time Guarantees}
\author{Shriya Pandey, Devaprakash Muniraj
\thanks{This work was supported by the Ministry of Electronics and Information Technology, India under grant no. SP24251719AEIIIT008980 and
IIT Madras under grant no. RF24250403AENFIG008980 (Corresponding author: Devaprakash Muniraj)} \thanks{Shriya Pandey is a Ph.D. candidate, and Devaprakash Muniraj is an Assistant Professor, both with the Department of Aerospace Engineering,
Indian Institute of Technology Madras, Chennai 600036, India; (E-mail: ae23d008@smail.iitm.ac.in; deva@smail.iitm.ac.in).}}

% \markboth{Journal of \LaTeX\ Class Files,~Vol.~18, No.~9, September~2020}%
% {How to Use the IEEEtran \LaTeX \ Templates}

\maketitle

\begin{abstract}

This article presents a closed-loop, uncertainty-aware solution to the problem of defending a protected zone from coordinated incursions by swarms of small uncrewed aircraft systems (UAS). We model the attackers’ interaction structure as time-varying and account for imperfect sensing by the defenders. The proposed criticality-driven defender-to-attacker assignment (CRIDAA) strategy integrates three key components: 1) a probabilistic graph-based representation of the attacking swarm inferred from imperfect sensing; 2) a risk-aware attacker criticality model that couples time-to-breach urgency with uncertainty; 3) an online defender allocation mechanism that assigns and selectively reassigns defenders while explicitly limiting switching-induced instability through robust execution constraints. We establish analytical guarantees within a filtration-based first-hitting-time framework. Specifically, we prove finite-time triggering of the first capture event following initial detection. We further derive explicit mixed linear–geometric upper bounds on the expected neutralization time, separating pre-detection delay from post-detection depletion as a function of defender parallel capacity, execution fraction, and robust capture probability. We demonstrate the effectiveness of the proposed approach through Monte Carlo simulations. The proposed framework achieves $85.6\%$ neutralization efficiency under probabilistic sensing and $99.9\%$ under deterministic sensing. Systematic ablation and sensitivity studies are conducted to quantify how detection gating and coordination/assignment parameters influence reliability and time-to-first-capture.

\end{abstract}


\section{INTRODUCTION}

\IEEEPARstart{T}{hreats} posed by small uncrewed aircraft systems (UAS), particularly platforms weighing less than $25\,\mathrm{kg}$, are now widely recognized due to numerous reported incidents involving unauthorized incursions into protected airspace. Such incursions could be inadvertent (e.g., operator error or navigation loss) or intentional (e.g., malicious intrusion). In both cases, they pose safety, security, and privacy concerns to the public and to critical infrastructure. To address these threats, a broad spectrum of counter-UAS (C-UAS) technologies has been investigated, spanning detection, tracking, and  neutralization capabilities. While several commercial C-UAS systems are currently available and actively deployed to safeguard critical assets, the development of reliable, scalable, and cost-effective systems remains an active research area.

While a single small UAS can endanger safety and privacy, its ability to cause large-scale physical damage to hardened critical infrastructure is often limited; particularly given the growing availability of C-UAS systems. However, an emerging and substantially more challenging threat is an aerial swarm attack, in which a large number of inexpensive UAS cooperate to intrude and disrupt (and potentially damage) a protected facility. Even when the high-level objective of a swarm attack resembles that of a single-UAS intrusion, swarms offer attackers key advantages such as redundancy, saturation of defensive resources, and resilience to partial losses, making defense fundamentally a multi-agent, time-critical sensing-and-allocation problem. This has further accelerated research on C-UAS capabilities spanning detection, tracking, and mitigation, and has motivated broader deployment of defensive systems around critical infrastructure \cite{c1,c3,c4,c5}. Nevertheless, the threat landscape is shifting from isolated intrusions to coordinated adversarial behaviors, requiring defenders to operate under stringent time, sensing, and resource constraints, often with incomplete or uncertain situational awareness \cite{c3,c4,c5}.

% Threats posed by small uncrewed aircraft systems (UAS), particularly platforms below $25\,\mathrm{kg}$, have become a persistent concern due to unauthorized incursions into protected airspace and the attendant safety and security risks \cite{c1,c3,c4,c5}. 

% This has accelerated research on counter-UAS (C-UAS) capabilities spanning detection, tracking, and mitigation, and has motivated the deployment of defensive systems around critical infrastructure \cite{c1,c3,c4,c5}. 
% Nevertheless, the threat landscape is evolving from isolated intrusions to coordinated adversarial behaviors, where the defender must operate under strict time, sensing, and resource constraints  \cite{c3,c4,c5}.

A particularly challenging problem concerns an aerial swarm whose objective is to breach a protected zone (PZ) \cite{c1,c3,c4,c5}. The defense against such attacks naturally leads to a multi-agent decision and resource-allocation problem characterized by (i) a time-varying interaction structure among the attackers, (ii) partial and uncertain information about the attackers available to the defenders, and (iii) stringent computational requirements for determining the defenders' actions \cite{c3,c4,c5, c7,c15}. A common assumption in the swarm-defense literature is that defenders possess perfect and reliable sensing of the attacker swarm. Consequently, interactions among attackers are typically modeled using a deterministic graph \cite{c1,c4,c5,c6}. However, operationally, reliable full-state observability of all attackers is rarely available: detections are intermittent, localization is noisy, and communication links in aerial networks can be disrupted, resulting in a time-varying and uncertain interaction structure \cite{c4,c5}. This motivates the need to model the attacker swarm as a \emph{dynamic graph with uncertainty}, where node attributes and edge existence/weights must be inferred from imperfect sensing, rather than assumed to be perfectly known. 

A second gap in the aerial swarm defense literature concerns the notion of \emph{critical nodes} within an attacker swarm. Identifying the agents that are essential to the swarm’s effectiveness and targeting these critical nodes is a key strategy for neutralizing the swarm. While network science offers a rich set of tools such as centrality and related measures for identifying influential nodes \cite{c8,c9,c13}, these metrics are not directly applicable to the problem of defending a PZ. In particular, they typically overlook defender kinematics, time-to-breach urgency, and sensing uncertainty, all of which are decisive in closed-loop engagement outcomes. This motivates the need for developing a criticality metric that (i) remains well-defined under partial observability, (ii) explicitly accounts for operational urgency and feasibility, and (iii) integrates seamlessly with online defender-assignment~strategies.

A third challenge is that realistic swarm defense is inherently \emph{closed-loop}: the defender team must repeatedly update its engagement plan as detections appear/disappear, localization uncertainty evolves, and the attacker interaction structure changes over time. 
In this setting, defender-to-attacker allocations must be revised online; however, frequent reassignment can induce switching losses (wasted transit and missed engagement opportunities) and can adversely interact with low-level safety constraints, whereas overly conservative reassignment may delay response to imminent breaches.

Existing works on aerial swarm defense address the problem through a range of formalisms, including interception and herding-based strategies for multi-attacker containment, learning-based and game-theoretic formulations for multi-agent decision-making for attack-defense confrontation \cite{c1,c4,c5,c7,c15}. Collectively, these works demonstrate that coordinated defender policies can effectively counter swarms, highlighting key mechanisms such as leveraging defender speed advantage, distributing defensive effort across threats, and shaping attacker trajectories through containment strategies \cite{c4,c5,c15}.

Most existing formulations rely on idealized models of attacker interaction and assume reliable sensing and communication \cite{c1,c4,c5,c7}. In addition, they abstract the decision layer in ways that do not explicitly regulate reassignment transients under intermittent detections. Consequently, the coupling between uncertainty-driven situational awareness updates and allocation switching is often not treated as a primary design objective, despite its direct impact on realized engagement throughput and breach incidence.

This work addresses the aforementioned challenges by developing a comprehensive defense framework for protecting critical infrastructure against aerial swarm attacks using a team of defender UAS. We move beyond deterministic graph models and perfect-information assumptions by explicitly modeling sensing uncertainty through probabilistic node features and uncertainty-aware edge formation. This formulation enables online reasoning over probabilistic, time-varying graphs induced by imperfect sensing by defenders. We further introduce hybrid attacker criticality metrics that integrate network-structural attributes with operational risk factors such as time-to-breach and probabilistic risk-to-PZ, thereby bridging the gap between purely structural importance measures and mission-constrained defense decision-making.

To facilitate near-real-time deployment, we incorporate Markov-chain-based risk prediction within a centralized defender-to-attacker assignment pipeline that enables reassignment and controlled switching. This design allows the defense strategy to adapt as detections appear or disappear and as the attacker interaction topology evolves. Beyond algorithmic development, we establish analytical guarantees that characterize sufficient conditions for persistent defensive progress under uncertainty. These guarantees are complemented by systematic simulation studies, including ablation and sensitivity analyses, to quantify robustness across diverse sensing conditions and agent kinematics. Fig.~\ref{fig:graph_abst} illustrates the proposed \emph{online, windowed} criticality-driven defender-to-attacker assignment (CRIDAA) framework.

\noindent The main contributions of this work are given below:
\begin{enumerate}
\setlength{\itemsep}{2pt}
\item We propose an uncertainty-aware, probabilistic dynamic-graph defense framework for countering aerial swarm attacks, relaxing the deterministic interaction-graph and perfect-sensing assumptions typically adopted in existing swarm defense works.

\item We develop hybrid criticality metrics for attacker prioritization that couple network information with operational urgency (e.g., time-to-breach and risk-to-PZ), addressing the mismatch between purely structural centrality measures and mission-constrained PZ defense under partial observability. 

\item We design a centralized defender allocation and replanning mechanism with explicit reassignment and switching, enabling defender-to-attacker assignment under time-varying risk estimates. 

\item We establish analytical guarantees for the proposed defense architecture and validate them through systematic Monte Carlo evaluation, ablation analysis, and sensitivity analysis, thereby providing both analytical and empirical evidence of robustness under uncertainty and resource constraints. 
\end{enumerate}

\begin{figure*}
    \centering
    \resizebox{\textwidth}{!}{%
    \includegraphics{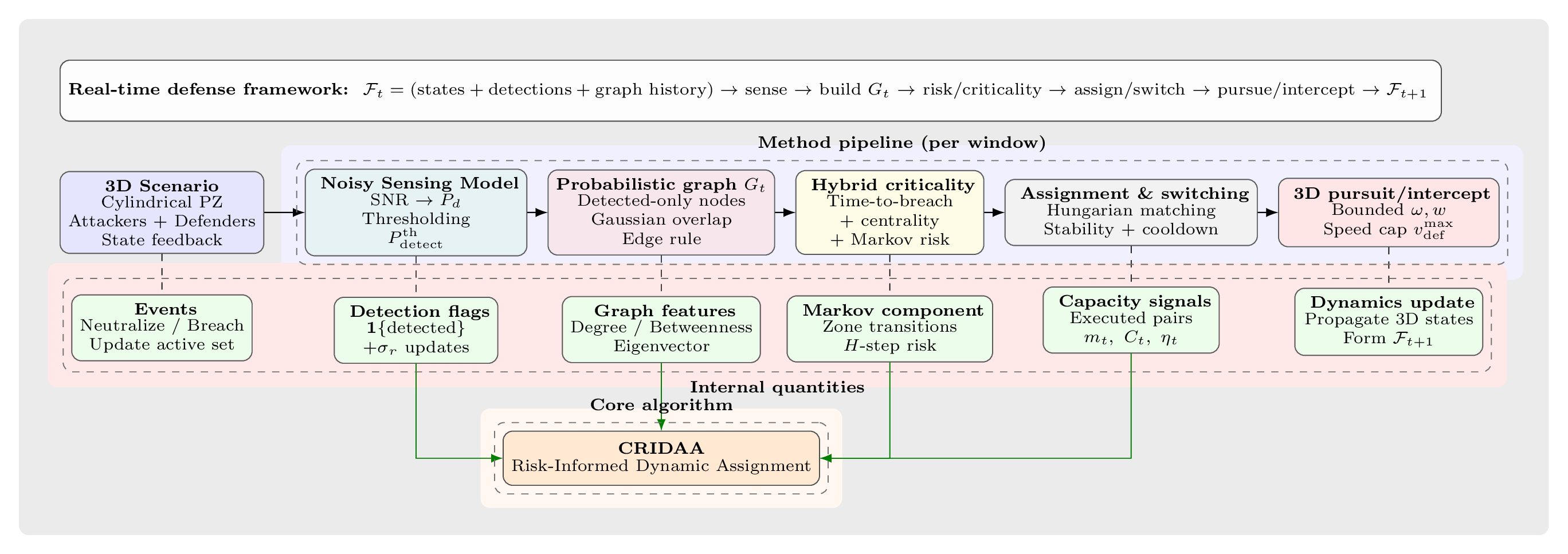}
    }%
    \caption{Overview of the proposed real-time CRIDAA defense framework.}
    \label{fig:graph_abst}
\end{figure*}

\section{Problem Formulation}
\label{sec:problem_formulation}

We consider a team of defenders tasked with protecting a critical infrastructure against an incoming swarm of attacker UAS, and the critical infrastructure is assumed to be contained within a cylindrical protected zone of radius $r_{hard}$ and height $h$. To keep the subsequent methodology, theoretical analysis, and simulation studies consistent, we explicitly distinguish \emph{continuous time} $t\in\mathbb{R}_{\ge 0}$ processes (e.g. vehicle dynamics) from \emph{decision windows} indexed by $k\in\mathbb{Z}_{\ge 0}$ (e.g. attacker detection, assignment).

\subsection{System and Threat Model}
\label{subsec:system_model}

\paragraph*{Protected Region}
Assuming an inertial reference frame defined in the plane $z=0$ with its origin at the center of the cylindrical area and all the UAS operating within a 3D environment $\mathcal{W} \subset \mathbb{R}^{3}$,  the 3D protected zone ($\mathcal{P} \subset \mathcal{W}$) is defined as
$$\mathcal{P}=\{\mathbf{r}\in\mathbb{R}^{3}:\sqrt{x^{2}+y^{2}}\leq r_{hard},0\leq z\leq h\}$$
where $\mathbf{r} = [x \ y \ z]^T.$ 
We define the \emph{hard} protected region (breach boundary) and an outer \emph{soft} boundary as
\begin{align}
\mathcal{B}_{\text{hard}}
&:= \bigl\{(x,y,z)\in\mathbb{R}^3{:}\ x^2+y^2 \le r_{\text{hard}}^2, 0\,{\le}\, z\,{\le}\, h \bigr\} \label{eq:B_hard}\\
\mathcal{B}_{\text{soft}}
&:= \bigl\{(x,y,z)\in\mathbb{R}^3{:}\ x^2+y^2 \le r_{\text{soft}}^2, 0\,{\le}\, z\,{\le}\, h \bigr\} \label{eq:B_soft}
\end{align}
where $0<r_{\text{hard}}<r_{\text{soft}}$.

\paragraph*{Agents and States}
Let $\mathcal{A}(t)$ denote the index set of \emph{active} attackers at time $t$, and let $\mathcal{D}:=\{1,\dots,|D|\}$ be the index set denoting the defenders. 
The state vectors of attacker $i\in\mathcal{A}(t)$ and defender $j\in\mathcal{D}$ at time $t$ are given by
\begin{align*}
x_{A,i}(t) &:= [p_{A,i}(t) \ \psi_{A,i}(t)]^T \ \text{and} \\ x_{D,j}(t) &:=  [p_{D,j}(t) \ \psi_{D,j}(t)]^T, 
\end{align*}
where $p_{A,i}(t)\in\mathbb{R}^3$ and $p_{D,j}(t)\in\mathbb{R}^3$ denote the positions of attacker $i$ and defender $j$, respectively. $\psi_{A,i}(t)\in\mathbb{S}^1$ and $\psi_{D,j}(t)\in\mathbb{S}^1$ denote their corresponding heading angles. 

\paragraph*{Motion Model}
We model both the attacker and defender UAS using a Dubins-like model with bounded turn-rate and vertical velocity inputs. The dynamics of attacker $i\in\mathcal{A}(t)$ and defender  $j\in\mathcal{D}$ are governed by following equations:
\begin{align}
\dot p_{A,i}(t) &=
[
v_{A,i}\cos \psi_{A,i}(t)\
\;
v_{A,i}\sin \psi_{A,i}(t)\
\;
w_{A,i}(t)
]^T, \notag\\
\dot\psi_{A,i}(t) & =\omega_{A,i}(t),
\label{eq:attacker_dynamics}
\end{align}
\begin{align}
\dot p_{D,j}(t) &=
[
v_{D,j}\cos \psi_{D,j}(t)\
\;
v_{D,j}\sin \psi_{D,j}(t)\
\;
w_{D,j}(t)
]^T,\notag \\
\dot\psi_{D,j}(t) &=\omega_{D,j}(t),
\label{eq:defender_dynamics}
\end{align}
where the control inputs satisfy box constraints
\begin{align}
      & |\omega_{A,i}(t)|\le \bar\omega_A,\ \ |w_{A,i}(t)|\le \bar w_A, ,\\
      & |\omega_{D,j}(t)|\le \bar\omega_D,\ \ |w_{D,j}(t)|\le \bar w_D .
\label{eq:input_bounds}
\end{align}
For each attacker $i\in\mathcal{A}(t)$ and defender $j\in\mathcal{D}$, the planar speeds $v_{A,i}$ and $v_{D,j}$ in \eqref{eq:attacker_dynamics}--\eqref{eq:defender_dynamics} are assumed \emph{constant in time} for a given agent (i.e., constant-speed Dubins-like kinematics in the horizontal plane). 
However, the speeds are allowed to take values in the following compact sets:
\begin{equation*}
v_{A,i} \in [\underline{v}_A,\overline{v}_A] 
\ \text{and} \ 
v_{D,j} \in [\underline{v}_D,\overline{v}_D], \ \forall \ i\in\mathcal{A}(t) \ \text{and} \ j\in\mathcal{D}.
\end{equation*}

In addition, we assume the following: (i) defenders move faster than attackers, i.e., $v_{D,j} > v_{A,i}$, to enable feasible interception; (ii) a ground-based UAS detection system is present within or on the boundary of the PZ; (iii) a central node is present within the PZ that interfaces with the detection system and the defenders to inform the defenders of the incoming threats and attacker states, thereby facilitating real-time situational awareness for defenders; (iv) the communication channels between the central node and the defenders are fully encrypted and secure, ensuring robust and reliable data transmission. 

\paragraph*{Breach and Neutralization}
An attacker $i$ is said to \emph{breach} the PZ when $p_{A,i}(t)\in\mathcal{B}_{\text{hard}}$ is first satisfied. A defender--attacker pair $(j,i)$ achieves \emph{neutralization} when the relative distance enters a capture ball defined by 
\begin{equation}
\|p_{D,j}(t)-p_{A,i}(t)\|\le r_{\text{cap}},
\label{eq:capture_condition}
\end{equation}
where $r_{\text{cap}}>0$ is the capture radius.
% \paragraph*{Control Objective}
The defender team's objective is to prevent any breach while neutralizing all attackers as quickly as possible, subject to dynamics, input limits, and safety constraints (e.g., collision avoidance). The attacker control laws are treated as part of the threat/environment model used to evaluate the defenders' policies. The focus of this paper is the defender-side closed-loop architecture, which comprises detection, probabilistic graph construction, risk assessment, defender–attacker assignment, and pursuit, and is detailed in the subsequent sections.

\subsection{Detection and Sensing Model}
\label{subsec:detection_model}

We assume a sensing system that provides observations (with uncertainty) of attacker positions. This section defines the \emph{detected set} used by the graph model. The sensed position of attacker $i$, denoted by $\hat p_{A,i}(t)$, is given by $\hat p_{A,i}(t) = p_{A,i}(t) + \varepsilon_i(t)$, $\varepsilon_i(t)\sim\mathcal{N}(0,\Sigma_i(t))$, 
where $p_{A,i}(t)$ is the true position of attacker $i$ and $\Sigma_i(t)\in\mathbb{S}^3_{\succeq 0}$ captures the sensing uncertainty. We develop our results under both deterministic as well as probabilistic sensing assumptions.

\paragraph*{Deterministic Detection}
 Deterministic sensing is characterized by idealized conditions with no uncertainty in sensor measurements. Following the binary detection model established in \cite{c17}, we define the sensing capability for a sensor positioned at $p_s\in\mathbb{R}^3$ based on a geometric range constraint, where a target is successfully detected if $\|\hat p_{A,i}(t)-p_s\| \le r_{\text{det}}$, where $r_{\text{det}}>0$ is the detection range.

\paragraph*{Probabilistic Detection}
Inspired in part by \cite{c17}, where detection probability is modeled as an exponentially decaying function of sensing distance, we adopt a similar exponential attenuation structure. To better capture real-world sensing uncertainties, we explicitly incorporate the signal-to-noise ratio (SNR) into the model, allowing the detection probability to reflect range-dependent degradation. The resulting detection probability is given by
\begin{equation}
P_d(i,t) = \exp\!\Bigl(-\frac{\|\hat p_{A,i}(t)-p_s\|^2}{2\sigma_r^2(i,t)}\Bigr)\,\bigl(1-P_{\mathrm{FN}}(i,t)\bigr),
\label{eq:pd_general}
\end{equation}
where $\sigma_r(i,t)>0$ is an uncertainty scale that can increase with range and $P_{\mathrm{FN}}(i,t)\in[0,1]$ is the false-negative probability. The false negative probability $P_{FN}$ represents the probability that $A_i$ is within the detection range but is not detected by the ground-based sensor system due to signal fluctuations, noise, or a low SNR.

The SNR-threshold model for $P_{\mathrm{FN}}$ is given by
\begin{equation}
P_{\mathrm{FN}}(i,t) = 1 - Q\!\left(\frac{\Gamma-\mathrm{SNR}_i(t)}{\sigma_n}\right),
\label{eq:pfn_q}
\end{equation}
where $\Gamma$ is a detection SNR threshold, $\sigma_n>0$ is a noise scaling constant, and $Q(\cdot)$ is the standard Gaussian tail function, which characterizes the probability that the received signal falls below the prescribed threshold and is therefore missed by the detection system. Low SNR values degrade detection performance, resulting in fewer attacker UAS being identified and consequently sparser graph connectivity, whereas higher SNR improves detection reliability.

\paragraph*{Detected Set}
Given a fixed decision window length $\Delta>0$, we define the decision epochs as the time intervals between the time instants $t_k=k\Delta$. At each epoch $k$, we declare that attacker $i$ is detected if $P_d(i,t_k)\ge P_{\text{det}}^{\text{th}}$, and define the detected active set as
\begin{equation}
\widehat{\mathcal{A}}_k
:= \bigl\{ i\in\mathcal{A}(t_k)\ :\ P_d(i,t_k)\ge P_{\text{det}}^{\text{th}} \bigr\}.
\label{eq:detected_set}
\end{equation}
This set is the input to the probabilistic interaction graph module and the defender-to-attacker assignment algorithm as shown in Fig.~\ref{fig:graph_abst}.

% ============================================================
\section{Deterministic Graph-Theoretic Modeling}
\label{sec:det_graph_model}
% ============================================================
This section provides an idealized (deterministic) attacker interaction model that is used for (i) defining threat-evolution features, most notably time-to-breach, and (ii) establishing a deterministic limit case for the probabilistic graph model of Section~\ref{sec:prob_graph_model}.
The formulation provided in this section will serve as an ideal-sensing baseline.

\subsection{Communication Network Model}
\label{subsec:comm_network_model}

Let $\mathcal{A}(t)$ denote the set of active attackers at time~$t$.
In this section, we model attacker-to-attacker communication by an undirected, time-varying graph
\begin{equation}
  \mathcal{G}_A(t) := \bigl(\mathcal{V}_A(t),\mathcal{E}_A(t)\bigr),
  \qquad
  \mathcal{V}_A(t) := \mathcal{A}(t).
  \label{eq:det_graph_def}
\end{equation}
Given a communication radius $R_{\mathrm{com}}>0$, we define the edge set as follows:
\begin{equation}
  \!\!\!\!(i,\ell)\,{\in}\,\mathcal{E}_A(t)
 \, {\Longleftrightarrow} \,
  \bigl\|p_{A,i}(t){-}p_{A,\ell}(t)\bigr\| \,{\le}\, R_{\mathrm{com}},
  \, \forall i\neq \ell. \!
  \label{eq:det_edges}
\end{equation}
The adjacency matrix corresponding to $\mathcal{G}_A(t)$ is denoted by $A_A(t)\in\{0,1\}^{|\mathcal{V}_A(t)|\times|\mathcal{V}_A(t)|}$, and the $(i,\ell)$ element of $A_A(t)$ is given by
\begin{equation}
  (A_A(t))_{i\ell} =
  \begin{cases}
    1, & (i,\ell)\in\mathcal{E}_A(t) \ \text{and} \  i \neq \ell\\
    0, & \text{otherwise}
  \end{cases}
  .
  \label{eq:adjacency_det}
\end{equation}
The degree matrix and the (combinatorial) Laplacian matrix are given by
\begin{align*}
D_A(t) & =\mathrm{diag}\bigl(\sum_{\ell}(A_A(t))_{i\ell}\bigr), \\
  L_A(t) & = D_A(t) - A_A(t).
  \label{eq:laplacian_det}
\end{align*}
The algebraic connectivity $\lambda_2\!\bigl(L_A(t)\bigr)$ is used as a metric to quantify the connectivity of the swarm. Under the assumption of perfect detection, the probabilistic graph (introduced in Section~\ref{sec:prob_graph_model}) reduces to the aforementioned deterministic graph with unit edge weights. Consequently, the deterministic model becomes equivalent to the proximity-graph baseline employed in the simulations.

\subsection{Time-To-Breach and Nominal Attacker Guidance}
\label{subsec:attacker_time_optimal}

The objective of an active attacker $i\in\mathcal{A}(t)$ is to reach the hard boundary $\mathcal{B}_{\mathrm{hard}}$ in minimum time. The state and control input vectors of attacker $i$ at time $t$ are given by
\begin{align}
  x_{A,i}(t) & =
  [
    p_{A,i}(t)\ \; \psi_{A,i}(t)
 ]^T \\
  % \in\mathbb{R}^3\times\mathbb{S}^1,
  u_{A,i}(t)  & =
  [
    \omega_{A,i}(t)\ \; w_{A,i}(t)
  ]^T
\end{align}
where $\omega_{A,i}(t) \in [-\bar\omega_A,\bar\omega_A]$ and $w_{A,i}(t) \in [-\bar w_A,\bar w_A]$. The equations governing the attacker dynamics can be written in control-affine form as
\begin{equation}
  \dot x_{A,i}(t) = f_A\!\bigl(x_{A,i}(t)\bigr) + g_A\!\bigl(x_{A,i}(t)\bigr)\,u_{A,i}(t),
  \label{eq:att_control_affine}
\end{equation}
where the drift term $f_A(\cdot)$ models the constant-speed planar motion parameterized by $\psi_{A,i}$ and $g_A(\cdot)$ is the control input matrix.

% \paragraph*{Minimum-Time Formulation}
The motion of attacker $i$ at a given time $t$ is governed by the solution to the following minimum-time optimal control problem:
\begin{subequations}
\label{eq:att_min_time}
\begin{align}
  & \min_{u_{A,i}(\cdot),\,T_i\ge 0}\quad
   \int_{t}^{t+T_i} 1\,d\tau 
  \label{eq:att_min_time_cost}\\
  & \text{s.t.}\ 
  \dot x_{A,i}(\tau) \,{=}\, f_A(x_{A,i}(\tau)) + g_A(x_{A,i}(\tau))u_{A,i}(\tau),
  \label{eq:att_min_time_dyn}\\
  & u_{A,i}(\tau)\in \mathcal{U}_A, \quad p_{A,i}(t+T_i)\in \mathcal{B}_{\mathrm{hard}}, 
  \label{eq:att_min_time_terminal}
\end{align}
\end{subequations}
where $ \mathcal{U}_A := [-\bar\omega_A,\bar\omega_A]\times[-\bar w_A,\bar w_A]$. The optimal value $T_i^\star(t)$ (or a computable approximation thereof) defines the time-to-breach used in the multi-layer criticality model explained in the sequel.
Moreover, the short-horizon predicted trajectories induced by the nominal solution provide the kinematic forecasts required for the evaluation of defender intercept costs in the defender-assignment~layer.

% \paragraph*{Nominal Guidance}
Let $u_{A,i}^{\mathrm{nom}}(t)$ denote the nominal reference trajectory obtained by solving \eqref{eq:att_min_time}. Safety and coordination requirements--specifically collision avoidance and connectivity maintenance--are not incorporated as penalty terms within the time-optimal formulation. Instead, they are enforced through a dedicated constraint filter described next. This design yields a clear algorithmic structure: first computing the nominal reference trajectory, and subsequently enforcing safety and coordination constraints by projection through a quadratic program (QP).

\subsection{Non-Optimal Attacker Dynamics}
\label{subsec:nonoptimal_attackers}

To assess robustness beyond the nominal time-optimal threat model, we also consider suboptimal control laws for each active attacker $i\in\mathcal{A}(t)$. The inputs $(\omega_{A,i}(t),w_{A,i}(t))$ follow one of four fixed heuristics: biased random, sinusoidal weaving, flanking arc, or pure random motion.

\noindent For biased random attacker dynamics,
\begin{align*}
\omega_{A,i}(t)=\mathrm{sat}_{\bar\omega_A}\!\big(k_\psi(\psi_i^\star-\psi_{A,i}(t))+\eta_i(t)\big),\\
w_{A,i}(t)=\mathrm{sat}_{\bar w_A}\!\big(k_z(z_c-z_{A,i}(t))+\xi_i(t)\big),
\end{align*}
where $\psi_i^\star=\operatorname{atan2}(-y_{A,i}(t),-x_{A,i}(t))$ toward $\mathcal{P}$, $z_c=h/2$, and $\eta_i(t),\xi_i(t)$ are zero-mean Gaussian noise. Sinusoidal weaving replaces $\psi_i^\star$ by
\begin{equation*}
\psi_i^\star(t)=\operatorname{atan2}(-y_{A,i}(t),-x_{A,i}(t))+A\sin(\nu t+\phi_i),
\end{equation*}
where $A,\nu>0$ are fixed amplitude and frequency, and $\phi_i\in[0,2\pi)$ is a per-attacker phase offset. Flanking arc blends tangential and inward headings:
\begin{equation*}
\psi_i^\star(t)=(1-\lambda_i(t))\big(\psi_i^\star(t)\pm\pi/2\big)+\lambda_i(t)\psi_i^\star(t),
\end{equation*}
where $\lambda_i(t)\in[0,1]$ increases as $p_{A,i}(t)$ approaches $\mathcal{B}_{\text{hard}}$. Pure random samples uniformly from the bounds in \eqref{eq:input_bounds}:
\begin{equation*}
\omega_{A,i}(t)\sim\mathcal{U}[-\bar\omega_A,\bar\omega_A],\quad
w_{A,i}(t)\sim\mathcal{U}[-\bar w_A,\bar w_A].
\end{equation*}

\subsection{Constraint Filter via CBF--QP}
\label{subsec:cbf_filter}

Inspired by \cite{c16}, we enforce safety-relevant invariance conditions through a control-barrier-function quadratic program (CBF--QP) that minimally perturbs the nominal reference trajectory described in the preceding section. Specifically, at each time $t$, the input is given by
\begin{equation}
  u_{A,i}(t) = u_{A,i}^{\mathrm{nom}}(t) + \Delta u_{A,i}(t),
  \label{eq:filter_update}
\end{equation}
where $\Delta u_{A,i}(t)$ is computed by solving a convex QP subject to the barrier constraints described below.

\subsubsection*{Collision-Avoidance Barrier Constraints}
\label{subsubsec:collision_cbf}

Given two distinct attackers $i$ and $\ell$, we define the following barrier function:
\begin{equation}
  h_{i\ell}(x_{A,i}(t)) = \|p_{A,i}(t)-p_{A,\ell}(t)\|^2 - r_{\mathrm{col}}^2,
  \label{eq:h_col}
\end{equation}
where $r_{\mathrm{col}}>0$ is a prescribed separation distance.
The corresponding safe set is given by $\mathcal{C}_{\mathrm{col}} := \{x_{A,i}(t):\ h_{i\ell}(x_{A,i}(t))\ge 0,\ \forall i\neq \ell\}$.
A standard zeroing-CBF condition is
\begin{equation}
  \dot h_{i\ell}(x_{A,i}(t)) + \kappa_{\mathrm{col}}\,h_{i\ell}(x_{A,i}(t)) \ge 0,
  \qquad \kappa_{\mathrm{col}}>0,
  \label{eq:zcbf_col}
\end{equation}
which yields an affine inequality in the decision variables once $\dot h_{i\ell}$ is expressed via the control-affine dynamics \eqref{eq:att_control_affine}.

\subsubsection*{Connectivity-Maintenance Barrier Constraints}
\label{subsubsec:connectivity_cbf}

We define the following differentiable function to enforce connectivity among the agents $i$ and $\ell$:
\begin{equation}
  w_{i\ell}(t) = \exp\!\Bigl(-\frac{\|p_{A,i}(t)-p_{A,\ell}(t)\|^2}{\sigma_c^2}\Bigr),
  \quad \sigma_c>0.
  \label{eq:wij_smooth}
\end{equation}
The weighted Laplacian can then be constructed as
\begin{multline}
   L_w(t) = D_w(t) - W(t), \quad
   W(t)=[w_{i\ell}(t)],\\
   D_w(t)=\mathrm{diag}\Bigl(\sum_{\ell} w_{i\ell}(t)\Bigr).
   \label{eq:weighted_laplacian}
\end{multline}
The connectivity-maintenance barrier function and the associated constraint are defined as 
\begin{align*}
  & h_{\mathrm{con}}(t) = \lambda_2\!\bigl(L_w(t)\bigr) - \lambda_{\min},
  \quad \lambda_{\min}>0, \\
  & \dot h_{\mathrm{con}}(t) + \kappa_{\mathrm{con}}\,h_{\mathrm{con}}(t) \ge 0,
  \quad \kappa_{\mathrm{con}}>0. 
\end{align*}

The preceding constraint ensures that the algebraic connectivity of the attacker network remains above the prescribed threshold $\lambda_{\min}$ under the smooth weighted-graph surrogate formulation.

Collecting all incremental control input variables as $\Delta u_A(t):=\{\Delta u_{A,i}(t)\}_{i\in\mathcal{A}(t)}$, the vector $\Delta u_A(t)$ can be computed by solving the following optimization problem:
\begin{subequations}
% \label{eq:cbf_qp}
\begin{align}
  & \min_{\{\Delta u_{A,i}\}} \quad
  \sum_{i\in\mathcal{A}(t)} \bigl\|\Delta u_{A,i}\bigr\|_2^2
  \label{eq:cbf_qp_obj}\\
  \text{s.t.}\quad
  & \dot h_{i\ell}(x_{A,i}(t)) + \kappa_{\mathrm{col}}\,h_{i\ell}(x_{A,i}(t)) \ge 0,
  \quad \forall\, i\neq\ell,
  \label{eq:cbf_qp_col}\\
  & \dot h_{\mathrm{con}}(t) + \kappa_{\mathrm{con}}\,h_{\mathrm{con}}(t) \ge 0,
  \label{eq:cbf_qp_con}\\
  & u_{A,i}^{\mathrm{nom}}(t)+\Delta u_{A,i}(t)\in \mathcal{U}_A,
  \qquad \forall\, i\in\mathcal{A}(t).
  \label{eq:cbf_qp_bounds}
\end{align}
\end{subequations}
The objective \eqref{eq:cbf_qp_obj} enforces minimal deviation from the nominal guidance, while \eqref{eq:cbf_qp_col}--\eqref{eq:cbf_qp_con} enforce collision avoidance and connectivity preservation. The outputs of this section comprises: (i) the ideal-sensing deterministic proximity graph $\mathcal{G}_A(t)$ and its associated Laplacian, (ii) a time-to-breach metric $T_i^\star(t)$ (or its computable approximation) used in the criticality model, and (iii) constraint-consistent attacker motion predictions, which are used for probabilistic graph construction in Section~\ref{sec:prob_graph_model} and for evaluating the cost during defender assignment.

% ============================================================
\section{Probabilistic Graph-Theoretic Modeling}
\label{sec:prob_graph_model}
% ============================================================

This section introduces an uncertainty-aware interaction graph for the attacker swarm that (i) explicitly accounts for sensing uncertainty and missed detections, and (ii) provides a weighted network representation suitable for downstream centrality computation and risk-aware defender assignment.

\subsection{Uncertainty Set and Detected Vertices}
\label{subsec:prob_vertices}

Let $\widehat{\mathcal{A}}_k$ denote the detected active attacker set at decision epoch $t_k$ as given in \eqref{eq:detected_set}.
For each $i\in\widehat{\mathcal{A}}_k$, the sensing module provides a Gaussian position estimate given by
\begin{equation}
\hat p_{A,i}(t_k) \sim \mathcal{N}\!\bigl(\mu_i(t_k),\,\Sigma_i(t_k)\bigr),
\label{eq:gaussian_estimate_prob}
\end{equation}
where $\mu_i(t_k)\in\mathbb{R}^3$ and $\Sigma_i(t_k)\in\mathbb{S}^3_{\succeq 0}$. To each detected attacker $i$, we associate a window-level confidence ellipsoid defined by
\begin{equation}
\mathcal{U}_i(t_k)
:=
\Bigl\{x\,{\in}\,\mathbb{R}^3: e_i(t_k)^\top \Sigma_i(t_k)^{-1} e_i(t_k) \le c_0^2 \Bigr\},
\label{eq:ellipsoid}
\end{equation}
where $e_i(t_k)=x-\mu_i(t_k)$ and $c_0^2$ is chosen as the $\chi^2$ quantile with $3$ degrees of freedom at a confidence of $0.95$ (numerically, $c_0^2\approx 7.81$). 
In the simulation pipeline used for ablation and sensitivity studies, we set $c_0^2=7.8147$ to match $\chi^2_{3}(0.95)$. The probabilistic attacker graph at epoch $k$ is defined on the detected vertex set as
\begin{equation}
G_k = \bigl(V_k,E_k,W_k\bigr), 
\qquad V_k = \widehat{\mathcal{A}}_k,
\end{equation}
where $W_k$ encodes edge weights that represent communication likelihood under uncertainty.

\subsection{Overlap-Based Probabilistic Edges}
\label{subsec:overlap_edges}

For $i,\ell\in V_k$ with $i\neq \ell$, let $f_i(\cdot\,;t_k)$ and $f_\ell(\cdot\,;t_k)$ denote the Gaussian PDFs induced by \eqref{eq:gaussian_estimate_prob}.
We quantify spatial co-location likelihood via the \emph{overlap coefficient}
% (Bhattacharyya-type overlap functional)
\begin{equation}
V_{i\ell}(t_k)
=
\int_{\mathbb{R}^3} \min\!\bigl\{f_i(x;t_k),\, f_\ell(x;t_k)\bigr\}\,dx,
\label{eq:overlap_def}
\end{equation}
which satisfies $V_{i\ell}(t_k)\in[0,1]$ by construction.

% \subsubsection*{Normalized Link Likelihood and Thresholding.}
Following \cite{c10}, we normalize the edge weights to be dimensionless and consistent across time steps, defined as
\begin{align}
\widetilde{V}_k &= \max_{(i,\ell)\in V_k\times V_k,\ i<\ell} V_{i\ell}(t_k), \nonumber
\\
P_{i\ell}(t_k) &=
\begin{cases}
\dfrac{V_{i\ell}(t_k)}{\widetilde{V}_k}, & \widetilde{V}_k>0,\\[6pt]
\quad 0 \hspace{5mm}, & \widetilde{V}_k=0,
\end{cases}
\label{eq:plink_norm}
\end{align}
so that $P_{i\ell}(t_k)\in[0,1]$.
Given a design parameter $\alpha_{\mathrm{go}}\,{\in}\,(0,1]$, we include an undirected edge if
\begin{equation}
(i,\ell)\in E_k
\ \Longleftrightarrow\
P_{i\ell}(t_k) \ge \alpha_{\mathrm{go}}.
\label{eq:edge_threshold}
\end{equation}
We set the corresponding edge weight as
\begin{equation}
w_{i\ell}(t_k) = P_{i\ell}(t_k),
\qquad
W_k = \bigl[w_{i\ell}(t_k)\bigr].
\label{eq:edge_weight}
\end{equation}

% \paragraph{Numerical approximation}
The integral in \eqref{eq:overlap_def} generally has no closed-form expression for arbitrary $\Sigma_i(t_k)$ and $\Sigma_\ell(t_k)$ when the minimum of PDFs is used; accordingly, we approximate $V_{i\ell}(t_k)$ by numerical quadrature on a 3D voxel grid over the operational volume.  
This construction corresponds to the \texttt{graphmode=overlap} variant used in the reported simulation studies, with $\alpha_{\mathrm{go}}$ corresponding to \texttt{alphagraphoverlap}.  

\subsection{Deterministic Limit Case and Baselines}
\label{subsec:baselines_rewrite}

This subsection formalizes the deterministic limit case implied by Section~\ref{sec:det_graph_model} and the additional graph baselines used in the reported studies.

% \paragraph{Perfect-sensing limit (deterministic disk graph)}
If localization uncertainty collapses ($\Sigma_i(t_k)\rightarrow 0$) and detection is perfect ($V_k=\mathcal{A}(t_k)$), then the interaction graph reduces to a deterministic disk graph with unit edge weights, consistent with \eqref{eq:det_edges} in Section~\ref{sec:det_graph_model}.

% \paragraph{Communication-radius proximity baseline}
To align the notation with Section~\ref{sec:results}, we denote the communication (proximity) radius by $r_{\mathrm{comm}}>0$ and define the proximity baseline graph by
\begin{multline}
\!\!\!\!\!(i,\ell)\,{\in}\, E_k
{\Longleftrightarrow}
\bigl\|\mu_i(t_k){-}\mu_\ell(t_k)\bigr\| \,{\le}\, r_{\mathrm{comm}}, \\  
w_{i\ell}(t_k)\equiv 1,
\label{eq:proximity_graph}
\end{multline}
which corresponds to \texttt{graphmode=proximity}.  
In an ideal-sensing deterministic scenario, \eqref{eq:proximity_graph} can equivalently be written using the true positions as
\begin{equation}
(i,\ell)\in E(t_k)
\ \Longleftrightarrow\
\bigl\|p_{A,i}(t_k)-p_{A,\ell}(t_k)\bigr\| \le r_{\mathrm{comm}},
\label{eq:proximity_graph_truepos}
\end{equation}
thereby recovering the deterministic disk graph \eqref{eq:det_edges} as a special case.
\noindent
The ``deterministic graph-based scenario'' reported in Section~\ref{sec:results} corresponds to the ideal-sensing case \eqref{eq:proximity_graph_truepos} of the proximity baseline \eqref{eq:proximity_graph}. 

% \paragraph{Null-graph baseline.}
Another baseline that is used in Section~\ref{sec:results} pertains to \texttt{graphmode=none}, where $E_k=\emptyset$, i.e., there are no edges between the attackers and only the detected vertex set $V_k$ is retained for downstream processing.  

\subsection{Adjacency, Laplacian, and Downstream Usage}
\label{subsec:prob_adjacency}

Let $A_k$ be the weighted adjacency matrix defined by
\begin{equation}
(A_k)_{i\ell} =
\begin{cases}
w_{i\ell}(t_k), & (i,\ell)\in E_k,\\
\quad 0, & \text{otherwise},
\end{cases}
\quad (A_k)_{ii}=0.
\label{eq:weighted_adjacency}
\end{equation}
Let $D_k:=\mathrm{diag}\!\bigl(\sum_{\ell}(A_k)_{i\ell}\bigr)$ be the (weighted) degree matrix, and the weighted Laplacian is defined as
\begin{equation}
L_k := D_k - A_k.
\label{eq:weighted_laplacian_prob}
\end{equation}
The graph $(V_k,E_k,W_k)$ and the associated matrices $(A_k,L_k)$ provide the network substrate for (i) composite centrality measures and (ii) risk/criticality fusion in the next section.

% ============================================================
\section{Multi-Layer Criticality Assessment Framework}
\label{sec:criticality_framework}
% ============================================================

This section defines the attacker criticality score that drives the defender-to-attacker assignment strategy.
At each decision epoch $t_k$, each detected active attacker $i\in V_k=\widehat{\mathcal{A}}_k$ is assigned a score that combines (i) time-to-breach urgency, (ii) network centrality on the attacker interaction graph, and (iii) Markov-chain-based breach risk term over a finite composition horizon, capturing the time-evolving threat via a probabilistic state model \cite{c18}.

\subsection{Time-to-Breach and Boundary Proximity}
\label{subsec:ttb_feature}

For each attacker $i\in\mathcal{A}(t_k)$, let $T_i(t_k)\in[0,\infty]$ denote the (estimated) time required to reach the hard boundary $\mathcal{B}_{\mathrm{hard}}$ under the nominal attacker guidance of Section~\ref{sec:det_graph_model}.
A monotone time-to-breach (TTB) risk map is defined as
\begin{equation}
R_i(t_k) = \frac{1}{1+\beta\,T_i(t_k)},
\qquad \beta>0,
\label{eq:ttb_risk}
\end{equation}
so that $R_i(t_k)\in(0,1]$ and larger values correspond to more urgent threats.

\begin{lemma}
\label{lem:ttb_bounds}
For any $\beta>0$ and any $T_i(t_k)\in[0,\infty]$, the risk $R_i(t_k)$ defined in \eqref{eq:ttb_risk} satisfies $0<R_i(t_k)\le 1$ and is strictly decreasing in $T_i(t_k)$ for finite $T_i(t_k)$.
\end{lemma}
\begin{proof} For $\beta>0$ and $T\ge 0$, $1+\beta T\ge 1$ implies $0<\frac{1}{1+\beta T}\le 1$. Moreover, $\frac{d}{dT}\bigl[(1+\beta T)^{-1}\bigr] = -\beta(1+\beta T)^{-2} < 0$ for all finite $T\ge 0$.
\end{proof}

In addition to TTB, the normalized distance-to-hard-boundary is defined to capture the geometric proximity:
\begin{multline}
d_i(t_k) = \mathrm{dist}\!\bigl(p_{A,i}(t_k),\,\mathcal{B}_{\mathrm{hard}}\bigr),
\\
D_i(t_k) = 1 - \min\!\left\{\frac{d_i(t_k)}{\max_{\ell\in\mathcal{A}(t_k)} d_\ell(t_k)},\,1\right\},
\label{eq:dist_feature}
\end{multline}
so that $D_i(t_k)\in[0,1]$ and larger values correspond to smaller clearance to the hard boundary. Here, $\mathrm{dist}(p,\mathcal{B}_{\mathrm{hard}}):=\inf_{q\in\mathcal{B}_{\mathrm{hard}}}\|p-q\|_2$ denotes the Euclidean point-to-set distance from attacker position $p$ to the hard cylindrical boundary.

\subsection{Composite Centrality Metric}
\label{subsec:centrality_metric}

Let $G_k=(V_k,E_k,W_k)$ be the attacker interaction graph constructed in Section~\ref{sec:prob_graph_model}, with weighted adjacency matrix $A_k$. We use three standard centrality measures, namely, weighted degree centrality, weighted eigenvector centrality, and weighted betweenness centrality, to quantify topological influence. For completeness, we first present these individual metrics before introducing the composite centrality metric.

\noindent The weighted degree centrality of node $i\in V_k$ is 
\begin{equation}
C_D(i;t_k) := \sum_{\ell\in V_k} (A_k)_{i\ell}.
\label{eq:degree_cent}
\end{equation}

\noindent The weighted eigenvector centrality $C_E(\cdot;t_k)\in\mathbb{R}^{|V_k|}_{\ge 0}$ is defined (up to scale) by
\begin{equation}
C_E(\cdot;t_k) = \frac{1}{\lambda_{\max}(A_k)}\,A_k\,C_E(\cdot;t_k),
\label{eq:eig_cent}
\end{equation}
where $\lambda_{\max}(A_k)$ denotes the dominant eigenvalue. The component $C_E(i;t_k)$ is interpreted as the centrality of node $i$.

\noindent  Let $\sigma_{st}(t_k)$ denote the number of shortest paths between nodes $s,t\in V_k$ under an edge-length convention induced by the edge weights, and let $\sigma_{st}(i;t_k)$ denote the number of such paths that pass through $i$.
The betweenness centrality is defined as
\begin{equation}
C_B(i;t_k) = \sum_{\substack{s,t\in V_k\\ s\neq t,\ s\neq i,\ t\neq i}}
\frac{\sigma_{st}(i;t_k)}{\sigma_{st}(t_k)}.
\label{eq:bet_cent}
\end{equation}

% \subsubsection{Centrality fusion}
% \label{subsubsec:centrality_fusion}
Having defined the three metrics, the composite centrality score is defined as
\begin{equation*}
C_{\mathrm{cent}}(i;t_k)
=
\alpha_1\,\widetilde{C}_D(i;t_k)
+\alpha_2\,\widetilde{C}_E(i;t_k)
+\alpha_3\,\widetilde{C}_B(i;t_k),
% \label{eq:cent_fusion}
\end{equation*}
where $\alpha_1,\alpha_2,\alpha_3\ge 0$, $\alpha_1+\alpha_2+\alpha_3=1$, and $\widetilde{(\cdot)}$ denotes a window-level normalization to $[0,1]$ over $V_k$.

\subsection{Markov-Chain-Based Breach Risk}
\label{subsec:markov_risk}

Markov chain models are typically used to model state transitions and systematically assess the risk of unfavorable events \cite{c14}. Herein, we define a Markov model to assess the risk of attacker UAS breaching the hard boundary. This subsection assigns each attacker a finite-horizon breach probability computed from a three-state absorbing Markov chain over zonal regions using a probabilistic state model to capture the evolving threat. Similar formulations have been explored in the HMM-based real-time threat assessment framework in \cite{c18}.

\noindent Let the zone mapping $z(\cdot):\mathbb{R}^3\rightarrow\{0,1,2\}$ be defined~as
\begin{equation}
z(p) =
\begin{cases}
0, & \|[p_x,p_y]^\top\| \ge r_{\mathrm{soft}},\\
1, & r_{\mathrm{hard}} < \|[p_x,p_y]^\top\| < r_{\mathrm{soft}},\\
2, & \|[p_x,p_y]^\top\| \le r_{\mathrm{hard}},
\end{cases}
\label{eq:zone_map}
\end{equation}
where zone $2$ pertains to a breach (hard boundary reached) and is treated as absorbing. For each attacker $i$, we define a discrete-time process $X_i(k)\in\{S_0,S_1,S_2\}$ with $X_i(k)=S_{z(p_{A,i}(t_k))}$ and with $S_2$ absorbing.

At epoch $k$, the following window-level transition matrix $P_i(t_k)\in[0,1]^{3\times 3}$ is defined: 
\begin{multline}
P_i(t_k) =
\begin{bmatrix}
p_{00,i}(t_k) & p_{01,i}(t_k) & p_{02,i}(t_k)\\
p_{10,i}(t_k) & p_{11,i}(t_k) & p_{12,i}(t_k)\\
0             & 0             & 1
\end{bmatrix},
\\
\sum_{b=0}^2 p_{ab,i}(t_k)=1,\ a\in\{0,1\}.
\label{eq:Pi_def}
\end{multline}

\noindent where the entries $p_{ab,i}(t_k)$ are obtained from the localization uncertainty set $\mathcal{U}_i(t_k)$ in \eqref{eq:ellipsoid} as follows.
Let $\bar p_{i}(t_k{+}\Delta)$ be a one-step nominal prediction of the attacker position over $[t_k,t_{k+1})$ under the nominal attacker guidance.
We then draw i.i.d.\ samples $\xi^{(n)}\in\mathcal{U}_i(t_k)$ and map the perturbed end-points $\bar p_{i}(t_k+\Delta)+\xi^{(n)}$ to zones via \eqref{eq:zone_map}.
The empirical frequencies are used to construct the row probabilities in \eqref{eq:Pi_def}, followed by row normalization.

% \paragraph{Detection-conditioned adjustment.}
To model degraded guidance under missed detection, the following conservative adjustment is applied when attacker $i$ is not detected at $t_k$:
\begin{multline}
p_{12,i}(t_k) \leftarrow \min\{1,\ p_{12,i}(t_k)+\varepsilon_{\mathrm{fail}}\},
\\
p_{11,i}(t_k) \leftarrow 1-p_{12,i}(t_k),
\label{eq:missed_detection_adjust}
\end{multline}
with fixed $\varepsilon_{\mathrm{fail}}\in(0,1)$, while preserving $S_2$ as absorbing.
Although the graph $G_k$ is constructed only on the currently detected set $V_k$, the Markov risk model is maintained for each previously identified active attacker using a propagated state estimate over windows in which direct detection may be temporarily lost. Hence, when attacker $i$ is not detected at $t_k$, the transition matrix $P_i(t_k)$ remains available from the propagated estimate, and we conservatively increase the zone-1 to zone-2 transition probability via \eqref{eq:missed_detection_adjust}.

\subsubsection*{Finite-Horizon Breach Probability and Risk Mapping}
\label{subsubsec:horizon_breach}

We fix a composition horizon $H\in\mathbb{Z}_{\ge 1}$ and define the $H$-step transition matrix as
\begin{equation}
P_i^{(H)}(t_k) = \prod_{\ell=0}^{H-1} P_i(t_{k+\ell}).
\label{eq:PH_def}
\end{equation}
When the transition is treated as window-homogeneous over the horizon, a computationally convenient approximation is $P_i^{(H)}(t_k) \approx \bigl(P_i(t_k)\bigr)^{H}$.

If $e_0,e_1,e_2$ denotes the canonical basis vectors in $\mathbb{R}^3$, $e_{z_i(k)}$ represents the one-hot vector of the current zone. Then, the finite-horizon breach probability is given~by
\begin{multline}
p_{\mathrm{br},i}(t_k;H)
= \mathbb{P}\!\left(X_i(k+H)=S_2 \,\middle|\, X_i(k)\right)\\
= e_{z_i(k)}^\top\,P_i^{(H)}(t_k)\,e_{2}.
\label{eq:pbr_def}
\end{multline}

\noindent As in \cite{c19}, we use this probability as a compact breach-likelihood descriptor for online ranking under uncertainty. To map $p_{\mathrm{br},i}$ to a bounded, smoothly saturating risk term, we define the following map:
\begin{equation}
R_{\mathrm{mkv},i}(t_k) = 1 - \exp\!\bigl(-\gamma_\phi\,p_{\mathrm{br},i}(t_k;H)\bigr),
\qquad \gamma_\phi>0,
\label{eq:markov_risk_map}
\end{equation}
where $R_{\mathrm{mkv},i}(t_k)\in[0,1)$ and is monotonically increasing with $p_{\mathrm{br},i}(t_k;H)$.

\subsection{Combined Criticality Score and Ranking}
\label{subsec:combined_criticality}

For each $i\in V_k$, the combined criticality score is defined as
\begin{multline}
C_{\mathrm{comb}}(i;t_k)
=
s_i(t_k)\Bigl(
w_{\mathrm{ttb}}\,R_i(t_k)
+ w_{\mathrm{cent}}\,C_{\mathrm{cent}}(i;t_k)\\
+ w_{\mathrm{dist}}\,D_i(t_k)
+ w_{\mathrm{mkv}}\,R_{\mathrm{mkv},i}(t_k)
\Bigr),
\label{eq:combined_criticality}
\end{multline}
where $w_{\mathrm{ttb}},w_{\mathrm{cent}},w_{\mathrm{dist}},w_{\mathrm{mkv}}\ge 0$ are weights and
$s_i(t_k)\in(0,1]$ is a confidence factor reflecting the reliability of the current state estimate for attacker $i$.
The ranking of the attackers at epoch $k$ is obtained by sorting
$\{C_{\mathrm{comb}}(i;t_k)\}_{i\in V_k}$ in descending order.

\subsection{Future Criticality Prediction}
\label{subsec:future_criticality}

To support proactive defender assignment, a predicted criticality is computed over a short horizon.
Let $\Delta_{\mathrm{pred}}=H\Delta$ denote a predefined prediction horizon and the predicted attacker state is then defined as 
\begin{equation}
\hat x_{A,i}(t_k+\Delta_{\mathrm{pred}})
=
F\!\bigl(\hat x_{A,i}(t_k),\,u_{A,i}(t_k)\bigr),
\label{eq:predicted_state}
\end{equation}
where $F(\cdot)$ denotes the nominal one-step propagation map applied for $\Delta_{\mathrm{pred}}$ seconds. Using the predicted states, a predicted graph $\hat G_{k+\mathrm{pred}}$ is constructed by the same graph-construction mode as in Section~\ref{sec:prob_graph_model}, and a predicted score is computed by evaluating \eqref{eq:combined_criticality} at $t_k+\Delta_{\mathrm{pred}}$ as
\begin{equation}
C_{\mathrm{comb}}^{\mathrm{fut}}(i;t_k)
=
C_{\mathrm{comb}}\!\bigl(i;\,t_k+\Delta_{\mathrm{pred}}\bigr).
\label{eq:future_combined_criticality}
\end{equation}

The assignment module combines current and predicted criticality scores using the weights $w_C,w_{C+1}\ge 0$ with $w_C+w_{C+1}=1$ to obtain the following score:
\begin{equation*}
C_{\mathrm{assign}}(i;t_k)
=
w_C\,C_{\mathrm{comb}}(i;t_k) + w_{C+1}\,C_{\mathrm{comb}}^{\mathrm{fut}}(i;t_k).
\label{eq:assign_score}
\end{equation*}
$C_{\mathrm{assign}}(i;t_k)$ prioritizes attackers that are both currently critical and likely to remain critical in the near future.

% ============================================================
\section{Time-Optimal Defender Pursuit and Robust Capture}
\label{sec:def_time_optimal}
% ============================================================

This section specifies the defender-UAS pursuit model used to compute (i) per-pair interception-time surrogates for assignment costs and (ii) a per-window capture-success lower bound under model/estimation mismatch.
The construction mirrors the attacker-side pipeline of Section~\ref{sec:det_graph_model}: we first generate a nominal time-optimal pursuit command and then enforce hard safety constraints through a CBF--QP filter. 

\subsection{Defender Dynamics and Engagement Geometry}
\label{subsec:def_dynamics_geometry}

Let $\mathcal{D}:=\{1,\dots,N_D\}$ denote the index set of the defenders. The defenders are modeled in $\mathbb{R}^3$ using Dubins-like kinematics \cite{c12} as given by
\begin{align*}
x_{D,j}(t) & =
[
p_{D,j}(t) \; \psi_{D,j}(t)]^T
\in\mathbb{R}^3\times\mathbb{S}^1,
\\
u_{D,j}(t) & =
[
\omega_{D,j}(t) \; w_{D,j}(t)
]^T\in[-\bar\omega_D,\bar\omega_D]\times[-\bar w_D,\bar w_D],
% \label{eq:def_state_input}
\end{align*}
with control-affine dynamics
\begin{equation}
\dot x_{D,j}(t) = f_D\!\bigl(x_{D,j}(t)\bigr) + g_D\!\bigl(x_{D,j}(t)\bigr)\,u_{D,j}(t),
\label{eq:def_control_affine}
\end{equation}
where $f_D$ encodes constant-speed planar motion parameterized by $\psi_{D,j}$ and $g_D$ injects $(\omega_{D,j},w_{D,j})$ into $(\dot\psi_{D,j},\dot z_{D,j})$. For attacker $i$ and defender $j$, we define the relative position for the pair $(i,j)$ as
\begin{equation}
x^{ij}_{\mathrm{rel}}(t) = p_{D,j}(t) - p_{A,i}(t)\in\mathbb{R}^3.
\label{eq:rel_state}
\end{equation}
Given a capture radius $r_{\mathrm{cap}}>0$, capture occurs at time~$t$ if $\|x^{ij}_{\mathrm{rel}}(t)\|\le r_{\mathrm{cap}}$.

\subsection{Time-Optimal Pursuit Constraint}
\label{subsec:def_time_optimal_pursuit}

At decision epoch $t_k$, consider that defender $j$ is assigned to attacker $i$. Let $\hat p_{A,i}(\tau\,{\mid}\, t_k)$ denote a short-horizon kinematic forecast of attacker $i$ over $\tau\ge t_k$ (generated from the attacker nominal guidance described earlier).
We associate to the ordered pair $(i,j)$ the following minimum-time pursuit problem:
\begin{subequations}
\label{eq:def_min_time}
\begin{align}
& \min_{u_{D,j}(\cdot),\,T_{ij}\ge 0}\quad T_{ij}
\label{eq:def_min_time_cost}\\
\begin{split}
& \dot x_{D,j}(\tau)= f_D(x_{D,j}(\tau)) + g_D(x_{D,j}(\tau))u_{D,j}(\tau),\\
&\qquad \tau\in[t_k,t_k+T_{ij}],
\label{eq:def_min_time_dyn} \end{split}\\
& u_{D,j}(\tau)\in \mathcal{U}_D := [-\bar\omega_D,\bar\omega_D]\times[-\bar w_D,\bar w_D],
\label{eq:def_min_time_bounds}\\
& \min_{\tau\in[t_k,t_k+T_{ij}]}\ \bigl\|p_{D,j}(\tau)-\hat p_{A,i}(\tau\mid t_k)\bigr\| \le r_{\mathrm{cap}} .
\label{eq:def_min_time_capture}
\end{align}
\end{subequations}
The optimal value obtained $T_{ij}^\star(t_k)$ is used as an engagement-time surrogate in the assignment layer.

%----------------------------------------------------

\subsection{Collision-Avoidance Filter (Defender CBF--QP)}
\label{subsec:def_collision_cbf}

Let $u^{\mathrm{nom}}_{D,j}(t_k)$ denote the nominal control input of the defender obtained by solving \eqref{eq:def_min_time}. To ensure inter-defender collision avoidance, we employ a CBF--QP that minimally perturbs $u^{\mathrm{nom}}_{D,j}$ while enforcing forward invariance of the pairwise separation constraints. 

To enable us to write the CBF condition, we define the following function given any two defenders $j\neq j'$:
\begin{equation}
h^{D}_{jj'}(t) := \|p_{D,j}(t)-p_{D,j'}(t)\|^2 - r_{\mathrm{col},D}^2,
\label{eq:def_h_col}
\end{equation}
where $r_{\mathrm{col},D}>0$ is a prescribed defender separation. A standard zeroing-CBF condition is then written as
\begin{equation}
\dot h^{D}_{jj'}(t) + \kappa_D\,h^{D}_{jj'}(t)\ge 0,
\qquad \kappa_D>0.
\label{eq:def_zcbf}
\end{equation}
The preceding condition becomes an affine inequality in the decision variables once $\dot h^{D}_{jj'}$ is expressed using \eqref{eq:def_control_affine}. 

\noindent At each $t_k$, the perturbed control input is given by
\begin{equation}
u_{D,j}(t_k) = u^{\mathrm{nom}}_{D,j}(t_k) + \Delta u_{D,j}(t_k),
\label{eq:def_filter_update}
\end{equation}
where $\Delta u_D(t_k)=\{\Delta u_{D,j}(t_k)\}_{j\in\mathcal{D}}$ is obtained by solving the following convex QP: 
\begin{subequations}
\label{eq:def_cbf_qp}
\begin{align}
& \min_{\{\Delta u_{D,j}\}} \quad
\sum_{j\in\mathcal{D}} \|\Delta u_{D,j}\|_2^2
\label{eq:def_cbf_qp_obj}\\
% \text{s.t.}\quad
& \dot h^{D}_{jj'}(t_k) + \kappa_D\,h^{D}_{jj'}(t_k)\ge 0,
\quad \forall\, j<j',
\label{eq:def_cbf_qp_col}\\
& u^{\mathrm{nom}}_{D,j}(t_k)+\Delta u_{D,j}(t_k)\in \mathcal{U}_D,
\quad \forall\, j\in\mathcal{D}.
\label{eq:def_cbf_qp_bounds}
\end{align}
\end{subequations}

\subsection{Capture Tube, Admissibility, and Per-Window Success Floor}
\label{subsec:tube_success_floor}

Given a defender--attacker pair $(j,i)$, let $\hat{\mathbf{x}}^{ji}_{\mathrm{rel}}(\tau\mid t_k)\in\mathbb{R}^3$ denote the nominal relative trajectory generated by the pursuit law over the time window $[t_k,t_{k+1})$ (prior to safety filtering), where $t_{k+1}=t_k+\Delta$, and let $\mathbf{x}^{ji}_{\mathrm{rel}}(\tau)$ denote the realized relative trajectory.

% \paragraph{Nominal capture and eroded capture sets.}
The nominal capture set $\mathcal{C}$ and its erosion by tube radius $r_t\in(0,r_{\mathrm{cap}})$ are defined as
\begin{align*}
\mathcal{C} &= \{\mathbf{x}\in\mathbb{R}^3:\ \|\mathbf{x}\|\le r_{\mathrm{cap}}\},
\\
\mathcal{C}^{\ominus r_t} &= \{\mathbf{x}\in\mathbb{R}^3:\ \|\mathbf{x}\|\le r_{\mathrm{cap}}-r_t\}.
\label{eq:capture_set_eroded}
\end{align*}

Here, $\mathbf{x}=x_{ij}^{\mathrm{rel}}(t)=p_{D,j}(t)-p_{A,i}(t)$ denotes the relative position between defender $j$ and attacker $i$ along their nominal trajectories generated by the time-optimal pursuit law. The capture set $\mathcal{C}$ is thus defined with respect to this nominal relative trajectory, representing all relative states where successful interception occurs ($\Vert x_{ij}^{\mathrm{rel}}\Vert\le r_{\mathrm{cap}}$). The eroded set $\mathcal{C}^{\ominus r_t}$ provides a robust inner region that guarantees capture even under bounded trajectory execution error of radius $r_t$.

% \paragraph{Nominal admissibility (window-feasible capture).}
\noindent We call the pair $(j,i)$ \emph{nominally admissible} at epoch $t_k$ if the nominal relative trajectory reaches the eroded capture set within the window, i.e.,
\begin{equation}
\exists\,\tau\in[t_k,t_{k+1})\ \text{s.t.}\ \hat{\mathbf{x}}^{ji}_{\mathrm{rel}}(\tau\mid t_k)\in \mathcal{C}^{\ominus r_t},
\label{eq:admissibility}
\end{equation}
or equivalently $\min_{\tau\in[t_k,t_{k+1})}\|\hat{\mathbf{x}}^{ji}_{\mathrm{rel}}(\tau\mid t_k)\|\le r_{\mathrm{cap}}-r_t$.

% \paragraph{Tube event and conditional tube probability.}
The window tube event at time $t_k$, denoted by $E_{\mathrm{tube}}(t_k)$, is said to occur if the following holds: 
\begin{equation}
\|\mathbf{x}^{ji}_{\mathrm{rel}}(\tau)-\hat{\mathbf{x}}^{ji}_{\mathrm{rel}}(\tau\mid t_k)\|\le r_t,\ \forall\tau\in[t_k,t_{k+1}).
\label{eq:tube_event}
\end{equation}
If $\mathcal{F}_k$ denotes the information $\sigma$-algebra at epoch $t_k$, then the conditional tube probability is defined as 
\begin{equation}
p_{\mathrm{min}}(t_k)=\mathbb{P}\!\left(E_{\mathrm{tube}}(t_k)\mid \mathcal{F}_k\right).
\label{eq:ptube}
\end{equation}

\begin{lemma}[Robust capture under tube event]
\label{lem:robust_capture}
If $(j,i)$ is nominally admissible at $t_k$ in the sense of \eqref{eq:admissibility} and the tube event $E_{\mathrm{tube}}(t_k)$ occurs, then capture occurs within the window: there exists $\tau\in[t_k,t_{k+1})$ such that $\|\mathbf{x}^{ji}_{\mathrm{rel}}(\tau)\|\le r_{\mathrm{cap}}$.
\end{lemma}
\begin{proof}
By nominal admissibility, there exists $\tau^\star\in[t_k,t_{k+1})$ with $\|\hat{\mathbf{x}}^{ji}_{\mathrm{rel}}(\tau^\star\mid t_k)\|\le r_{\mathrm{cap}}-r_t$.
Since $E_{\mathrm{tube}}(t_k)$ occurs, we have $\|\mathbf{x}^{ji}_{\mathrm{rel}}(\tau^\star)-\hat{\mathbf{x}}^{ji}_{\mathrm{rel}}(\tau^\star\mid t_k)\|\le r_t$. Thus, by the triangle inequality,
$\|\mathbf{x}^{ji}_{\mathrm{rel}}(\tau^\star)\|\le (r_{\mathrm{cap}}-r_t)+r_t=r_{\mathrm{cap}}$.
\end{proof}

For an executed engagement $(j,i)$ in window $k$, we define the success indicator as
\begin{equation}
Y_{k,ji}=\mathbf{1}\!\left\{\min_{\tau\in[t_k,t_{k+1})}\ \|\mathbf{x}^{ji}_{\mathrm{rel}}(\tau)\|\le r_{\mathrm{cap}}\right\}.
\label{eq:success_indicator}
\end{equation}
By Lemma~\ref{lem:robust_capture}, for every \emph{nominally admissible} executed pair, we have 
\begin{equation}
\mathbb{P}\!\left(Y_{k,ji}=1\mid \mathcal{F}_k\right)
\;\ge\; p_{min}(t_k),
\label{eq:success_floor}
\end{equation}
where $p_{min}(t_k) = \mathbb{P}\!\left(E_{\mathrm{tube}}(t_k)\mid \mathcal{F}_k\right)$.

% ============================================================
\section{Criticality-Driven Defender-to-Attacker Assignment (CRIDAA) Strategy}
\label{sec:cridaa}
% ============================================================

This section maps the attacker criticality scores from Section~\ref{sec:criticality_framework} and the defender interception-time surrogates from Section~\ref{sec:def_time_optimal} into a windowed defender-to-attacker assignment policy.  
At each decision epoch $t_k$, CRIDAA computes an assignment on an admissible bipartite domain and executes safety-filtered pursuit controls over $[t_k,t_k\,{+}\,\Delta)$.  
Algorithm~\ref{alg:cridaa} summarizes the complete closed-loop pipeline used in the experiments, including the graph update, criticality evaluation, short-horizon prediction, assignment, execution, and logging.  

\subsection{Admissible Bipartite Domain}
\label{subsec:admissible_domain}

Recall that $\mathcal{D}$ and $V_k$ denote the defender index set and the detected active attacker set at $t_k$, respectively. We define the following binary assignment variables:
\begin{equation}
\delta_{ji}(t_k)\in\{0,1\},\qquad j\in\mathcal{D},\ i\in V_k,
\label{eq:delta_def}
\end{equation}
where $\delta_{ji}(t_k)=1$ indicates that defender $j$ is assigned to attacker $i$ during window $k$. Candidate assignments are restricted to the admissible bipartite domain $\mathcal{B}_{\mathcal{DA}}(t_k)\subseteq \mathcal{D}\times V_k$ defined by window feasibility and the robust tube margin as described in Section~\ref{sec:def_time_optimal}, as given by  
\begin{multline}
\mathcal{B}_{\mathcal{DA}}(t_k)
=
\Bigl\{(j,i)\in \mathcal{D}\times V_k:\ T_{ji}(t_k)\le \Delta,\\\ \min_{\tau\in[t_k,t_k+\Delta]}\|\hat x_{\mathrm{rel}}^{ij}(\tau\mid t_k)\|\le r_{\mathrm{cap}}-r_t \Bigr\}.
\label{eq:BDA_def}
\end{multline}

\subsection{Current-Time MIQP}
\label{subsec:miqp_current}

Let $C_i(t_k)$ be the combined current criticality of attacker $i$ as described in Section~\ref{sec:criticality_framework}.  
 Let $T_{ji}(t_k)$ denote the interception-time surrogate for defender $j$ to capture attacker $i$ (Section~\ref{sec:def_time_optimal}).  
Given a large constant $M\gg 0$, we define the following interception cost:
\begin{equation}
C^I_{ji}(t_k)=
\begin{cases}
T_{ji}(t_k), & T_{ji}(t_k) < T_i(t_k),\\
M, & \text{otherwise},
\end{cases}
\label{eq:CI_def}
\end{equation}
where $T_i(t_k)$ denotes the attacker time-to-breach proxy employed by the criticality layer.  

To encode collision-coupled assignment penalties, we introduce a nonnegative quadratic cost $COL_{jj'}^{ii'}(t_k)\ge 0$, representing the penalty incurred when defenders $j\neq j'$ are assigned to attackers $i\neq i'$ and a collision is predicted under nominal pursuit.  
One admissible choice (consistent with the original exposition) is
\begin{equation}
COL_{jj'}^{ii'}(t_k)=
\begin{cases}
\dfrac{1}{T^{C}_{jj'}(t_k)}, & \text{if collision predicted},\\
\hspace{5mm} 0, & \text{otherwise},
\end{cases}
\label{eq:COL_def}
\end{equation}
where $T^{C}_{jj'}(t_k)$ is the predicted time-to-collision. The current-time assignment is obtained from the MIQP
\begin{subequations}
\label{eq:cridaa_miqp}
\begin{align}
\begin{split}\min_{\{\delta_{ji}(t_k)\}}
\quad
& \sum_{j\in\mathcal{D}}\ \sum_{i\in V_k}
\delta_{ji}(t_k)\Bigl(
w_T\,C^I_{ji}(t_k) - w_C\,C_i(t_k)
\Bigr)\\
& + \!\!\!\!\sum_{\substack{j,j'\in\mathcal{D}\\ j<j'}}\ \ \sum_{\substack{i,i'\in V_k\\ i\neq i'}}
\delta_{ji}(t_k)\,\delta_{j'i'}(t_k)\,COL_{jj'}^{ii'}(t_k)
\label{eq:cridaa_obj}\end{split}\\
\text{s.t.}\quad
& \sum_{i\in V_k}\delta_{ji}(t_k)\le 1,\ \forall\, j\in\mathcal{D},
\label{eq:cridaa_con_def}\\
& \sum_{j\in\mathcal{D}}\delta_{ji}(t_k)\le 1,\ \forall\, i\in V_k,
\label{eq:cridaa_con_att}\\
& \delta_{ji}(t_k)=0,\ \forall\,(j,i)\notin\mathcal{B}_{\mathcal{DA}}(t_k),\\
&\delta_{ji}(t_k)\in\{0,1\},\ \forall\,(j,i)\in\mathcal{B}_{\mathcal{DA}}(t_k).
\label{eq:cridaa_con_domain}
\end{align}
\end{subequations}

\subsection{Predictive Extension and Switching Regulation}
\label{subsec:predictive_switching}

Let $C_i^{\mathrm{fut}}(t_k)$ be the predicted future criticality score over the horizon $\Delta_{\mathrm{pred}}$ as explained in Section~\ref{sec:criticality_framework}.  
We define the following assignment-relevant criticality: 
\begin{align*}
& C_i^{\mathrm{asgn}}(t_k) =
w_{\mathrm{cur}}\,C_i(t_k) + w_{\mathrm{fut}}\,C_i^{\mathrm{fut}}(t_k), \\
& w_{\mathrm{cur}}+w_{\mathrm{fut}}=1,
% \label{eq:assign_crit_mix}
\end{align*}
which matches the current–future mixing used by the experimental solver. To regulate reassignment transients, we define the following switching penalty:
\begin{equation}
\Phi_{ji}(t_k)= w_{\mathrm{sw}}\,S_{ji}(t_k)\,\mathbf{1}\{a_j(t_k)\neq i\},
\label{eq:switch_penalty}
\end{equation}
where $a_j(t_k)\in V_k\cup\{\varnothing\}$ denotes the previous assignment of defender $j$ and $S_{ji}(t_k)\ge 0$ is a switching-cost surrogate.

The full MIQP in \eqref{eq:cridaa_miqp} is computationally demanding because of the quadratic collision-coupling terms and binary decision variables. Accordingly, Algorithm~\ref{alg:cridaa} implements a centralized assignment stage with infeasible-pair penalties and detection weighting, followed by a switching heuristic to prevent unnecessary reassignment oscillations. Collision avoidance is enforced at execution time by the defender CBF--QP filter in Section~\ref{sec:def_time_optimal}. This separation of planning and safety preserves tractability while retaining the collision-avoidance guarantees needed in the closed-loop experiments.

We record the following variables by solving the MIQP: (i) the number of executed defender–attacker pairs $m_k$, (ii) the window capacity $C_k=\min\{N_D,|V_k|\}$, and (iii) the engagement ratio $\eta_k=m_k/C_k$.

\begin{figure}[t]
\centering
\begin{minipage}{0.96\linewidth}
\begin{algorithm}[H]
\caption{Pseudocode of the CRIDAA algorithm}
\label{alg:cridaa}
\small
\textbf{Input:} Attacker states $\{x_{A,i}(t_0)\}$, defender states $\{x_{D,j}(t_0)\}$, window length $\Delta$, prediction horizon $\Delta_{\mathrm{pred}}$, weights.  \\
\textbf{Output:} Defender assignments $\{\delta_{ji}(t_k)\}$, defender trajectories, engagement/success logs.
\begin{algorithmic}[1]
\STATE Initialize active attacker set $\mathcal{A}(t_0)\leftarrow\{1,\dots,N_A\}$ and defender set $\mathcal{D}\leftarrow\{1,\dots,N_D\}$.  
\STATE Initialize previous assignments $a_j\leftarrow\varnothing$ and cooldown timers $\mathrm{cd}_j\leftarrow 0$ for all $j\in\mathcal{D}$.  
\FOR{$k=0,1,2,\dots$}
    \STATE Set $t_k\leftarrow k\Delta$ and compute detected set $V_k=\widehat{\mathcal{A}}_k\subseteq\mathcal{A}(t_k)$.  
    \STATE Construct $G_k$ (Section~\ref{sec:prob_graph_model}) and compute centrality/risk features.
    \STATE Compute current criticality $\{C_i(t_k)\}_{i\in V_k}$ and predicted criticality $\{C_i^{\mathrm{fut}}(t_k)\}_{i\in V_k}$.  
    \STATE Compute interception times $\{T_{ji}(t_k)\}$ and admissible domain $\mathcal{B}_{\mathcal{DA}}(t_k)$.  
    \STATE Solve assignment (MIQP \eqref{eq:cridaa_miqp} to obtain $\{a_j(t_k)\}_{j\in\mathcal{D}}$.  
    \STATE If reassignment switching is active, update $\{a_j(t_k)\}$ via thresholded switching with cooldown.  
    \STATE Execute one-step defender controls toward $a_j(t_k)$, using the safety filter of Section~\ref{sec:def_time_optimal}.  
    \STATE Propagate attacker dynamics and remove intercepted/breached attackers from $\mathcal{A}(t_k)$.  
    \STATE Update logs $\{m_k,C_k,\eta_k,p_{min}\}$ and switching events.  
\ENDFOR
\STATE \textbf{return} Assignment/trajectory histories and engagement outcomes.
\end{algorithmic}
\end{algorithm}
\end{minipage}
\end{figure}

\section{Theoretical Analysis}
\label{sec:THan}

This section provides window-level performance guarantees for the closed-loop defense architecture in Algorithm~\ref{alg:cridaa}.  
Specifically, it (i) lower-bounds the expected per-window depletion of the live attacker count under an execution fraction $\eta_k$ and a per-engagement success floor $p^\star(t_k)$, (ii) establishes stochastic absorption (eventual neutralization) and corresponding bounds on the expected absorption time under uniform drift conditions, and (iii) shows that, after first detection, the within-window capture condition $\mathrm{IF}_k$ is triggered in finite expected time under a uniform tube-hold probability floor.  
These results are stated at the decision-epoch level with filtration $\mathcal{F}_k$ and are consistent with the logged quantities $(m_k,C_k,\eta_k)$ and the first-detection epoch $\tau_1$ used in the simulator.

% ============================================================
\subsection{Neutralization Bound and Depletion Drift}
\label{subsec:per-window-drift}

\paragraph*{Windowed setup and logged quantities}
For a given window length $\Delta>0$, we define the decision epochs $t_k=k\Delta$, $k\in\mathbb{Z}_{\ge 0}$.
Let $(\mathcal{F}_{t})_{t\ge 0}$ denote the filtration generated by the closed-loop process and write $\mathcal{F}_k:=\mathcal{F}_{t_k}$.
Let $A_k=|\mathcal{A}(t_k)|\in\mathbb{Z}_{\ge 0}$ be the number of active attackers at the beginning of window $k$, and let $|\mathcal{D}|\in\mathbb{Z}_{\ge 1}$ be the number of defenders.
We define the algebraic parallel capacity as
\begin{equation}
  C_k := \min\{|\mathcal{D}|,A_k\}.
  \label{eq:Ck_def_analysis}
\end{equation}
In the experimental pipeline, $C_k$ is logged as the window capacity (\texttt{Chistory}). 

\paragraph*{Executed engagements and neutralizations}

Let $J_k\subseteq \mathcal{D}\times \mathcal{A}(t_k)$ denote the set of \emph{executed} defender--attacker engagements on $[t_k,t_{k+1})$ after within-window retasking as in Algorithm~\ref{alg:cridaa}, and let $m_k=|J_k|$ be the number of executed engagements (\texttt{mhistory}). For each executed engagement $(j,i)\in J_k$, we define the following Bernoulli indicator: 
\begin{equation*}
  Y_{k,ji} = \mathbf{1}\{\text{attacker $i$ is neutralized during } [t_k,t_{k+1})\}.
  \label{eq:Ykji_def}
\end{equation*}
The number of neutralizations completed in window $k$ is given by
\begin{equation}
  K_k := \sum_{(j,i)\in J_k} Y_{k,ji}.
  \label{eq:Kk_def}
\end{equation}
Then, $0\le K_k \le m_k$ and the attacker count satisfies
\begin{equation}
  A_{k+1} = A_k - K_k.
  \label{eq:A_update}
\end{equation}

\paragraph*{Per-engagement success floor}
Assuming that there exists a time-varying conservative success lower bound $p^\star(t_k)\in(0,1]$ such that
\begin{equation}
  \mathbb{E}[Y_{k,ji}\mid \mathcal{F}_k] \ge p^\star(t_k),
  \qquad \forall (j,i)\in J_k.
  \label{eq:Y_floor}
\end{equation}
By linearity of conditional expectation,
\begin{equation}
  \mathbb{E}[K_k\mid \mathcal{F}_k]
  = \sum_{(j,i)\in J_k} \mathbb{E}[Y_{k,ji}\mid \mathcal{F}_k]
  \ge p^\star(t_k)\,m_k.
  \label{eq:EK_lower}
\end{equation}

\paragraph*{Planned vs.\ executed pairs and switching losses}
Let $\widetilde{J}_k$ be the set of \emph{planned} feasible pairs given by the assignment stage at epoch $t_k$ (MIQP in the ideal formulation, or linear assignment with post-processing in the experimental pipeline; cf.\ Algorithm~\ref{alg:cridaa}), with $\widetilde{m}_k=|\widetilde{J}_k|$.  
If we denote the number of planned engagements canceled due to within-window switching/retasking by $\delta_k^{\mathrm{sw}}$, we have
\begin{equation}
  m_k = \widetilde{m}_k - \delta_k^{\mathrm{sw}},
  \qquad
  0\le m_k \le \widetilde{m}_k \le C_k.
  \label{eq:m_decomp}
\end{equation}
In the experimental pipeline, switching events are explicitly logged (\texttt{switchinghistory}) and regulated by a threshold and cool down mechanism.  

\paragraph*{Execution fraction}
We define the per-epoch execution fraction as
\begin{equation}
  \eta_k =
  \begin{cases}
    \dfrac{m_k}{C_k}, & C_k>0,\\[6pt]
    \hspace{2mm} 0, & C_k=0,
  \end{cases}
  \qquad \Rightarrow\qquad \eta_k\in[0,1].
  \label{eq:eta_def_analysis}
\end{equation}
This quantity measures what fraction of the ideal parallel capacity $C_k$ is realized in window $k$ after accounting for feasibility and switching, and it is logged as \texttt{etahistory}.  

\begin{lemma}[Positivity of $\eta_k$ under bounded cancellation]
\label{thm:eta_floor}
Let $\underline{\eta}\in(0,1]$ and $\rho\in[0,1)$ be constants. If $C_k>0$, then 
\begin{equation}
  \widetilde{m}_k \ge \underline{\eta}\,C_k
  \qquad\text{and}\qquad
  \delta_k^{\mathrm{sw}} \le \rho\,C_k.
  \label{eq:eta_assumptions}
\end{equation}
Then, for all $k$ with $C_k>0$,
\begin{equation}
  m_k \ge (\underline{\eta}-\rho)\,C_k
  \quad\Rightarrow\quad
  \eta_k \ge \eta = (\underline{\eta}-\rho) > 0,
  \label{eq:eta_floor_conclusion}
\end{equation}
provided $\underline{\eta}>\rho$.
\end{lemma}

\begin{proof}
From \eqref{eq:m_decomp} and \eqref{eq:eta_assumptions},
$
m_k = \widetilde{m}_k-\delta_k^{\mathrm{sw}}
\ge \underline{\eta}C_k - \rho C_k
= (\underline{\eta}-\rho)C_k,
$
and dividing by $C_k$ gives \eqref{eq:eta_floor_conclusion}.
\end{proof}

\begin{lemma}[Execution fraction and per-window drift]
\label{thm:drift_one_step}
With algebraic capacity $C_k$ and execution fraction $\eta_k$, the attacker count satisfies the conditional one-step drift
\begin{equation}
  \mathbb{E}[A_{k+1}\mid \mathcal{F}_k]
  \le
  A_k - p^\star(t_k)\,\eta_k\,\min\{|\mathcal{D}|,A_k\}.
  \label{eq:drift_general}
\end{equation}
\end{lemma}
\begin{proof}
Using \eqref{eq:A_update} and $K_k\ge 0$,
\begin{multline}
\mathbb{E}[A_{k+1}\mid \mathcal{F}_k]
=
A_k - \mathbb{E}[K_k\mid \mathcal{F}_k]
\le
A_k - p^\star(t_k)\,m_k\\
=
A_k - p^\star(t_k)\,\eta_k\,C_k.
\end{multline}
Since $C_k=\min\{|\mathcal{D}|,A_k\}$, \eqref{eq:drift_general} follows.
\end{proof}

% \paragraph*{Effectiveness parameter}
\noindent We define the per-window effectiveness as
\begin{equation}
  \Theta_k = p^\star(t_k)\,\eta_k.
  \label{eq:Theta_k_def}
\end{equation}

\begin{corollary}[Mixed linear$\rightarrow$exponential regimes]
\label{cor:mixed_regimes}
For each decision window $k$:
\begin{enumerate}
\item[\textnormal{(i)}] If $A_k\ge |\mathcal{D}|$ (defenders saturated), then
\begin{equation}
  \mathbb{E}[A_{k+1}-A_k\mid \mathcal{F}_k] \le -\Theta_k\,|\mathcal{D}|.
\end{equation}
\item[\textnormal{(ii)}] If $A_k<|\mathcal{D}|$ (sub-capacity), then
\begin{equation}
  \mathbb{E}[A_{k+1}\mid \mathcal{F}_k] \le (1-\Theta_k)\,A_k.
\end{equation}
\end{enumerate}
\end{corollary}
\begin{proof}
The proof follows immediately from \eqref{eq:drift_general} and the identity $C_k=\min\{|\mathcal{D}|,A_k\}$:
if $A_k\ge |\mathcal{D}|$, then $C_k=|\mathcal{D}|$; if $A_k<|\mathcal{D}|$, then $C_k=A_k$.
\end{proof}
% ============================================================
\subsection{Absorption and Expected Time}
\label{subsec:absorption-expected-time}

Let $|\mathcal{D}|\in\mathbb{Z}_{\ge 1}$ and $A_k\in\mathbb{Z}_{\ge 0}$ denote the number of defenders and active attackers at decision epoch $t_k$, respectively. Let $(\mathcal{F}_{t})_{t\ge 0}$ be the filtration generated by the closed-loop system up to continuous time $t$, and we define the window-level filtration as $\mathcal{F}_k=\mathcal{F}_{t_k}$. Let $C_k = \min\{|\mathcal{D}|,A_k\}$, and $\eta_k\in(0,1]$ be the \emph{execution efficiency} (fraction of feasible concurrent engagements actually executed). $p^\star(t_k)\in(0,1]$ denotes a conservative \emph{tube-hold / success} lower bound for any executed engagement in window $k$. Under the following design enforcing uniform floors: $\eta_k \ge \eta \in (0,1]$,
  $p^\star(t_k) \ge p_{\min}\in(0,1], \ \forall k\ge 0$, the per-window effectiveness admits the uniform bound
\begin{equation}
\label{eq: theta}
  \Theta_k = p^\star(t_k)\eta_k \ge \Theta, \ \text{where}\ \Theta = p_{\min}\eta > 0.
\end{equation}
The expected depletion drift can then be written~as
\begin{equation}
\label{eq:drift-min-form}
  \!\!\!\!\!\mathbb{E}[A_{k+1}\,{\mid}\,\mathcal{F}_k]
  \,{\le}\,  A_k \,{-}\, \Theta_k\,C_k
  \,{\le}\,
  A_k \,{-}\, \Theta\min\{|\mathcal{D}|,A_k\}.
\end{equation}

The following theorem is a specialization of standard drift-based hitting-time arguments: general quantitative bounds for first-hitting and occupation times under uniform negative drift (plus increment-tail conditions) are provided in ~\cite{c20}, while Foster--Lyapunov drift criteria for Markov chains (including state-dependent forms) are developed in ~\cite{c24}.
Related moment/regularity conditions for negative-drift processes with bounded $L^p$ increments are discussed in ~\cite{c25}.
In our setting, $(A_k)_{k\ge 0}$ models the number of surviving attackers and \eqref{eq:drift-min-form} encodes that whenever at least one attacker remains, the defender policy removes attackers at a strictly positive rate in expectation, with at most $|\mathcal{D}|$ removals per step.

\begin{theorem}[Stochastic absorption under uniform drift]
\label{thm:absorption_uniform_drift}
Let $(A_k)_{k\ge 0}$ be a discrete-time process with values in $\mathbb{Z}_+$ adapted to a filtration $(\mathcal{F}_k)_{k\ge 0}$, and the absorption time
\begin{equation}
  T_0 \;=\; \inf\{k\ge 0:\ A_k = 0\}.
\end{equation}
If the following holds: 
\begin{enumerate}
  \item[\textnormal{(i)}] \textit{Absorption at $0$:} If $A_k=0$, then $A_{k+1}=0$ a.s.
  \item[\textnormal{(ii)}] \textit{Uniform negative drift outside $\{0\}$:} There exists $\Theta>0$ such that for all $k\ge 0$,
  \begin{equation}
    \label{eq:uniform-drift}
    \!\!\!\!\!\!\!\mathbb{E}\!\bigl[A_{k+1}{-}A_k {\mid }\mathcal{F}_k\bigr]
    {\le} -\Theta
    \  \text{on the event} \{A_k>0\}.
  \end{equation}
  \item[\textnormal{(iii)}] \textit{Bounded increments:} There exists $C<\infty$ such that for all $k\ge 0$, $|A_{k+1}-A_k| \;\le\; C$ a.s. 
\end{enumerate}
Then, $\{0\}$ is absorbing and the process exhibits stochastic absorption,i.e., $\mathbb{P}(T_0<\infty)=1$ and
  $\mathbb{E}[T_0]\le \frac{\mathbb{E}[A_0]}{\Theta}<\infty.$
\end{theorem}

\begin{remark}[Interpretation under stochastic variability]
The conclusion $\mathbb{P}(T_0<\infty)=1$ asserts \textbf{eventual} absorption with probability one.
This does not preclude finite-horizon stochastic variability: for any fixed horizon $K<\infty$, it may still hold that $\mathbb{P}(A_K=0)<1$.
The uniform drift $\Theta>0$ aggregates compounded inefficiencies (e.g., imperfect execution $\eta_k\le 1$ and uncertainty $p^\star(t_k)\le 1$) while still guaranteeing eventual neutralization almost surely.
\end{remark}

\begin{proof}
The proof is a standard drift/optional-stopping argument for a nonnegative supermartingale (see~\cite{c20} for related drift-based hitting-time bounds). We define the following stopped process: $\widetilde{A}_k = A_{k\wedge T_0},\ k\ge 0$.
Then, $(\widetilde{A}_k)_{k\ge 0}$ is nonnegative, integer-valued, and adapted to $(\mathcal{F}_k)$.
By (i), once the process hits $0$ it remains there, hence $\widetilde{A}_k=0, \forall k\ge T_0$. On $\{k<T_0\}$, we have $A_k>0$, so \eqref{eq:uniform-drift} yields
\begin{equation}
  \mathbb{E}\!\bigl[\widetilde{A}_{k+1}-\widetilde{A}_k\mid \mathcal{F}_k\bigr]
  =
  \mathbb{E}\!\bigl[A_{k+1}-A_k\mid \mathcal{F}_k\bigr]
  \le -\Theta.
\end{equation}
On $\{k\ge T_0\}$, we have $\widetilde{A}_{k+1}=\widetilde{A}_k=0$ and therefore
$\mathbb{E}[\widetilde{A}_{k+1}-\widetilde{A}_k\mid \mathcal{F}_k]=0$.
Combining both cases gives
\begin{equation}
  \label{eq:drift-stopped}
  \mathbb{E}\!\bigl[\widetilde{A}_{k+1}-\widetilde{A}_k\mid \mathcal{F}_k\bigr]
  \le -\Theta\,\mathbf{1}_{\{k<T_0\}},
  \qquad \forall k\ge 0.
\end{equation}

We now define $ M_k = \widetilde{A}_k + \Theta (k\wedge T_0), \ k\ge 0.$ Using \eqref{eq:drift-stopped} and the identity
\(
(k+1)\wedge T_0 = (k\wedge T_0) + \mathbf{1}_{\{k<T_0\}},
\)
we have
\begin{align}
  & \mathbb{E}[M_{k+1}\mid\mathcal{F}_k]
  =
  \mathbb{E}[\widetilde{A}_{k+1}\mid\mathcal{F}_k] + \Theta\,\mathbb{E}[(k+1)\wedge T_0\mid\mathcal{F}_k] \nonumber\\
  & \le
  \widetilde{A}_k - \Theta\mathbf{1}_{\{k<T_0\}} + \Theta\bigl((k\wedge T_0){+}\mathbf{1}_{\{k<T_0\}}\bigr)
  = M_k. \!\!
\end{align}
Hence, $(M_k)_{k\ge 0}$ is a supermartingale.

From the assumption of bounded increments and nonnegativity, for each fixed $k$, the random variable $M_k$ is integrable, and therefore $ \mathbb{E}[M_k] \le \mathbb{E}[M_0] = \mathbb{E}[A_0], \ \forall k\ge 0.$
Since $M_k \ge \Theta (k\wedge T_0)$, we have $\Theta\,\mathbb{E}[k\wedge T_0] \le \mathbb{E}[M_k] \le \mathbb{E}[A_0],\ \forall k\ge 0,$
hence $\mathbb{E}[k\wedge T_0] \le \frac{\mathbb{E}[A_0]}{\Theta},
  \ \forall k\ge 0.$
Letting $k\to\infty$ and using monotone convergence (since $k\wedge T_0 \uparrow T_0$) yields
\begin{equation}
  \mathbb{E}[T_0] = \lim_{k\to\infty}\mathbb{E}[k\wedge T_0]
  \le \frac{\mathbb{E}[A_0]}{\Theta} < \infty.
\end{equation}
In particular, $\mathbb{P}(T_0=\infty)=0$ (otherwise $\mathbb{E}[T_0]=\infty$), so $\mathbb{P}(T_0<\infty)=1$.
Finally, (i) ensures $\{0\}$ is absorbing.
\end{proof}

\begin{remark}[Verifying the drift from the policy model]
If the closed-loop policy ensures \eqref{eq:drift-min-form} for some $\Theta>0$, then for all $k$ on $\{A_k>0\}$, we have $\min\{|\mathcal{D}|,A_k\}\ge 1$ and thus
\(
\mathbb{E}[A_{k+1}-A_k\mid \mathcal{F}_k]\le -\Theta,
\)
which is exactly \eqref{eq:uniform-drift}.
The use of a linear Lyapunov function $V(i)=i$ and the small/absorbing set $\{0\}$ is consistent with Foster--Lyapunov methodology for Markov chains; see ~\cite{c24}.
More general moment conditions for negative-drift processes with bounded increments appear in ~\cite{c25}.
\end{remark}

%------------------------------------------------
\begin{corollary}[{Explicit mixed linear--geometric bound on $\mathbb{E}[T_0]$}]
\label{cor:ET0_mixed_with_tau1}
Let $n_d(k)\in\mathbb{Z}_{\ge 0}$ denote the number of detected attackers at epoch $t_k$. 
We define the first-detection time and the neutralization (absorption) time as $\tau_1 \;=\; \inf\{k\ge 0:\; n_d(k)\ge 1\}$ and $T_0 \;=\; \inf\{k\ge 0:\; A_k = 0\}$, respectively.
Let $\tau_1$ denote an $(\mathcal{F}_k)$-stopping time and assume $\mathbb{P}\{\tau_1<\infty\}=1$ and a uniform effectiveness parameter $\Theta = p_{\min}\eta$ with $0<\Theta\le 1$.
Assuming that the mixed drift inequalities of Corollary~\ref{cor:mixed_regimes} hold \emph{for all epochs $k\ge \tau_1$} i.e., once detection is available 
(these are standard additive- and multiplicative-drift conditions for first-hitting time analysis under filtrations and stopping times; see \cite{c26,c21,c22}), then the following conditional bound holds a.s.:
\begin{multline}
\mathbb{E}[T_0 \mid \mathcal{F}_{\tau_1}]
\;\le\;
\tau_1
\;+\;
\frac{(A_{\tau_1}-|\mathcal{D}|)_{+}}{\Theta|\mathcal{D}|}\\
\;+\;
\frac{1+\ln\!\big(\min\{|\mathcal{D}|,\,\max\{1,A_{\tau_1}\}\}\big)}{\Theta}.
\label{eq:ET0_cond_bound_post_detect}
\end{multline}
Consequently,
\begin{multline}
\mathbb{E}[T_0]
\;\le\;
\mathbb{E}[\tau_1]
\;+\;
\mathbb{E}\!\left[\frac{(A_{\tau_1}-|\mathcal{D}|)_{+}}{\Theta|\mathcal{D}|}\right]\\
\;+\;
\mathbb{E}\!\left[\frac{1+\ln\!\big(\min\{|\mathcal{D}|,\,\max\{1,A_{\tau_1}\}\}\big)}{\Theta}\right].
\label{eq:ET0_uncond_bound_post_detect}
\end{multline}

\noindent Moreover, if $A_{\tau_1}\le A_0$ a.s. (e.g., when $A_k$ is nonincreasing), then
\begin{multline}
\mathbb{E}[T_0]
\;\le\;
\mathbb{E}[\tau_1]
\;+\;
\frac{(A_0-|\mathcal{D}|)_{+}}{\Theta|\mathcal{D}|}\\
\;+\;
\frac{1+\ln\!\big(\min\{|\mathcal{D}|,\,\max\{1,A_0\}\}\big)}{\Theta}.
\label{eq:ET0_bound_with_A0_fixed}
\end{multline}

Under immediate detection ($\tau_1=0$ a.s.), \eqref{eq:ET0_bound_with_A0_fixed} yields an explicit bound solely in terms of $A_0$,$|\mathcal{D}|$, and $\Theta)$.
\end{corollary}

\begin{proof}
Let $\widetilde{A}_\ell := A_{\tau_1+\ell}$ for $\ell\ge 0$ denote the post-detection process, and its absorption time is defined as $\widetilde{T}_0 \;=\; \inf\{\ell\ge 0:\; \widetilde{A}_\ell = 0\}.$
Then, $T_0 = \tau_1 + \widetilde{T}_0$ and hence
\begin{equation}
\mathbb{E}[T_0\mid \mathcal{F}_{\tau_1}] \;=\; \tau_1 + \mathbb{E}[\widetilde{T}_0\mid \mathcal{F}_{\tau_1}],
\label{eq:ET0_split_cond}
\end{equation}
where conditioning on $\mathcal{F}_{\tau_1}$ and shifting the process is standard for drift/hitting-time arguments under filtrations and stopping times (see \cite{c26}). We define the threshold-crossing time as $\widetilde{T}_{\downarrow}
\;=\;
\inf\{\ell\ge 0:\; \widetilde{A}_\ell < |\mathcal{D}|\}$.
\noindent \textit{Phase 1 (Linear depletion while $\widetilde{A}_\ell\ge |\mathcal{D}|$)}: 
On $\{\widetilde{A}_\ell\ge |\mathcal{D}|\}$, the mixed drift inequalities at time $k=\tau_1+\ell$ yields $\mathbb{E}[\widetilde{A}_{\ell+1}\mid \mathcal{F}_{\tau_1+\ell}]
\;\le\;
\widetilde{A}_\ell - \Theta|\mathcal{D}|.$
By an additive-drift (linear-progress) first-hitting-time bound proved via a stopped-supermartingale argument (see \cite{c21,c26}),
\begin{equation}
\mathbb{E}[\widetilde{T}_{\downarrow}\mid \mathcal{F}_{\tau_1}]
\;\le\;
\frac{(\widetilde{A}_0-|\mathcal{D}|)_{+}}{\Theta|\mathcal{D}|}
\;=\;
\frac{(A_{\tau_1}-|\mathcal{D}|)_{+}}{\Theta|\mathcal{D}|}.
\label{eq:ETdown_bound_cond}
\end{equation}

\noindent \textit{Phase 2 (Multiplicative decay while $\widetilde{A}_\ell<|\mathcal{D}|$).}
By definition, $\widetilde{A}_{\widetilde{T}_{\downarrow}}<|\mathcal{D}|$.
On $\{\widetilde{A}_{\widetilde{T}_{\downarrow}}>0\}$, applying the mixed drift inequalities at time $k=\tau_1+\widetilde{T}_{\downarrow}+\ell$ gives
\begin{equation}
\mathbb{E}\!\left[\widetilde{A}_{\widetilde{T}_{\downarrow}+\ell} - \widetilde{A}_{\widetilde{T}_{\downarrow}+\ell+1}
\;\middle|\; \mathcal{F}_{\tau_1+\widetilde{T}_{\downarrow}+\ell}\right]
\;\ge\;
\Theta\,\widetilde{A}_{\widetilde{T}_{\downarrow}+\ell}.
\label{eq:mult_drift_form}
\end{equation}

\noindent A standard multiplicative-drift hitting-time bound yields 
% (see \cite{c22,c26})
\begin{equation}
\mathbb{E}[\widetilde{T}_0-\widetilde{T}_{\downarrow}\mid \mathcal{F}_{\tau_1}]
\;\le\;
\frac{1+\ln\!\big(\max\{1,\widetilde{A}_{\widetilde{T}_{\downarrow}}\}\big)}{\Theta}.
\label{eq:ET2_bound_cond}
\end{equation}
Since $\widetilde{A}_{\widetilde{T}_{\downarrow}} \le \min\{|\mathcal{D}|,A_{\tau_1}\}$ a.s., monotonicity of $\ln(\max\{1,\cdot\})$ implies
\begin{equation}
\mathbb{E}[\widetilde{T}_0-\widetilde{T}_{\downarrow}\mid \mathcal{F}_{\tau_1}]
\;\le\;
\frac{1+\ln\!\big(\min\{|\mathcal{D}|,\,\max\{1,A_{\tau_1}\}\}\big)}{\Theta}.
\label{eq:finref}
\end{equation}
\noindent Combining \eqref{eq:ETdown_bound_cond} with \eqref{eq:finref} gives
\begin{multline}
\mathbb{E}[\widetilde{T}_0\mid \mathcal{F}_{\tau_1}]
\;\le\;
\frac{(A_{\tau_1}-|\mathcal{D}|)_{+}}{\Theta|\mathcal{D}|}\\
\;+\;
\frac{1+\ln\!\big(\min\{|\mathcal{D}|,\,\max\{1,A_{\tau_1}\}\}\big)}{\Theta}.
\end{multline}
Using \eqref{eq:ET0_split_cond} yields \eqref{eq:ET0_cond_bound_post_detect}, and taking expectation yields \eqref{eq:ET0_uncond_bound_post_detect}.
Finally, if $A_{\tau_1}\le A_0$ a.s., then \eqref{eq:ET0_bound_with_A0_fixed} follows from monotonicity of $(\cdot)_+$ and $\ln\!\big(\min\{|\mathcal{D}|,\max\{1,\cdot\}\}\big)$.
\end{proof}

%---------------------------------

% ============================================================
\subsection{Finite-Time Triggering of the Capture Condition}
\label{sec:IF-triggering}
% ============================================================

Consider the windowed assignment--execution loop in Algorithm~\ref{alg:cridaa}. We define the within-window capture event (IF) as follows:
\begin{equation*}
\label{eq:IF-condition}
\mathrm{IF}_k
\,{=}\,
\Big\{\!
\exists(j,i)\,{\in}\, J_k\,{:}
\min_{\tau\in [t_k,t_{k+1})}\!\!
\|\mathbf{x}_{D,j}(\tau)-\mathbf{x}_{A,i}(\tau)\|
\,{\le}\, r_{\mathrm{cap}} \!
\Big\},
\end{equation*}
where $J_k\subseteq \mathcal{D}\times \mathcal{A}(t_k)$ denotes the executed defender--attacker engagement pairs during $[t_k,t_{k+1})$ after assignment and switching at epoch $t_k$.

\begin{definition}[First capture / IF-trigger epoch]
\label{def:kappa1_capture}
The first capture (first IF-trigger) epoch is defined as
\begin{equation}
\kappa_1=\inf\{k\ge 0:\, \mathrm{IF}_k \text{ holds}\},
\end{equation}
and for $k\ge 2$ define $\kappa_k=\inf\{\ell>\kappa_{k-1}:\ \mathrm{IF}_{\ell} \text{ holds}\}$.
\end{definition}

\begin{theorem}[Finite-time triggering of the capture condition after detection]
\label{thm:IF-triggering}
% Let $\{A_k\}_{k\ge 0}$ denote the live attacker count at decision epochs $t_k$.
If the following are satisfied:
\begin{enumerate}
\item \emph{Robust admissibility}: Each $(j,i)\in J_k$ satisfies $(j,i)\in\mathcal{B}_{\mathcal{DA}}(t_k)$, i.e., window-feasibility and robust tube-margin conditions of \eqref{eq:BDA_def} hold. 
% (Section~\ref{subsec:admissible_domain}).
\item \emph{Per-engagement success floor}: For $(j,i)\in J_k$, let
\begin{equation*}
Y_{k,ji}=\mathbf{1}\Big\{\min_{\tau\in [t_k,t_{k+1})}\|\mathbf{x}_{D,j}(\tau)-\mathbf{x}_{A,i}(\tau)\|\le r_{\mathrm{cap}}\Big\}.
% \label{eq:Ykji_capture}
\end{equation*}
Then, for all $k$, $\mathbb{P}\{Y_{k,ji}=1\mid\mathcal{F}_k\}\ge p^\star(t_k)$, where $p^\star(t_k) =\mathbb{P}\{E_{\mathrm{tube}}(t_k)\mid\mathcal{F}_k\}\ge p_{\min}>0$ and

$E_{\mathrm{tube}}(t_k)$ is the tube-hold event of Lemma~\ref{lem:robust_capture}. 
% (a standard ``first-hitting'' argument; cf.\ \cite{c26,c21}).
\item \emph{Robust execution (post-detection)}: For all $k\ge \tau_1$, $m_k \ge \eta C_k$ with $\eta\in(0,1]$ uniform.
\item \emph{Non-degeneracy (post-detection)}: On $\{k\ge \tau_1\}\cap\{A_k\ge 1\}$, $m_k\ge 1$\\
((2)--(4) impose a uniform conditional success floor for at least one executed engagement per post-detection epoch, which leads to geometric tail and finite mean bounds for the first-trigger time).
\end{enumerate}
Then, the following holds:
\begin{enumerate}
\item For all $k\ge \tau_1$, on $\{A_k\ge 1\}$, $\mathbb{P}\{\mathrm{IF}_k\mid\mathcal{F}_k\}\ge p_{\min}$.

\item For all $n\in \mathbb{Z}_{+}$, $\mathbb{P}\{\kappa_1>\tau_1+n\mid \mathcal{F}_{\tau_1}\}
\;\le\; (1-p_{\min})^{n+1}$ a.s. on $\{A_{\tau_1}\ge 1\}$. 

\item Consequently, $\mathbb{E}[\kappa_1-\tau_1\mid \mathcal{F}_{\tau_1}] \;\le\; \frac{1}{p_{\min}}$ a.s. on $\{A_{\tau_1}\ge 1\}$, 
and hence $\mathbb{E}[\kappa_1]\le \mathbb{E}[\tau_1]+\frac{1}{p_{\min}}$ whenever $\mathbb{E}[\tau_1]<\infty$.
\item If $T_0=\inf\{k\ge 0:\ A_k=0\}$ satisfies $\mathbb{P}\{T_0<\infty\}=1$ (by Theorem~\ref{thm:absorption_uniform_drift}), then
\begin{equation*}
\sum_{k\ge 0}\mathbf{1}\{\mathrm{IF}_k\}<\infty\  \text{a.s.},
\;
\text{and } \sum_{k\ge 0}\mathbf{1}\{\mathrm{IF}_k\}\le T_0 \ \text{a.s.}
\end{equation*}
\end{enumerate}
\end{theorem}

\begin{proof}
Fix $k\ge \tau_1$ and consider the event $\{A_k\ge 1\}$.
Then $C_k\ge 1$ and by Lemma~\ref{thm:eta_floor}, we have $m_k\ge \eta C_k\ge \eta>0$, hence $m_k\ge 1$ and $J_k\neq\emptyset$.
Pick any $(j,i)\in J_k$.
From~\eqref{eq:success_floor}, $\mathbb{P}\{Y_{k,ji}=1\mid\mathcal{F}_k\}\ge p_{\min}$.
Since $\mathrm{IF}_k=\bigcup_{(j,i)\in J_k}\{Y_{k,ji}=1\}$, it follows that
\begin{equation*}
\mathbb{P}\{\mathrm{IF}_k\mid\mathcal{F}_k\}
\;\ge\;
\max_{(j,i)\in J_k}\mathbb{P}\{Y_{k,ji}=1\mid\mathcal{F}_k\}
\;\ge\;
p_{\min}.
\end{equation*}

\noindent Define $F_k=\mathrm{IF}_k^c$ and for $k\ge \tau_1$, we have $\mathbb{P}\{F_k\mid\mathcal{F}_k\}\le 1-p_{\min}$ on $\{A_k\ge 1\}$ from the preceding. Moreover, $\{\kappa_1>\tau_1+n\}=\bigcap_{\ell=0}^{n} F_{\tau_1+\ell}$.
Iterating conditional expectations (tower property) yields (2). 
This type of filtration-based iteration to obtain geometric tails from a uniform conditional failure bound is standard in first-hitting-time analyses for adapted processes; see, e.g., \cite{c26}.

\vspace{2mm}
\noindent By the tail-sum formula and (2), we have $\mathbb{E}[\kappa_1-\tau_1\mid \mathcal{F}_{\tau_1}]
=
\sum_{n\ge 0}\mathbb{P}\{\kappa_1>\tau_1+n\mid \mathcal{F}_{\tau_1}\} \le
\sum_{n\ge 0}(1-p_{\min})^{n+1}
=
\frac{1}{p_{\min}}$.
Taking expectations gives $\mathbb{E}[\kappa_1]\le \mathbb{E}[\tau_1]+\frac{1}{p_{\min}}$.

\vspace{2mm}
\noindent On $\{T_0<\infty\}$, we have $\mathbf{1}\{\mathrm{IF}_k\}=0$ for all $k\ge T_0$, hence
$\sum_{k\ge 0}\mathbf{1}\{\mathrm{IF}_k\}\le T_0<\infty$.
This is immediate from the definition of $T_0$ as an absorption time.
% (cf.\ standard hitting-time arguments in \cite{c26}).
\end{proof}

\begin{remark}[Multiplicity refinement]
If $\{Y_{k,ji}\}_{(j,i)\in J_k}$ are conditionally independent given $\mathcal{F}_k$, then for all $k$,
\begin{multline}
\mathbb{P}\{\mathrm{IF}_k\mid\mathcal{F}_k\}
=
1-\prod_{(j,i)\in J_k}\mathbb{P}\{Y_{k,ji}=0\mid\mathcal{F}_k\}
\ge\\
1-(1-p_{\min})^{m_k}
\ge
1-(1-p_{\min})^{\eta C_k}.
\end{multline}
\end{remark}

% ============================================================
\subsection{Computational Complexity Analysis}
\label{subsec:complexity}
% ============================================================

This subsection summarizes the leading-order computational and memory cost of the online controller executed per decision window in Algorithm~\ref{alg:cridaa}, and of the offline Monte Carlo evaluation used for ablation/sensitivity studies. 
We follow the simulator pipeline: overlap-graph construction (voxelized Gaussian overlap), centrality computation, Markov risk evaluation, centralized assignment with optional switching/cooldown, and one-step propagation with interception/breach checks. 

At each epoch $t_k$, computation is causal with respect to the window-level information set $\mathcal{F}_k$; repeated Monte Carlo trials are offline evaluation of the same online policy. Let $g$ be the voxel grid resolution used for overlap (so the voxel count scales as $g^3$). Let $N_s$ denote the number of accepted Monte Carlo samples used to estimate the Markov transition probabilities, $H$ the Markov look-ahead horizon, and $K_{\mathrm{eig}}$ the iteration budget for eigenvector-based centrality computation. In overlap mode, probabilistic graph construction evaluates voxelized Gaussian pdfs and pairwise overlaps over detected attackers, giving leading-order $\mathcal{O}\!\left(n_d(k)^2\,g^3\right)$.
% \begin{equation}
% \mathcal{O}\!\left(n_d(k)^2\,g^3\right).
% \label{eq:complex_graph}
% \end{equation}
Centrality computation is linear for degree, $\mathcal{O}(K_{\mathrm{eig}}|E_k|)$ for eigenvector-type measures, and at most $\mathcal{O}(|V_k||E_k|)$ when betweenness is computed (Brandes).
Markov risk evaluation (three-zone transition estimation plus $H$-step propagation) scales as $\mathcal{O}\!\left(N_k\,(N_s+H)\right)$,
% \begin{equation}
% \mathcal{O}\!\left(N_k\,(N_s+H)\right),
% \label{eq:complex_markov}
% \end{equation}
where $(N_s,H)$ are explicit fidelity--cost~knobs. 

The assignment phase builds a rectangular cost structure over defender--attacker pairs and solves a linear sum assignment problem (LSAP) using the Hungarian method in the default setting, with a standard rectangular worst-case bound $\mathcal{O}\!(\min\{N_D,N_k\}^2\max\{N_D,N_k\})$,
% \[
% \mathcal{O}\!\left(\min\{N_D,N_k\}^2\max\{N_D,N_k\}\right),
% \]
after $\mathcal{O}(N_DN_k)$ cost construction; switching/cooldown tracking is lower order. Dynamics propagation and event checks contribute $\mathcal{O}(N_k+N_D)$, with interception checks at most linear in the number of executed engagements. Collecting the dominant terms, the per-window complexity is $\mathcal{O}\!\Big( n_d(k)^2 g^3 + |V_k||E_k| + K_{\mathrm{eig}}|E_k| + N_k(N_s+H) + \min\{N_D,N_k\}^2\max\{N_D,N_k\} \Big),$
% \begin{multline}
% \mathcal{O}\!\Big(
% n_d(k)^2 g^3
% \;+\;
% |V_k||E_k|
% \;+\;
% K_{\mathrm{eig}}|E_k|
% \;+\;
% N_k(N_s+H)\\
% \;+\;
% \min\{N_D,N_k\}^2\max\{N_D,N_k\}
% \Big),
% \label{eq:complex_total_step}
% \end{multline}
where lower-order additions such as $\mathcal{O}(N_DN_k)$ and $\mathcal{O}(N_k+N_D)$~are~omitted. 

% \paragraph{Per-episode and experiment complexity (offline evaluation).}
For an episode horizon of $T$ decision windows, the episode runtime is the sum of the per window complexity over $k=0,\dots,T-1$. For sensitivity sweeps and Monte Carlo evaluation, total compute scales linearly with the number of parameter cells $N_{\mathrm{grid}}$ and replications per cell $N_{\mathrm{rep}}$: $\mathcal{O}\!
N_{\mathrm{grid}}\,N_{\mathrm{rep}} \sum_{k=0}^{T-1} C^{\mathrm{comp}}_k$, where $C^{\mathrm{comp}}_k$ denotes the per-window cost. If only state trajectories and scalar logs are retained, storage is $\mathcal{O}(T(N_A+N_D))$. Storing full graph snapshots adds $\sum_{k=0}^{T-1}(|V_k|+|E_k|)$ (up to constant factors depending on the chosen graph representation).

\section{Results and Discussion}
\label{sec:results}

\subsection{Simulation Setup}
We define PZ as a vertical cylinder of height $h=20\,\mathrm{u}$ with hard boundary radius $r_{\mathrm{hard}}=10\,\mathrm{u}$ and soft boundary radius $r_{\mathrm{soft}}=15\,\mathrm{u}$ (Fig.~\ref{sim}), where $\mathrm{u}$ denotes units. 
Attackers breach when they reach the hard boundary.
Defenders intercept attackers within a capture radius $r_{\mathrm{cap}}=1.5\,\mathrm{u}$. We simulate $N_A=10$ attacker UAS with speeds in $[0.5,1.0]\,\mathrm{u}$ per time step and $N_D=6$ defenders with speed bound $v_{\mathrm{def}}^{\max}=3.5\,\mathrm{u}$ per time step.
Each UAS uses angular speed $\omega$ and vertical speed $w$ inputs, with collision-avoidance and connectivity/safety constraints enforced through the low-level control computation.
Attackers are initialized randomly outside the SB and move toward the PZ; defenders are initialized around the hard boundary to provide perimeter coverage. We evaluate two attacker-interaction models, namely deterministic proximity graph and probabilistic overlap graph, that induce different information/coordination structure. A communication radius of $r_{\mathrm{comm}}=25\,\mathrm{u}$ is used. All simulations are run on a 12th Gen Intel(R) Core(TM) i5-12450H CPU (2.00 GHz) with 16 GB RAM using Python/NumPy/SciPy/NetworkX/Matplotlib.

\begin{figure}[bp]
    \centering
    \includegraphics[width=\columnwidth]{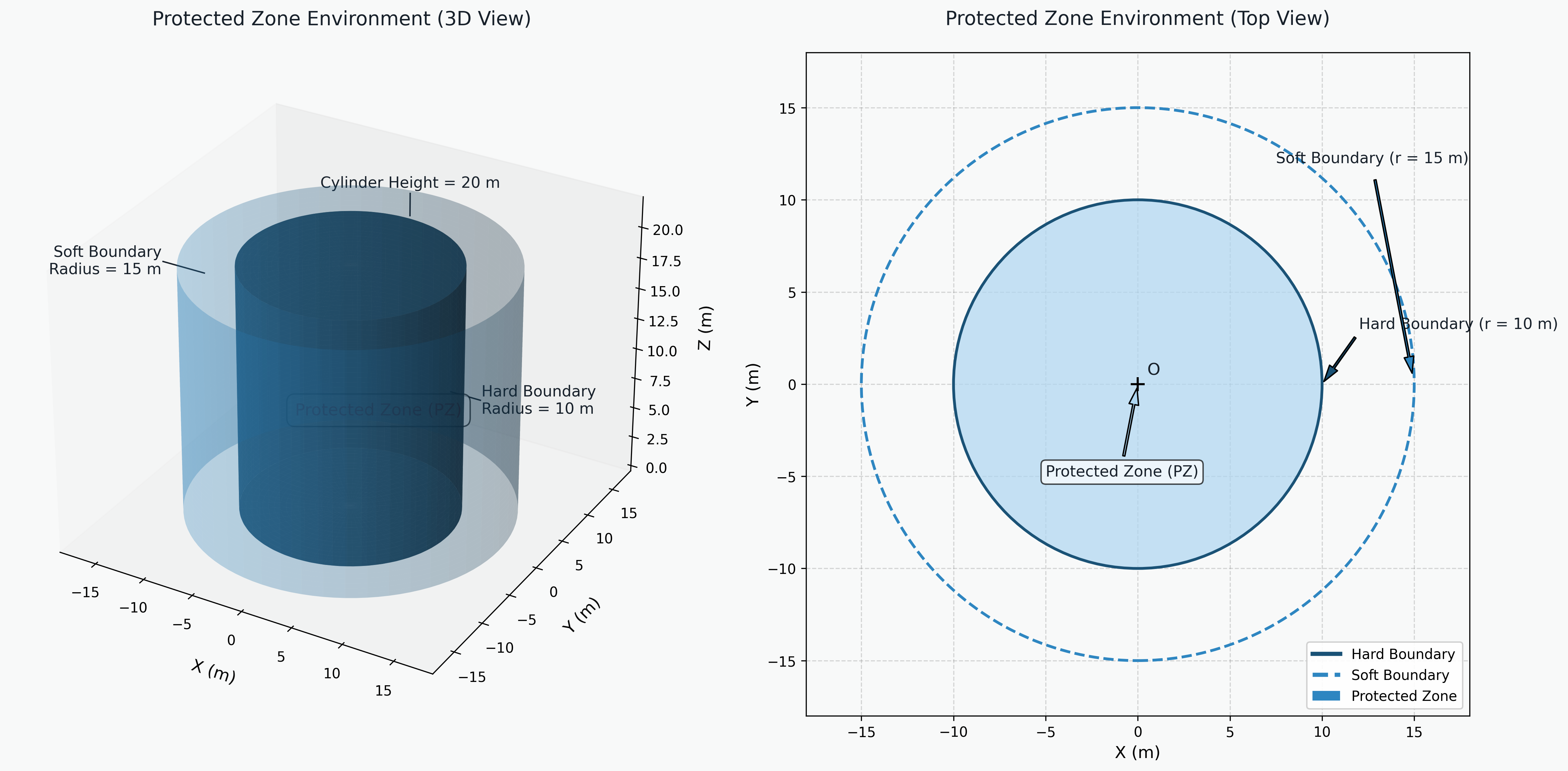}
    \caption{Simulation environment showing the protected zone (hard/soft boundaries).}
    \label{sim}
\end{figure}

\subsection{Nominal Simulation Results}

In both settings, namely deterministic and probabilistic, attackers are prioritized via the combined criticality metric and defenders are allocated via CRIDAA (bipartite matching with replanning and safety filtering), and an attacker is counted intercepted once it enters the capture ball of radius $r_{\mathrm{cap}}$. Fig.~\ref{UAS_defense_det} shows a representative deterministic simulation. 
Fig.~\ref{fig7} shows the attacker criticality over time with the end points representing the intercepted time for each attacker UAS and Fig.~\ref{fig8} shows the attacker status at each timestep. Over $5000$ Monte Carlo simulation runs are conducted in the deterministic setting, where the aggregate averages are: (i) interception rate $99.9\%$, (ii) breach rate $0.1 \%$, and  (iii) mean interception distance from the hard boundary $5.160\,\mathrm{u}$. 
(These aggregate metrics quantify typical performance; the figures are a single-run visualization.)

% \subsection{Probabilistic setting: aggregate performance and interpretation}
We next evaluate the probabilistic overlap-graph setting using the same environment/dynamics, with probabilistic parameters
$\sigma_{r_0}=10.0$, $k=0.2$, $\Gamma=8$, $\sigma_n=5$, $P_{\mathrm{detect}}^{\mathrm{th}}=0.10$, and $\alpha' = 0.04$.
Fig.~\ref{fig13} shows the overall performance of proposed defense architecture under probabilistic setting. 
This illustrates a key theme that is quantified in our ablation/sensitivity analysis: under stochastic triggering, the closed loop often achieves substantial neutralization throughput while still suffering occasional breaches. We evaluate $5000$ Monte Carlo runs for $200$ time steps while varying attacker initial conditions and probabilistic parameters.
Over $5000$ Monte Carlo runs in the probabilistic setting, the aggregate averages are: (i) interception rate $85.6\%$, (ii)  breach rate $14.4\%$ with average breach time of $10.53\,\mathrm{s}$, and (iii) mean interception distance from the hard boundary $4.88\,\mathrm{u}$ (Fig.~\ref{fig13}).

\captionsetup[subfigure]{labelformat=simple}
\renewcommand\thesubfigure{(\alph{subfigure})}

% \begin{figure}[ht]
%     \centering
    
%         \begin{minipage}{0.94\linewidth}
%             \centering
            
%             % First row: Criticality over Time
%             \begin{subfigure}{0.90\linewidth}
%                 \centering
%                 \includegraphics[width=\linewidth]{criticality_over_time.png}
%                 \caption{Criticality over Time}
%                 \label{fig7}
%             \end{subfigure}
%             \\ \vspace{0.3cm} % Forces the next subfigure to the next row with spacing
            
%             % Second row: Attacker Status
%             \begin{subfigure}{0.90\linewidth}
%                 \centering
%                 \includegraphics[width=\linewidth]{attacker_status_over_time.png}
%                 \caption{Attacker Status}
%                 \label{fig8}
%             \end{subfigure}
            
%         \end{minipage}
    
%     \caption{Deterministic graph-based UAS swarm defense}
%     \label{UAS_defense_det}
% \end{figure}

\begin{figure}[ht]
    \centering
        \begin{minipage}{0.94\linewidth}
            \centering            
            % First row: Criticality over Time
            \subfloat[Criticality over Time]{
            \includegraphics[width=0.90\linewidth]{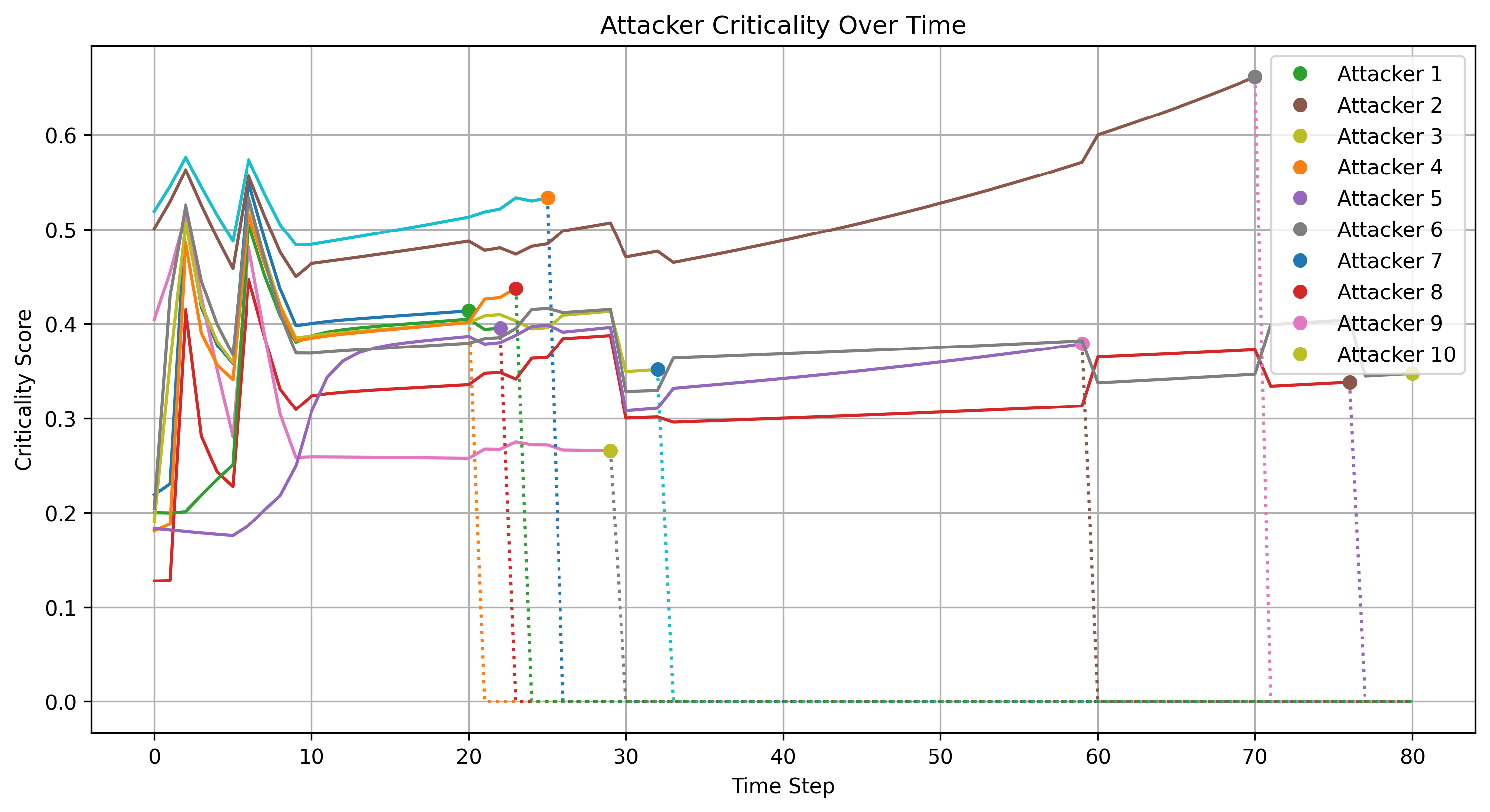}
            \label{fig7} }
            \\ \vspace{0.3cm} % Forces the next subfigure to the next row with spacing
            \subfloat[Attacker Status]{
            \includegraphics[width=0.90\linewidth]{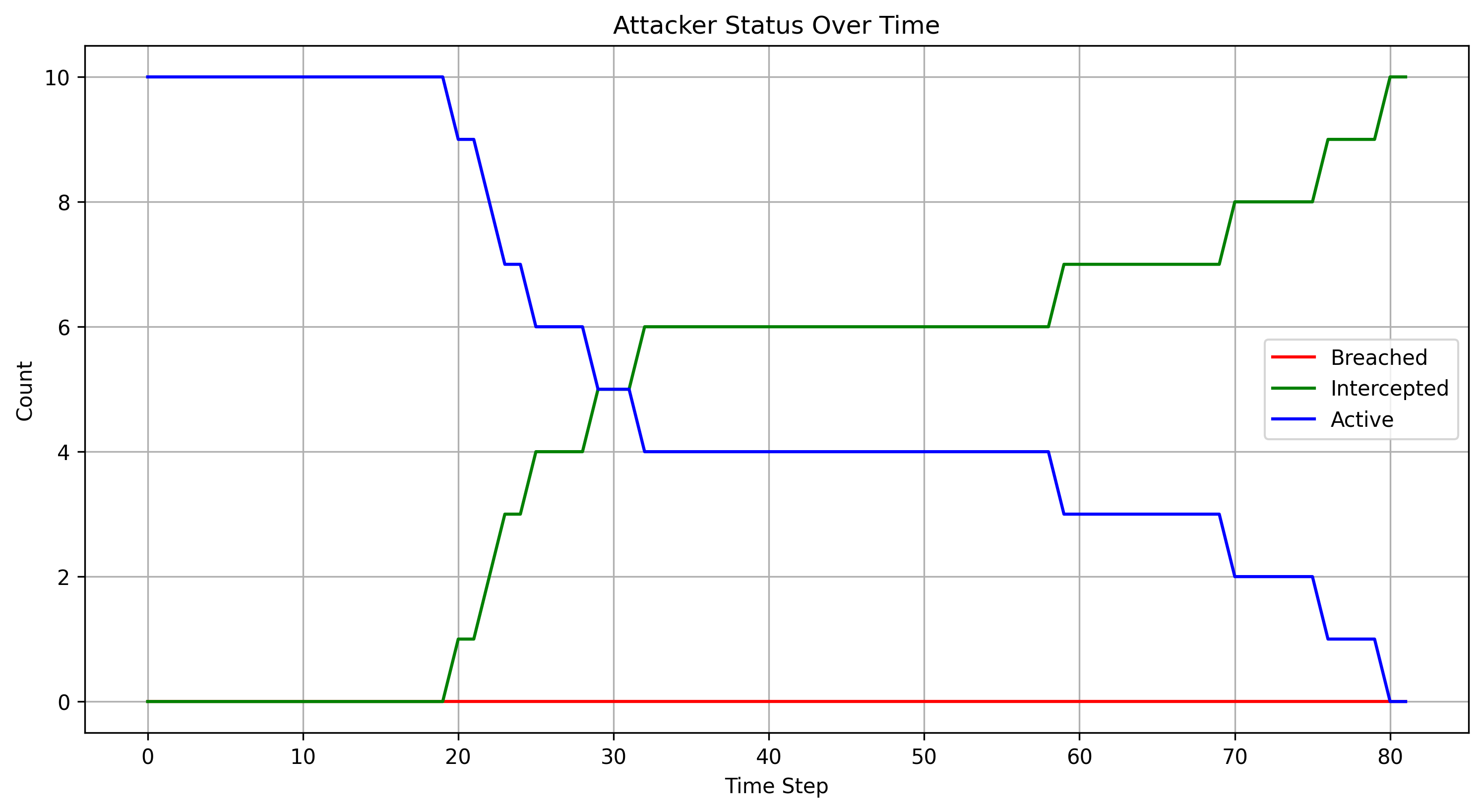}
            \label{fig8}}            
        \end{minipage}    
    \caption{Deterministic graph-based UAS swarm defense}
    \label{UAS_defense_det}
\end{figure}

\begin{figure}[htp]
    \centering
    \includegraphics[width=\linewidth]{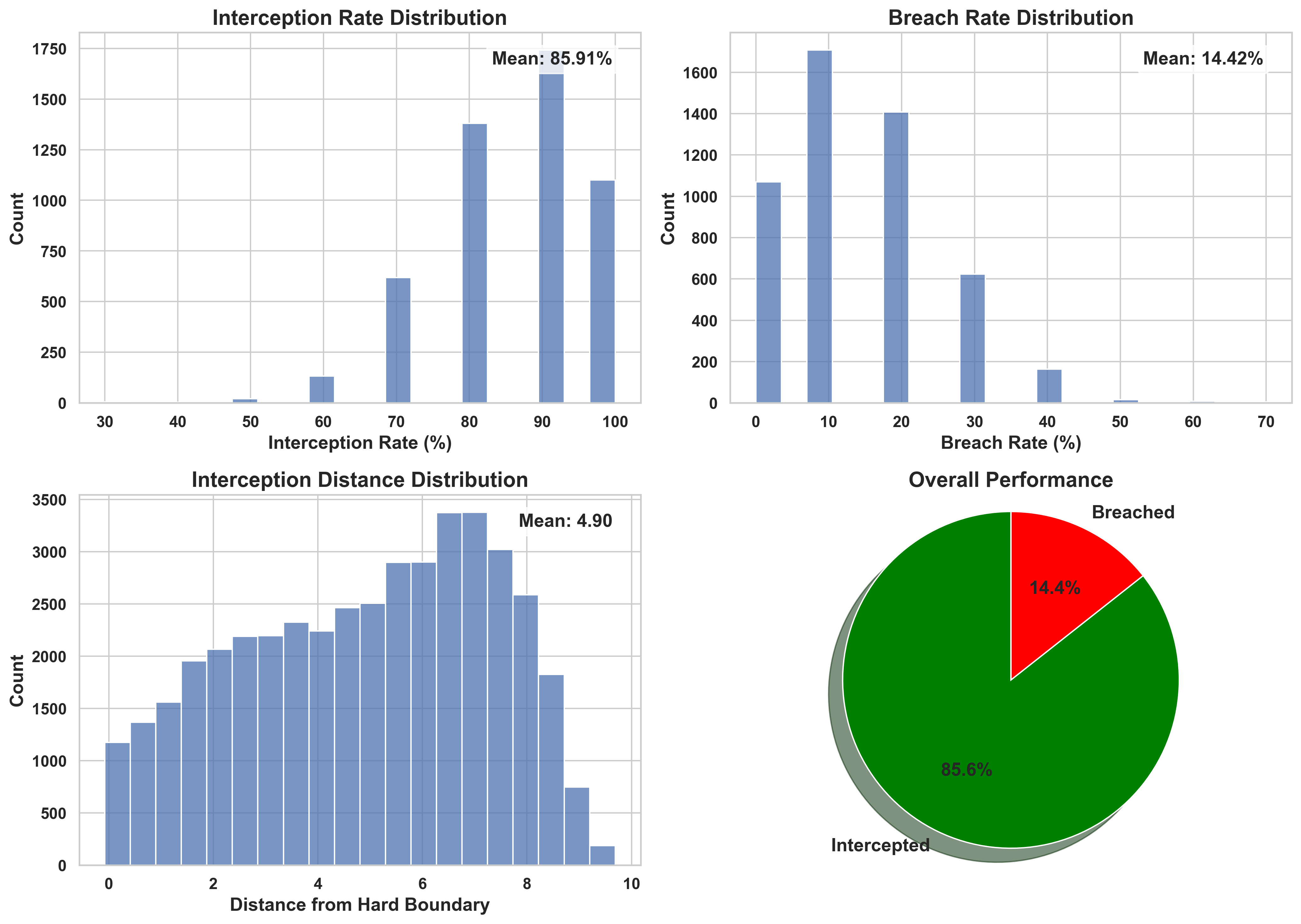}
    \caption{Probabilistic graph-based UAS swarm defense}
    \label{fig13}
\end{figure}

\subsection{Simulations with Non-Optimal Attacker Dynamics}

To assess whether the proposed defender-side architecture depends on the nominal minimum-time attacker guidance model, we evaluated the complete closed-loop defense pipeline under non-optimal attacker dynamics. In this regime, attackers do not follow the nominal time-optimal control law; instead, they evolve according to fixed heuristic motion patterns, as described in Section~\ref{subsec:nonoptimal_attackers}. All other components of the framework including sensing, probabilistic graph construction, criticality evaluation, assignment, and defender pursuit—are kept unchanged. A representative simulation from this non-optimal setting is shown in Fig.~\ref{fig:nonoptimal_summary}. The top-left panel of Fig.~\ref{fig:nonoptimal_summary} reports the evolution of attacker-status over time. The number of active attackers decreases monotonically from $10$ to $0$, while the cumulative number of intercepted attackers increases from $0$ to $10$ by approximately \(t \approx 140\). Over the same horizon, the breached-PZ curve remains identically zero in the displayed run. Thus, for the illustrated experiment, all attackers are neutralized before any hard-boundary breach occurs.

The top-right panel shows the mean criticality trajectories stratified by attacker behaviour mode. The four behaviour classes exhibit visibly different temporal criticality profiles, indicating that the proposed criticality layer remains responsive to behavioural heterogeneity even when the attacker dynamics are no longer generated by an optimal-control law. In particular, the biased-random and flanking-arc modes attain sustained moderate-to-high criticality over portions of the engagement, whereas the pure-random mode decays earlier and remains comparatively lower for much of the run shown. The bottom-left panel reports the graph statistics for the same experiment. Both the edge-count and average-edge-weight traces remain at zero throughout the run, indicating that, under the particular parameter realization shown, the uncertainty-aware graph layer does not form attacker-attacker edges. Consequently, in this representative case, the prioritization mechanism is driven primarily by the non-topological terms in the composite criticality score, rather than by graph-centrality contributions. The bottom-right panel provides the per-behaviour interception breakdown. The displayed totals are 4 biased-random attackers, 1 sinusoidal-weave attacker, 1 flanking-arc attacker, and 4 pure-random attackers, and in each case the intercepted count matches the total count. Accordingly, every attacker in each behaviour class is intercepted in the representative run.

\begin{figure}[t]
    \centering
    \includegraphics[width=\linewidth]{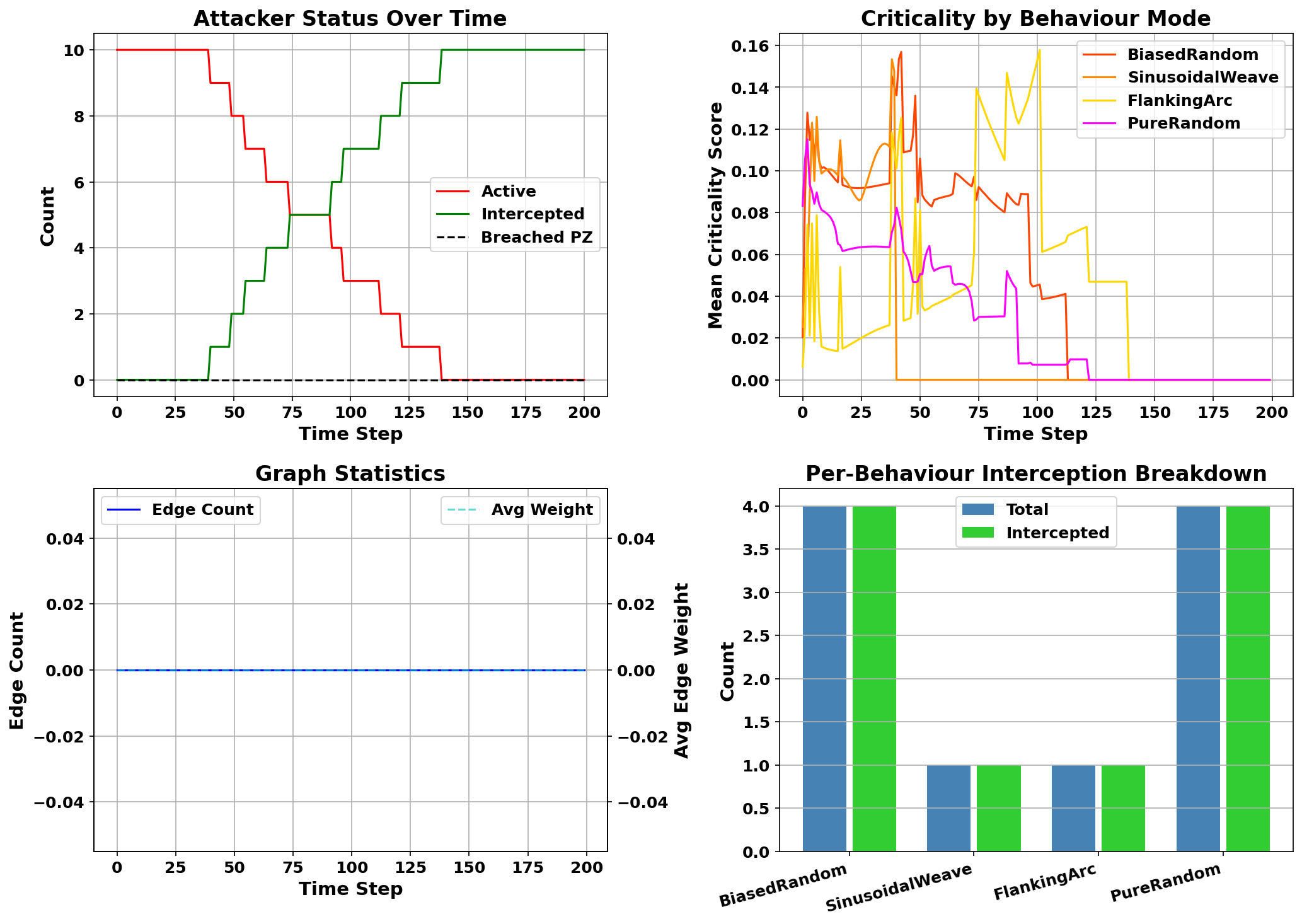}
    \caption{Representative simulation under non-optimal attacker dynamics. Top-left: time history of active, intercepted, and breached attackers. Top-right: mean attacker criticality grouped by behaviour mode. Bottom-left: edge count and average edge weight of the uncertainty-aware attacker graph. Bottom-right: per-behaviour totals and intercepted counts. In the illustrated run, all 10 attackers are intercepted, no attacker breaches the protected zone, and no graph edges are formed over the displayed horizon.}
    \label{fig:nonoptimal_summary}
\end{figure}

\subsection{Ablation Analysis}
\label{subsec:ablation_revised}

We perform an ablation analysis to isolate the role of the principal modules of the proposed defense architecture: (i) sensing-induced graph construction, (ii) centrality-based prioritization, (iii) Markov-chain risk prediction, (iv) switching and reassignment regulation, and (v) assignment strategy, under operating regimes chosen to activate different limiting mechanisms (Fig.~\ref{fig:ablation_blockdiagram}). Because the architecture is designed for uncertainty-driven closed-loop operation, this study has two goals: first, to identify the \emph{performance-limiting} module in each regime, and second, to quantify the resulting trade-off among reliability, first-capture latency, and computational load (Table~\ref{tab:ablation_summary}). All variants are evaluated through Monte Carlo simulations with $N_{\mathrm{runs}}=100$ trials per (regime, variant) pair, using identical seeded randomization across variants within each regime. We consider three regimes: (i) \emph{Nominal}, a decision-limited regime in which sensing and kinematics are sufficiently permissive that threat ranking and allocation materially affect the outcome; (ii) \emph{DegradedSensing}, a detection-limited regime induced by a conservative detection threshold, in which actionable multi-target windows are sparse and intermittent; and (iii) \emph{HardKinematics}, a physics-limited regime induced by reduced defender authority and tighter interception geometry, in which pursuit feasibility dominates closed-loop performance.

% \paragraph{Run-level severity and throughput.}
Since multi-attacker defense can produce partially successful episodes, breach-free reliability alone is not enough to describe the closed-loop behavior. We therefore also report the run-level breach and neutralization counts. Specifically, \texttt{breach\_count} measures breach severity, namely the number of attackers that penetrate in failure runs, while \texttt{neutralized\_count} measures neutralization throughput; both can vary with sensing uncertainty, switching, and allocation under uncertainty. This separation is important for interpretation: a module may have little effect on $\mathbb{P}(\mathrm{no\ breach})$ in a bottle necked regime, yet still improve neutralization throughput or reduce breach severity in failure runs.

% \paragraph{Internal statistics and alignment with the drift model.}
To connect the empirical outcomes to the drift-type mechanism underlying the theoretical analysis, we report the empirical drift proxy $\widehat{\Theta}=\bar{\eta}\,p_{\min}$, where $\bar{\eta}$ is the run-averaged execution fraction, consistent with $\eta_k$, and $p_{\min}$ is a robust lower-tail estimate of the logged feasibility/success statistic computed over eligible decision windows (Table~\ref{tab:ablation_summary}). This proxy is consistent with the effectiveness parameter $\Theta_k=p^\star(t_k)\eta_k$ used in the post-detection drift inequalities, and it serves as a diagnostic to distinguish between (i) \emph{information/feasibility-gated} regimes, where $\widehat{\Theta}$ sits near its floor, and (ii) \emph{decision-limited} regimes, where $\widehat{\Theta}$ remains nontrivial and ranking and assignment determine how opportunities are exploited. Accordingly, ablation effects are interpreted jointly through $\mathbb{P}(\mathrm{no\ breach})$, $\mathbb{E}[\kappa_1]$, and $\widehat{\Theta}$. As summarized in Fig.~\ref{fig:ablation_blockdiagram}, \texttt{NO\_MARKOV} disables the Markov risk predictor, \texttt{NO\_SWITCH} disables reassignment switching, \texttt{DET\_GRAPH} replaces overlap-based graph construction with a proximity graph, \texttt{NO\_CENTRALITY} removes centrality-based prioritization, \texttt{GREEDY\_ASSIGN} replaces the assignment method with a greedy rule, and \texttt{TIME\_ONLY} removes risk-aware weighting and collapses the decision objective to time-only weighting.

\subsubsection{Nominal regime}
\label{subsubsec:ablation_nominal_revised}

In the Nominal regime, \texttt{FULL} attains $\mathbb{P}(\mathrm{no\ breach})=0.230$ with $95\%$ CI $[0.147,0.313]$ and $\mathbb{E}[\kappa_1]=17.890$ with 95\% CI $[16.329,19.451]$ (Table~\ref{tab:ablation_summary}). Because sensing and kinematics are sufficiently permissive, the drift proxy remains of the order of $10^{-1}$ across the reported decision-layer variants, with $\widehat{\Theta}=0.155$ for \texttt{FULL}. This indicates that the system is not opportunity-scarce and that decision-layer choices can materially change breach outcomes by directing available engagement capacity toward the most consequential threats. In this regime, assignment quality is the most sensitive determinant of breach-free reliability: \texttt{GREEDY\_ASSIGN} reduces $\mathbb{P}(\mathrm{no\ breach})$ to $0.070$ with 95\% CI $[0.020,0.120]$ despite a comparable $\mathbb{E}[\kappa_1]$ (Table~\ref{tab:ablation_summary}). This behavior is consistent with misallocation: early captures may still occur, but later windows are spent on subcritical or poorly matched attackers, which increases the probability that at least one high-urgency attacker reaches the protected boundary.

% \paragraph{Role of Markov risk prediction and switching.}
Disabling Markov prediction (\texttt{NO\_MARKOV}) preserves breach-free reliability at this operating point while reducing mean runtime from $2.611\,\mathrm{s}$ to $0.387\,\mathrm{s}$  (Table~\ref{tab:ablation_summary}). In a decision-limited regime where the assignment layer already has frequent actionable windows, the Markov layer mainly refines ordering and short-horizon prioritization, so it need not change the leading-order breach-free reliability, although it remains relevant under stress and for shaping severity and throughput in failure runs. Similarly, disabling switching (\texttt{NO\_SWITCH}) yields comparable reliability ($0.220$) and $\mathbb{E}[\kappa_1]=17.760$ (Table~\ref{tab:ablation_summary}), which is consistent with switching being a non-monotone design lever whose benefit depends on whether the information state changes enough to justify reassignment without introducing switching-induced execution loss.

\subsubsection{DegradedSensing regime}
\label{subsubsec:ablation_degraded_revised}

The DegradedSensing regime enforces conservative detection and produces an information-limited closed loop in which multi-target actionable windows are infrequent, so downstream refinements are often masked. This bottleneck is reflected by saturation of $\widehat{\Theta}$ at its floor, $\widehat{\Theta}=10^{-3}$, for all reported variants (Table~\ref{tab:ablation_summary}), which indicates that the post-detection drift mechanism is seldom realized at a rate sufficient to guarantee consistent depletion within the finite horizon. Under this regime, \texttt{FULL} achieves $\mathbb{P}(\mathrm{no\ breach})=0.270$ with $95\%$ CI $[0.183,0.357]$ and $\mathbb{E}[\kappa_1]=16.890$ with $95\%$ CI $[15.635,18.145]$, and disabling Markov prediction produces comparable aggregates (Table~\ref{tab:ablation_summary}). This is consistent with the interpretation that predictive refinement cannot compensate for the scarcity of high-confidence engagement opportunities created by conservative sensing. Even so, allocation remains consequential: \texttt{GREEDY\_ASSIGN} reduces breach-free reliability to $0.110$ with $95\%$ CI $[0.048,0.172]$ (Table~\ref{tab:ablation_summary}), showing that rare opportunities must still be used efficiently to prevent breaches under partial observability.

\subsubsection{HardKinematics regime}
\label{subsubsec:ablation_hardkin_revised}

HardKinematics stresses pursuit feasibility by reducing defender authority and tightening interception geometry, so the dominant limitation shifts from decision quality to physical feasibility. In this regime, all variants exhibit reduced breach-free reliability and increased first-hit latency relative to Nominal (Table~\ref{tab:ablation_summary}). In particular, \texttt{NO\_MARKOV} achieves a breach-free reliability close to \texttt{FULL} while being substantially faster computationally, which indicates that when capture feasibility and geometry dominate the closed-loop outcome, additional predictive computation does not change the active bottleneck, even though it remains aligned with the risk-aware design of the full architecture.

% \paragraph*{Implications}
Collectively, Fig.~\ref{fig:ablation_blockdiagram} and Table~\ref{tab:ablation_summary} support three regime-consistent conclusions: (i) globally optimal assignment is performance-critical in decision-limited regimes, (ii) prediction and switching act as information-to-action refinements whose value increases when sensing provides sufficiently informative and frequent decision windows, and (iii) in detection-limited and physics-limited regimes, the masking of decision-layer refinements is diagnostically informative for regime identification and tuning, rather than contradictory to the role of uncertainty-aware modules.
\begin{table*}[t]
\centering
\caption{Ablation results across operating regimes. We report breach-free reliability $\mathbb{P}(\mathrm{no\ breach})$ with 95\% CI, mean first-interception (first IF-trigger) epoch $\mathbb{E}[\kappa_1]$ with 95\% CI, mean wall-clock runtime, and the empirical drift proxy $\widehat{\Theta}$ with 95\% CI. Each (regime, variant) uses $N_{\mathrm{runs}}=100$ Monte Carlo trials.}
\label{tab:ablation_summary}
\scriptsize
\setlength{\tabcolsep}{4.5pt}
\begin{adjustbox}{max width=\textwidth}
\begin{tabular}{llcccccccccccc}
\toprule
Regime & Variant &
$\mathbb{P}(\mathrm{no\ breach})$ & CI$_{95}$ low & CI$_{95}$ high &
$\mathbb{E}[\kappa_1]$ & CI$_{95}$ low & CI$_{95}$ high &
Runtime (s)  & $\widehat{\Theta}$ & CI$_{95}$ low & CI$_{95}$ high \\
\midrule
\multicolumn{12}{l}{\textbf{Nominal} $(v_{\mathrm{def}}^{\max},P_{\mathrm{detect}}^{\mathrm{th}},r_{\mathrm{int}},\alpha_{\mathrm{graph\text{-}overlap}})=(4.0,0.10,1.5,0.01)$} \\
\midrule
Nominal & FULL & 0.23 & 0.147& 0.313 & 17.89 & 16.329 & 19.451 & 2.611  & 0.155 & 0.127 & 0.183 \\
Nominal & DET\_GRAPH & 0.25 & 0.165 & 0.335 & 17.98 & 16.431 & 19.528 & 2.770  & 0.160 & 0.131 & 0.189 \\
Nominal & NO\_CENTRALITY & 0.26 & 0.173 & 0.346 & 17.92 & 16.377 & 19.463 & 2.835 & 0.164 & 0.136 & 0.192 \\
Nominal & NO\_MARKOV & 0.23 & 0.147 & 0.313 & 18.14 & 16.536 & 19.743 & 0.386 & 0.160 & 0.133 & 0.188 \\
Nominal & NO\_SWITCH & 0.22 & 0.138 & 0.301 & 17.76 & 16.323 & 19.191 & 2.963 & 0.159 & 0.131 & 0.188 \\
Nominal & GREEDY\_ASSIGN & 0.07 & 0.019 & 0.120 & 18.09 & 16.720 & 19.459 & 3.237  & 0.190 & 0.163 & 0.218 \\
Nominal & TIME\_ONLY & 0.26 & 0.173 & 0.346 & 19.47 & 17.546 & 21.393 & 0.230  & 0.147 & 0.118 & 0.175 \\
\midrule
\multicolumn{12}{l}{\textbf{DegradedSensing} $(v_{\mathrm{def}}^{\max},P_{\mathrm{detect}}^{\mathrm{th}},r_{\mathrm{int}},\alpha_{\mathrm{graph\text{-}overlap}})=(4.0,0.35,1.5,0.01)$} \\
\midrule
DegradedSensing & FULL & 0.27 & 0.182 & 0.357 & 16.89 & 15.635 & 18.144 & 9.221  & 0.001 & 0.001 & 0.001 \\
DegradedSensing & DET\_GRAPH & 0.27 & 0.182 & 0.357 & 16.89 & 15.635 & 18.144 & 2.641  & 0.001 & 0.001 & 0.001 \\
DegradedSensing & NO\_CENTRALITY & 0.29 & 0.200 & 0.379 & 16.80 & 15.589 & 18.010 & 2.414  & 0.001 & 0.001 & 0.001 \\
DegradedSensing & NO\_MARKOV & 0.29 & 0.200 & 0.379 & 16.88 & 15.633 & 18.126 & 0.225  & 0.001 & 0.001 & 0.001 \\
DegradedSensing & NO\_SWITCH & 0.27 & 0.182 & 0.357 & 16.89 & 15.635 & 18.144 & 2.622  & 0.001 & 0.001 & 0.001 \\
DegradedSensing & GREEDY\_ASSIGN & 0.11 & 0.048 & 0.171 & 16.90 & 15.784 & 18.015 & 2.832 & 0.001 & 0.001 & 0.001 \\
DegradedSensing & TIME\_ONLY & 0.27 & 0.182 & 0.357 & 17.98 & 16.548 & 19.411 & 0.239  & 0.001 & 0.001 & 0.001 \\
\midrule
\multicolumn{12}{l}{\textbf{HardKinematics} $(v_{\mathrm{def}}^{\max},P_{\mathrm{detect}}^{\mathrm{th}},r_{\mathrm{int}},\alpha_{\mathrm{graph\text{-}overlap}})=(3.5,0.10,1.2,0.01)$} \\
\midrule
HardKinematics & FULL & 0.17 & 0.096 & 0.243 & 21.39 & 19.463 & 23.316 & 3.169  & 0.190 & 0.159 & 0.220 \\
HardKinematics & DET\_GRAPH & 0.16 & 0.087 & 0.232 & 21.24 & 19.346 & 23.133 & 6.020  & 0.180 & 0.150 & 0.210 \\
HardKinematics & NO\_CENTRALITY & 0.18 & 0.104 & 0.255 & 21.54 & 19.605 & 23.474 & 21.225  & 0.178 & 0.147 & 0.209 \\
HardKinematics & NO\_MARKOV & 0.17 & 0.096 & 0.243 & 20.83 & 19.212 & 22.447 & 0.427  & 0.183 & 0.151 & 0.215 \\
HardKinematics & NO\_SWITCH & 0.17 & 0.096 & 0.243 & 20.91 & 19.174 & 22.645 & 6.313  & 0.174 & 0.142 & 0.205 \\
HardKinematics & GREEDY\_ASSIGN & 0.09 & 0.033 & 0.146 & 20.58 & 19.265 & 21.894 & 6.344  & 0.222 & 0.195 & 0.248 \\
HardKinematics & TIME\_ONLY & 0.16 & 0.087 & 0.232 & 22.98 & 20.666 & 25.293 & 0.716  & 0.192 & 0.162 & 0.222 \\
\bottomrule
\end{tabular}
\end{adjustbox}
\end{table*}

% ------------------ FIGURE --------------------
\begin{figure*}[http]
\centering
\resizebox{\textwidth}{!}{%
\begin{tikzpicture}[
    font=\small,
    % ---- Color palette (soft, print-friendly) ----
    cDet/.style   ={fill=Blue!10,   draw=Blue!55!black},
    cGraph/.style ={fill=TealBlue!10, draw=TealBlue!60!black},
    cCent/.style  ={fill=OliveGreen!10, draw=OliveGreen!55!black},
    cMkv/.style   ={fill=Plum!10,   draw=Plum!55!black},
    cSw/.style    ={fill=Orange!12, draw=Orange!65!black},
    cAsg/.style   ={fill=Red!10,    draw=Red!60!black},
    cExec/.style  ={fill=Gray!12,   draw=Gray!70!black},
    % ---- Node styles ----
    block/.style={rectangle, rounded corners=1.2pt, align=center,
                  minimum height=8mm, minimum width=20mm, line width=0.45pt},
    switch/.style={rectangle, align=left, inner sep=1.5pt, font=\scriptsize,
                   fill=Gray!6, draw=Gray!55, line width=0.35pt},
    arrow/.style={-Latex, line width=0.45pt},
    dashedbox/.style={draw, dashed, rounded corners=2pt, inner sep=3pt,
                      fill=Gray!4, fill opacity=0.35, text opacity=1}
]

% Blocks
\node[block, cDet] (det) {Perception / Detection\\[-1mm]\scriptsize detections, $P_{\mathrm{detect}}$};

\node[block, cGraph, right=5mm of det] (graph) {Graph Builder};
\node[switch, below=1mm of graph] (graphsw) {Mode select:\\
\texttt{overlap} (FULL)\\
\texttt{proximity} (DET\_GRAPH)};

\node[block, cCent, right=6mm of graph] (cent) {Centrality Scoring};
\node[switch, below=1mm of cent] (centsw) {\textbf{OFF} in \texttt{NO\_CENTRALITY}};

\node[block, cMkv, right=6mm of cent] (mkv) {Markov Risk Predictor\\[-1mm]\scriptsize risk score};
\node[switch, below=1mm of mkv] (mkvsw) {\textbf{OFF} in \texttt{NO\_MARKOV}};

\node[block, cSw, right=6mm of mkv] (sw) {Policy Switching};
\node[switch, below=1mm of sw] (swsw) {\textbf{OFF} in \texttt{NO\_SWITCH}};

\node[block, cAsg, right=6mm of sw] (assign) {Assignment};
\node[switch, below=1mm of assign] (asw) {Mode select:\\
Hungarian (FULL)\\
Greedy (\texttt{GREEDY\_ASSIGN})};

\node[block, cExec, right=5mm of assign] (exec) {Interception / Execution\\[-1mm]\scriptsize neutralizations, breaches};

% Arrows
\draw[arrow] (det) -- (graph);
\draw[arrow] (graph) -- (cent);
\draw[arrow] (cent) -- (mkv);
\draw[arrow] (mkv) -- (sw);
\draw[arrow] (sw) -- (assign);
\draw[arrow] (assign) -- (exec);

% Decision layer box
\node[dashedbox, fit=(graph) (cent) (mkv) (sw) (assign) (graphsw) (centsw) (mkvsw) (swsw) (asw),
      label={[font=\small]above:Decision Layer}] (dl) {};

\node[switch, below=8mm of dl.south, anchor=north] (leg)
{Legend: \textbf{OFF} disables module (ablation); Mode select switches implementation.};

\end{tikzpicture}%
}
\caption{Modular defense pipeline and ablation toggles (e.g., \texttt{NO\_MARKOV}, \texttt{NO\_SWITCH}).}
\label{fig:ablation_blockdiagram}
\end{figure*}
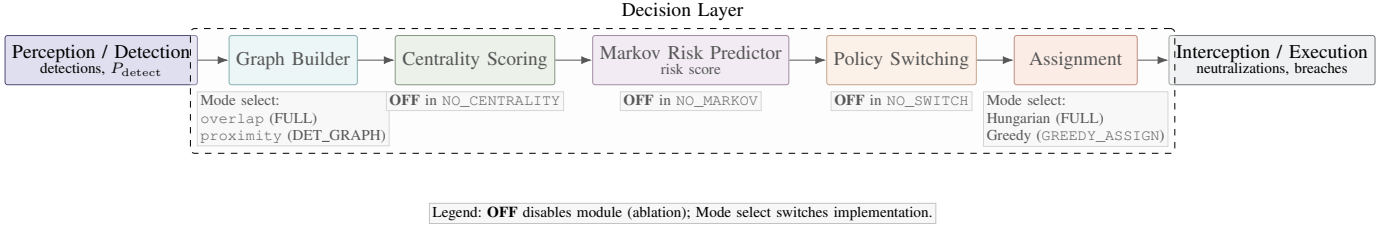

\subsection{Sensitivity Analysis}
\label{subsec:sensitivity_revised}

In this section, we report results from sensitivity analysis to quantify how breach-free reliability and first-interception latency vary with key tunable parameters governing (i) detection triggering via $P_{\det}^{\mathrm{th}}$, (ii) defender actuation authority via $v_{def}^{\max}$, and (iii) interaction-graph coupling via the overlap threshold $\alpha_{\mathrm{graph\text{-}overlap}}$.
Unlike the ablation study, which removes modules to attribute performance, the sensitivity study retains the full architecture and evaluates parametric robustness under stochastic sensing and engagement (Tables~\ref{tab:sensA_table}--~\ref{tab:sensB_table}; Fig.~\ref{fig:sensitivity_combined}). For each parameter cell, we run $N_{\mathrm{reps}}=100$ Monte Carlo trials over a horizon of $T=150$ timesteps and report aggregated statistics (Tables~\ref{tab:sensA_table}--~\ref{tab:sensB_table}). We report breach-free reliability $\mathbb{P}(\mathrm{no\ breach})$ with $95\%$ confidence intervals and the mean breach severity $\mathbb{E}[\texttt{breach\_count}]$.
To quantify capture triggering under stochastic sensing, we report the mean ($\mathbb{E}[\kappa_1]$) and median first-interception epoch ($\mathrm{med}(\kappa_1)$), together with the mean active edge count as an empirical coordination proxy; we additionally report drift-aligned progress diagnostics (including $\Theta_{\mathrm{mean}}$ and the feasibility statistic $p_{\min,\mathrm{med}}$) because edge density alone does not certify effective coordination under feasibility and execution constraints.

% \paragraph{Interpretation principle- reliability versus latency}

The first-hit statistic $\kappa_1$ quantifies how rapidly the closed loop \emph{triggers its first successful capture}, whereas $\mathbb{P}(\mathrm{no\ breach})$ is a strict end-to-end metric that depends on the \emph{entire subsequent sequence} of detections, allocations, executed engagements, and residual attacker dynamics over the full episode.
Consequently, parameter cells can exhibit similar $\mathrm{med}(\kappa_1)$ yet materially different breach-free reliability if the post-first-capture execution fraction and the effective engagement success floor differ, which is precisely the role of the progress diagnostics reported in Tables~\ref{tab:sensA_table}--~\ref{tab:sensB_table}.

\subsubsection{Sweep A}
\label{subsubsec:sweepA_revised}

We sweep the detection threshold $P_{\det}^{\mathrm{th}}$ and defender actuation authority $v_{def}^{\max}$ while holding the overlap parameter fixed (Table~\ref{tab:sensA_table}).
Figures~\ref{fig:sensA_success}--~\ref{fig:sensA_tau1} visualize heatmaps of $\mathbb{P}(\mathrm{no\ breach})$ and $\mathrm{med}(\kappa_1)$, respectively.
Across the grid, breach-free reliability spans $0.10$--$0.31$ (Table~\ref{tab:sensA_table}), confirming that the sweep probes a nontrivial operating envelope and is therefore informative for tuning rather than being confined to degenerate success or failure regions.
At fixed $P_{\det}^{\mathrm{th}}$, $\mathbb{P}(\mathrm{no\ breach})$ is not uniformly monotone in $v_{def}^{\max}$, indicating a genuine interaction between kinematic authority, windowed feasibility, and stochastic triggering under reassignment and execution constraints, rather than a trivial ``more speed always helps'' artifact.
Furthermore, for conservative thresholds, the progress diagnostics collapse (Table~\ref{tab:sensA_table} reports $\Theta_{\mathrm{mean}}\approx 10^{-3}$ and $p_{\min,\mathrm{med}}\approx 10^{-3}$ in the most gated cells), which is consistent with an information-limited regime in which the post-detection depletion mechanism is rarely activated; in this region, improvements in actuation authority cannot be systematically converted into reliable containment within the finite horizon.

\subsubsection{Sweep B}
\label{subsubsec:sweepB_revised}

We sweep the detection threshold $P_{\det}^{\mathrm{th}}$ and overlap edge-formation threshold $\alpha_{\mathrm{graph\text{-}overlap}}$ while holding defender speed fixed (Table~\ref{tab:sensB_table}).
Figs.~\ref{fig:sensB_success}-~\ref{fig:sensB_tau1} depict heatmaps of $\mathbb{P}(\mathrm{no\ breach})$ and $\mathrm{med}(\kappa_1)$, respectively.
Reliability is not expected to vary monotonically with $\alpha_{\mathrm{graph\text{-}overlap}}$: decreasing $\alpha_{\mathrm{graph\text{-}overlap}}$ makes the inferred interaction graph dense (increasing coupling and potential coordination overhead), whereas increasing it results in the graph becoming sparse (reducing coupling and potentially discarding informative structure), and the net effect depends on whether the induced coupling aligns with feasible engagement opportunities and the allocation objective under uncertainty.
Accordingly, Table~\ref{tab:sensB_table} reports both edge density and progress diagnostics. Regimes in which $\Theta_{\mathrm{mean}}$ and $p_{\min,\mathrm{med}}$ saturate near their floor are effectively driftless over the finite horizon; in such cases, adjusting $\alpha_{\mathrm{graph\text{-}overlap}}$ cannot restore sustained progress. By contrast, in coordination-active cells characterized by nontrivial $\Theta_{\mathrm{mean}}$, overlap tuning directly reshapes the risk-ranking landscape that governs assignment and can therefore jointly influence reliability and first-hit latency.
Finally, the joint reporting of mean and median $\kappa_1$ reveals tail behavior under stochastic triggering: cells with comparable medians but inflated means indicate heavy-tailed first-hit times, motivating the use of both summaries to characterize stability of capture triggering under intermittent observability rather than relying on a single central tendency~statistic.

% ------------------- Table A (UPDATED, 100 runs/cell) -------------------
\begin{table*}[t]
\centering
\caption{Sensitivity sweep A: aggregated Monte Carlo results over $(P_{\mathrm{detect}}^{\mathrm{th}}, v_{\mathrm{def}}^{\max})$ with $N_{\mathrm{reps}}=100$ trials per cell and horizon $T=150$.}
\label{tab:sensA_table}
\scriptsize
\setlength{\tabcolsep}{4.0pt}
\begin{tabular}{ccccccccccccc}
\toprule
$P_{\mathrm{detect}}^{\mathrm{th}}$ & $v_{\mathrm{def}}^{\max}$ & runs &
$\mathbb{P}(\mathrm{no\ breach})$ & CI$_{95}$ (low) & CI$_{95}$ (high) &
$\mathbb{E}[\texttt{breach\_count}]$ &
$\kappa_1^{\mathrm{mean}}$ & $\kappa_1^{\mathrm{med}}$ &
$\Theta$ & $p_{\min}^{\mathrm{med}}$ &
$\mathbb{E}[\#\mathrm{edges}]$ & runtime (s) \\
\midrule
0.05 & 3.0 & 100 & 0.210 & 0.142 & 0.300 & 1.460 & 20.780 & 19.000 & 0.760 & 0.833 & 2.338 & 6.819 \\
0.05 & 3.5 & 100 & 0.210 & 0.142 & 0.300 & 1.590 & 19.890 & 17.500 & 0.753 & 0.824 & 2.061 & 6.552 \\
0.05 & 4.0 & 100 & 0.150 & 0.093 & 0.233 & 1.700 & 17.730 & 16.000 & 0.722 & 0.750 & 2.113 & 231.112 \\
0.05 & 4.5 & 100 & 0.160 & 0.101 & 0.244 & 1.570 & 16.530 & 14.000 & 0.736 & 0.833 & 2.205 & 6.195 \\
\midrule
0.10 & 3.0 & 100 & 0.190 & 0.125 & 0.278 & 1.480 & 19.860 & 18.000 & 0.103 & $10^{-3}$ & 0.823 & 6.057 \\
0.10 & 3.5 & 100 & 0.170 & 0.109 & 0.255 & 1.700 & 17.790 & 16.000 & 0.062 & $10^{-3}$ & 0.810 & 2.734 \\
0.10 & 4.0 & 100 & 0.190 & 0.125 & 0.278 & 1.480 & 16.590 & 15.000 & 0.112 & $10^{-3}$ & 0.839 & 4.464 \\
0.10 & 4.5 & 100 & 0.230 & 0.158 & 0.322 & 1.530 & 16.010 & 14.000 & 0.057 & $10^{-3}$ & 0.717 & 6.000 \\
\midrule
0.20 & 3.0 & 100 & 0.100 & 0.055 & 0.174 & 1.660 & 20.460 & 19.000 & $10^{-3}$ & $10^{-3}$ & 0.033 & 5.467 \\
0.20 & 3.5 & 100 & 0.210 & 0.142 & 0.300 & 1.340 & 16.900 & 16.000 & $10^{-3}$ & $10^{-3}$ & 0.030 & 3.654 \\
0.20 & 4.0 & 100 & 0.310 & 0.225 & 0.410 & 1.110 & 15.290 & 14.000 & $10^{-3}$ & $10^{-3}$ & 0.026 & 3.186 \\
0.20 & 4.5 & 100 & 0.220 & 0.151 & 0.310 & 1.280 & 12.950 & 12.000 & $10^{-3}$ & $10^{-3}$ & 0.053 & 3.379 \\
\midrule
0.25 & 3.0 & 100 & 0.160 & 0.101 & 0.244 & 1.430 & 18.080 & 17.000 & $10^{-3}$ & $10^{-3}$ & 0.000 & 3.922 \\
0.25 & 3.5 & 100 & 0.240 & 0.167 & 0.332 & 1.300 & 16.650 & 16.000 & $10^{-3}$ & $10^{-3}$ & 0.006 & 2.835 \\
0.25 & 4.0 & 100 & 0.230 & 0.158 & 0.322 & 1.280 & 15.600 & 14.000 & $10^{-3}$ & $10^{-3}$ & 0.001 & 3.011 \\
0.25 & 4.5 & 100 & 0.240 & 0.167 & 0.332 & 1.350 & 15.510 & 14.000 & $10^{-3}$ & $10^{-3}$ & 0.002 & 3.324 \\
\midrule
0.35 & 3.0 & 100 & 0.260 & 0.186 & 0.352 & 1.300 & 18.630 & 17.000 & $10^{-3}$ & $10^{-3}$ & 0.000 & 3.721 \\
0.35 & 3.5 & 100 & 0.290 & 0.210 & 0.384 & 1.120 & 17.460 & 16.000 & $10^{-3}$ & $10^{-3}$ & 0.000 & 2.424 \\
0.35 & 4.0 & 100 & 0.280 & 0.201 & 0.375 & 1.210 & 16.020 & 15.000 & $10^{-3}$ & $10^{-3}$ & 0.000 & 3.192 \\
0.35 & 4.5 & 100 & 0.250 & 0.175 & 0.343 & 1.370 & 14.880 & 13.000 & $10^{-3}$ & $10^{-3}$ & 0.000 & 5.299 \\
\midrule
0.40 & 3.0 & 100 & 0.280 & 0.201 & 0.375 & 1.270 & 18.360 & 18.000 & $10^{-3}$ & $10^{-3}$ & 0.000 & 3.152 \\
0.40 & 3.5 & 100 & 0.240 & 0.167 & 0.332 & 1.260 & 16.860 & 16.000 & $10^{-3}$ & $10^{-3}$ & 0.000 & 2.460 \\
0.40 & 4.0 & 100 & 0.280 & 0.201 & 0.375 & 1.340 & 16.710 & 14.000 & $10^{-3}$ & $10^{-3}$ & 0.000 & 52.577 \\
0.40 & 4.5 & 100 & 0.170 & 0.109 & 0.255 & 1.560 & 14.180 & 13.500 & $10^{-3}$ & $10^{-3}$ & 0.000 & 4.797 \\
\bottomrule
\end{tabular}
\end{table*}

% ------------------- Table B (UPDATED, 100 runs/cell) -------------------
\begin{table*}[htp]
\centering
\caption{Sensitivity sweep B: aggregated Monte Carlo results over $(P_{\mathrm{detect}}^{\mathrm{th}}, \alpha_{\mathrm{graph\text{-}overlap}})$ with $N_{\mathrm{reps}}=100$ trials per cell and horizon $T=150$.}
\label{tab:sensB_table}
\scriptsize
\setlength{\tabcolsep}{4.0pt}
\begin{tabular}{ccccccccccccc}
\toprule
$P_{\mathrm{detect}}^{\mathrm{th}}$ & $\alpha_{\mathrm{graph\text{-}overlap}}$ & runs &
$\mathbb{P}(\mathrm{no\ breach})$ & CI$_{95}$ (low) & CI$_{95}$ (high) &
$\mathbb{E}[\texttt{breach\_count}]$ &
$\kappa_1^{\mathrm{mean}}$ & $\kappa_1^{\mathrm{med}}$ &
$\Theta$ & $p_{\min}^{\mathrm{med}}$ &
$\mathbb{E}[\#\mathrm{edges}]$ & runtime (s) \\
\midrule
0.05 & 0.005 & 100 & 0.260 & 0.186 & 0.352 & 1.360 & 17.960 & 16.000 & 0.728 & 0.733 & 2.279 & 6.588 \\
0.05 & 0.010 & 100 & 0.130 & 0.075 & 0.214 & 1.880 & 18.120 & 16.000 & 0.751 & 0.817 & 2.249 & 6.455 \\
0.05 & 0.020 & 100 & 0.150 & 0.093 & 0.233 & 1.650 & 16.270 & 14.000 & 0.733 & 0.775 & 1.930 & 3.327 \\
0.05 & 0.050 & 100 & 0.210 & 0.142 & 0.300 & 1.460 & 17.410 & 15.000 & 0.712 & 0.735 & 1.751 & 4.191 \\
0.05 & 0.060 & 100 & 0.220 & 0.151 & 0.310 & 1.410 & 17.360 & 15.000 & 0.690 & 0.733 & 1.587 & 4.012 \\
0.05 & 0.080 & 100 & 0.140 & 0.082 & 0.225 & 1.780 & 17.350 & 15.000 & 0.742 & 0.817 & 1.582 & 5.979 \\
\midrule
0.10 & 0.005 & 100 & 0.210 & 0.142 & 0.300 & 1.510 & 16.530 & 14.000 & 0.079 & $10^{-3}$ & 0.839 & 5.520 \\
0.10 & 0.010 & 100 & 0.200 & 0.133 & 0.287 & 1.450 & 16.500 & 15.000 & 0.109 & $10^{-3}$ & 0.853 & 5.478 \\
0.10 & 0.020 & 100 & 0.240 & 0.167 & 0.332 & 1.340 & 15.780 & 14.000 & 0.071 & $10^{-3}$ & 0.717 & 3.325 \\
0.10 & 0.050 & 100 & 0.240 & 0.167 & 0.332 & 1.340 & 15.890 & 14.000 & 0.085 & $10^{-3}$ & 0.746 & 3.548 \\
0.10 & 0.060 & 100 & 0.210 & 0.142 & 0.300 & 1.450 & 15.950 & 14.000 & 0.090 & $10^{-3}$ & 0.707 & 5.550 \\
0.10 & 0.080 & 100 & 0.240 & 0.167 & 0.332 & 1.300 & 18.120 & 15.000 & 0.070 & $10^{-3}$ & 0.695 & 5.633 \\
\midrule
0.20 & 0.005 & 100 & 0.370 & 0.282 & 0.467 & 0.910 & 15.990 & 15.000 & $10^{-3}$ & $10^{-3}$ & 0.025 & 5.178 \\
0.20 & 0.010 & 100 & 0.260 & 0.186 & 0.352 & 1.180 & 15.830 & 14.000 & $10^{-3}$ & $10^{-3}$ & 0.030 & 5.785 \\
0.20 & 0.020 & 100 & 0.270 & 0.194 & 0.363 & 1.110 & 14.430 & 14.000 & $10^{-3}$ & $10^{-3}$ & 0.056 & 3.211 \\
0.20 & 0.050 & 100 & 0.260 & 0.186 & 0.352 & 1.170 & 16.660 & 15.000 & $10^{-3}$ & $10^{-3}$ & 0.034 & 3.380 \\
0.20 & 0.060 & 100 & 0.200 & 0.133 & 0.287 & 1.380 & 16.880 & 15.000 & $10^{-3}$ & $10^{-3}$ & 0.051 & 4.931 \\
0.20 & 0.080 & 100 & 0.260 & 0.186 & 0.352 & 1.090 & 16.060 & 15.000 & $10^{-3}$ & $10^{-3}$ & 0.018 & 5.789 \\
\midrule
0.25 & 0.005 & 100 & 0.250 & 0.175 & 0.343 & 1.180 & 16.520 & 15.000 & $10^{-3}$ & $10^{-3}$ & 0.005 & 5.278 \\
0.25 & 0.010 & 100 & 0.270 & 0.194 & 0.363 & 1.090 & 15.340 & 14.000 & $10^{-3}$ & $10^{-3}$ & 0.002 & 5.503 \\
0.25 & 0.020 & 100 & 0.160 & 0.101 & 0.244 & 1.520 & 16.810 & 15.000 & $10^{-3}$ & $10^{-3}$ & 0.001 & 3.194 \\
0.25 & 0.050 & 100 & 0.230 & 0.158 & 0.322 & 1.170 & 14.610 & 14.500 & $10^{-3}$ & $10^{-3}$ & 0.002 & 3.060 \\
0.25 & 0.060 & 100 & 0.240 & 0.167 & 0.332 & 1.150 & 15.980 & 14.000 & $10^{-3}$ & $10^{-3}$ & 0.004 & 4.922 \\
0.25 & 0.080 & 100 & 0.250 & 0.175 & 0.343 & 1.160 & 16.550 & 14.500 & $10^{-3}$ & $10^{-3}$ & 0.004 & 5.647 \\
\midrule
0.35 & 0.005 & 100 & 0.320 & 0.234 & 0.419 & 0.960 & 15.660 & 15.000 & $10^{-3}$ & $10^{-3}$ & 0.000 & 5.427 \\
0.35 & 0.010 & 100 & 0.290 & 0.210 & 0.384 & 1.070 & 16.100 & 15.000 & $10^{-3}$ & $10^{-3}$ & 0.000 & 5.918 \\
0.35 & 0.020 & 100 & 0.270 & 0.194 & 0.363 & 1.010 & 16.550 & 15.000 & $10^{-3}$ & $10^{-3}$ & 0.000 & 3.261 \\
\bottomrule
\end{tabular}
\end{table*}

\begin{figure*}[t]
    \centering    
    
    % Top Left: Success Rate
    \subfloat[Reliability $\mathbb{P}(\mathrm{no\ breach})$ over $(P_{\mathrm{detect}}^{\mathrm{th}}, v_{\mathrm{def}}^{\max})$]{
        \includegraphics[width=0.48\textwidth]{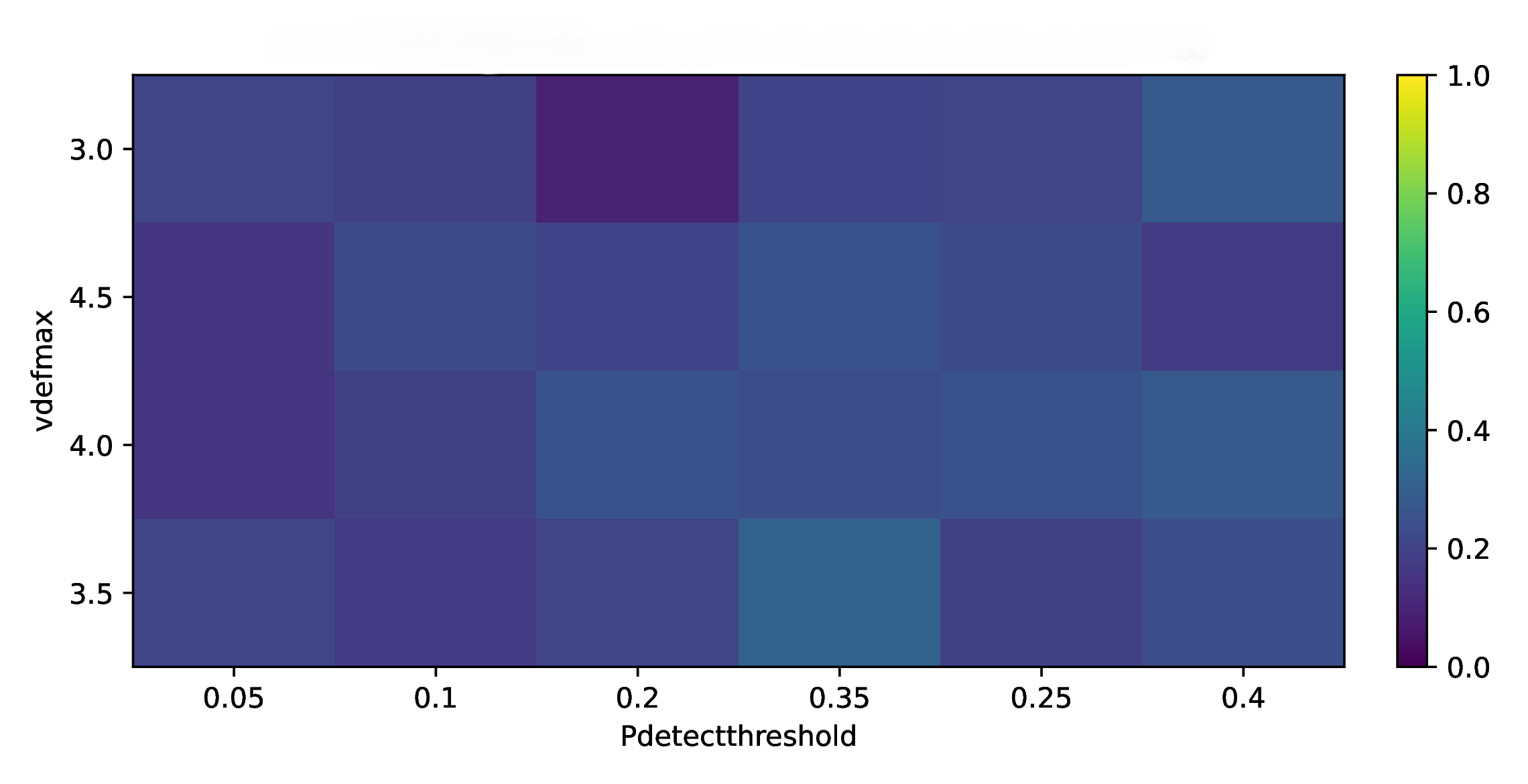}
        \label{fig:sensA_success}
    }
    \hfill
    % Top Right: Median Tau1
    \subfloat[Median $\kappa_1$ over $(P_{\mathrm{detect}}^{\mathrm{th}}, v_{\mathrm{def}}^{\max})$]{
        \includegraphics[width=0.48\textwidth]{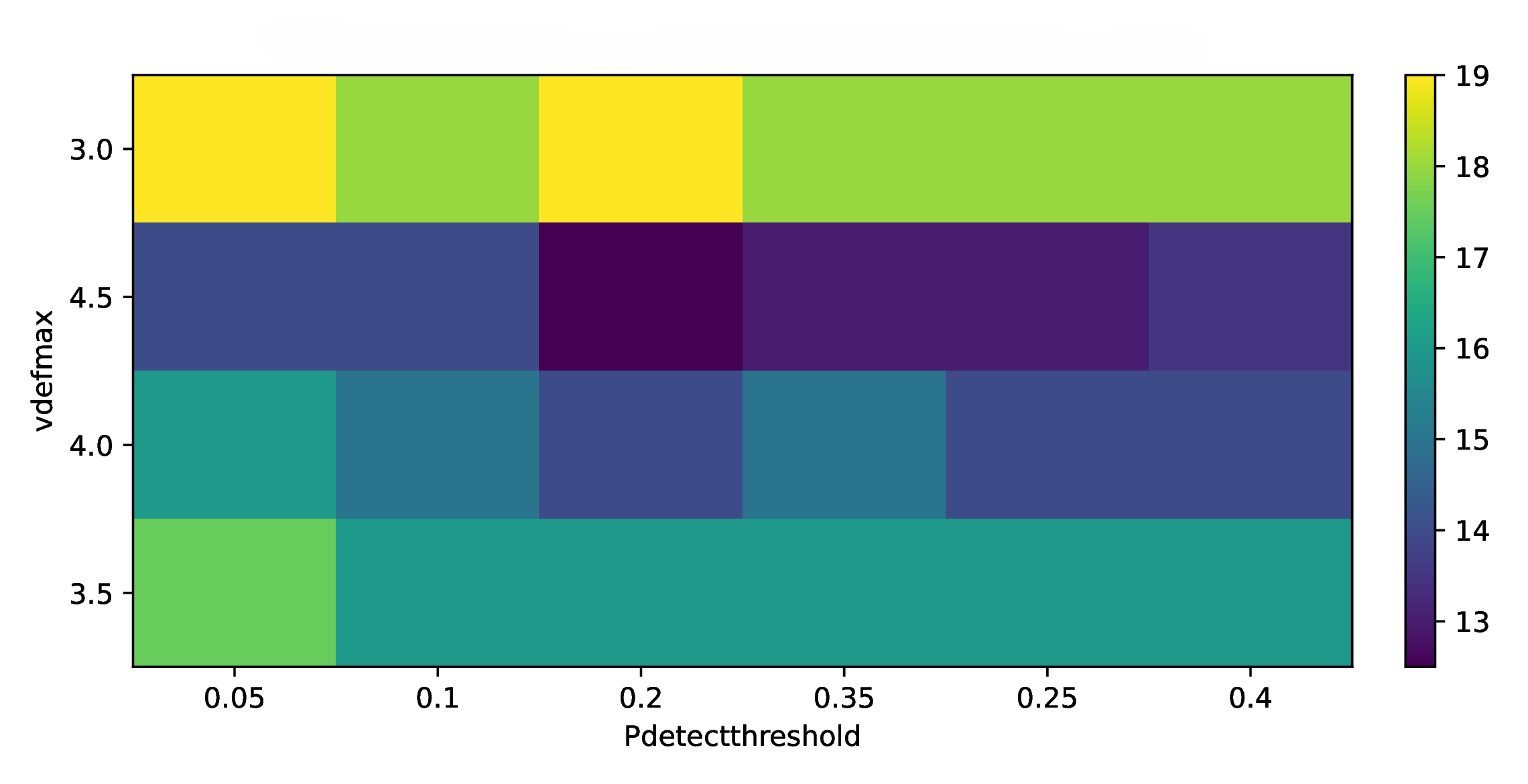}
        \label{fig:sensA_tau1}
    }
    
    \vspace{0.5cm}
    
    % Bottom Left: Success Rate
    \subfloat[Reliability $\mathbb{P}(\mathrm{no\ breach})$ over $(P_{\mathrm{detect}}^{\mathrm{th}}, \alpha_{\mathrm{graph\text{-}overlap}})$]{
        \includegraphics[width=0.48\textwidth]{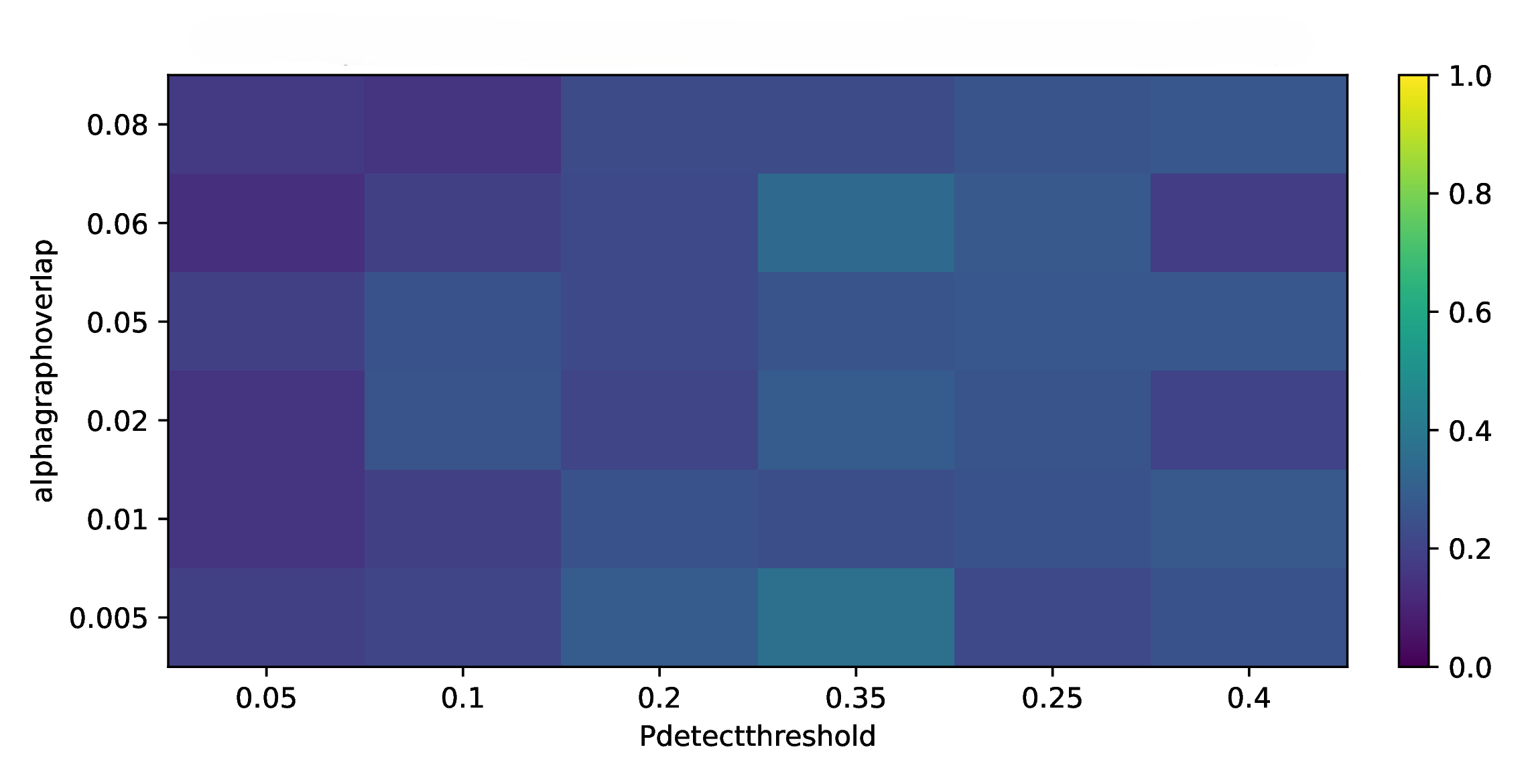}
        \label{fig:sensB_success}
    }
    \hfill
    % Bottom Right: Median Tau1
    \subfloat[Median $\kappa_1$ over $(P_{\mathrm{detect}}^{\mathrm{th}}, \alpha_{\mathrm{graph\text{-}overlap}})$]{
        \includegraphics[width=0.48\textwidth]{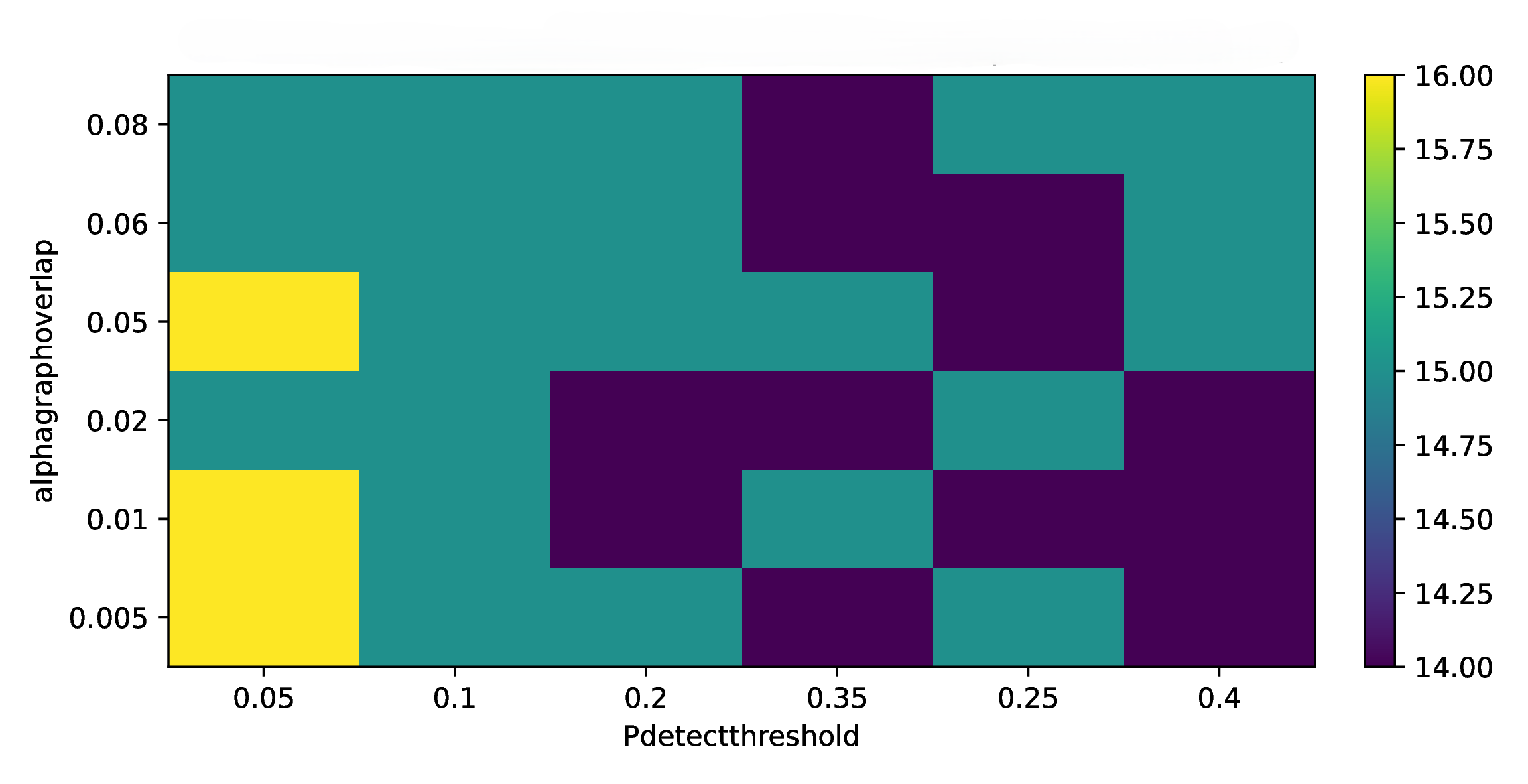}
        \label{fig:sensB_tau1}
    }
    
    \caption{Sensitivity analysis heatmaps. The top row demonstrates the impact of defender velocity ($v_{\mathrm{def}}^{\max}$), while the bottom row shows the impact of the graph edge formation threshold ($\alpha$). Left column: Breach-free reliability. Right column: Median time-to-first-breach ($\kappa_1$).}
    
    \label{fig:sensitivity_combined}
\end{figure*}
    
\section{Conclusion and Future Work}
\label{sec:conclusion}
This article addresses protected-zone (PZ) defense against coordinated incursions by small uncrewed aircraft systems (UAS), with emphasis on aerial swarm attacks under intermittent detections, partial observability, and time-varying attacker interaction structure. We proposed a closed-loop CRIDAA defense architecture that integrates uncertainty-aware swarm graph inference, risk/criticality-driven prioritization, and online defender--attacker assignment with explicit robust execution to control reassignment-induced inefficiency. We provided finite-time guarantees for the windowed assignment--execution loop via a filtration-based first-hitting-time analysis: after first detection, the first capture triggers in finite time under a uniform per-engagement success floor and non-degenerate execution, and the expected neutralization time admits an explicit mixed linear--geometric bound separating pre-detection delay from post-detection depletion in terms of defender parallel capacity, execution fraction, and capture probability. Large-scale Monte Carlo evaluation ($5000$ trials) validated the approach, achieving $85.6\%$ neutralization efficiency under probabilistic sensing and $99.9\%$ under deterministic sensing, while ablation and sensitivity studies quantified how detection gating and coordination/assignment parameters shape reliability and time-to-first-capture.

Several extensions can further strengthen capability and deployability. First, improving sensing and fusion models for cluttered environments (including intermittent detections and uncertainty quantification) can increase robustness in realistic operating conditions. Second, richer decision layers that combine principled risk modeling with data-driven components can enhance prioritization and adaptation as attacker behaviors evolve. Third, extending the framework to heterogeneous defenders/effectors and scalable distributed coordination can better reflect operational C-UAS systems. Finally, systematic validation on high-fidelity testbeds and field data, together with stress-testing across diverse threat models, will be essential to assess generalization and guide practical tuning.

\section*{Acknowledgement}
The authors acknowledge the use of Perplexity and ChatGPT large language models (LLMs) for language refinement and grammatical corrections throughout the manuscript. These LLMs were not used for generating any technical content, data, analyses, or figures.

\bibliographystyle{IEEEtran}
\bibliography{ref}

@article{c1,
  author   = {W. Chen and X. Meng and J. Liu and H. Guo and B. Mao},
  journal  = {IEEE Transactions on Vehicular Technology},
  title    = {Countering Large-Scale drone swarm attack by efficient splitting},
  year     = {2022},
  volume   = {71},
  number   = {9},
  pages    = {9967-9979},
   
}

@article{c3,
  author   = {J. Choi and M. Seo and H.-S. Shin and H. Oh},
  journal  = {IEEE Transactions on Aerospace and Electronic Systems},
  title    = {Adversarial swarm defence using multiple Fixed-Wing unmanned aerial vehicles},
  year     = {2022},
  volume   = {58},
  number   = {6},
  pages    = {5204-5219},
   
}

@article{c4,
  author   = {V. S. Chipade and D. Panagou},
  journal  = {IEEE Transactions on Robotics},
  title    = {Aerial swarm defense using interception and herding strategies},
  year     = {2023},
  volume   = {39},
  number   = {5},
  pages    = {3821-3837},
   
}

@article{c5,
  author   = {S. Velhal and S. Sundaram and N. Sundararajan},
  journal  = {IEEE Transactions on Systems, Man, and Cybernetics: Systems},
  title    = {Priority-Based DREAM approach for highly manoeuvring intruders in a perimeter defense problem},
  year     = {2024},
  pages    = {1-12},
  note     = {Early Access},
   
}

@inproceedings{c6,
  author    = {Y. E. Yao and P. Dash and K. Pattabiraman},
  booktitle = {2023 53rd Annual IEEE/IFIP International Conference on Dependable Systems and Networks (DSN)},
  title     = {SwarmFuzz: Discovering GPS Spoofing Attacks in Drone Swarms},
  year      = {2023},
  address   = {Porto, Portugal}
}

@article{c7,
  author   = {L. Yue and R. Yang and J. Zuo and Y. Zhang and Q. Li and Y. Zhang},
  journal  = {IEEE Access},
  title    = {Unmanned aerial vehicle swarm Cooperative Decision-Making for SEAD Mission: A hierarchical multiagent reinforcement learning approach},
  year     = {2022},
  volume   = {10},
  pages    = {92177-92191},
   
}

@article{c8,
  author   = {E.-Y. Yu and Y. Fu and X. Chen and M. Xie and D.-B. Chen},
  journal  = {Scientific Reports},
  title    = {Identifying critical nodes in temporal networks by network embedding},
  year     = {2020},
  volume   = {10},
  number   = {1},
   
}

@article{c9,
  author   = {Z. Zeng and W. Zhang and H. Jin},
  journal  = {Systems},
  title    = {A 'C3-TOPSIS-Pareto' Based Model for Identifying Critical Nodes in Complex Networks},
  year     = {2025},
  volume   = {13},
  number   = {2},
  pages    = {138},
   
}

@article{c10,
  author   = {D. M. Rom and E. Hwang},
  journal  = {Statistics in Medicine},
  title    = {Testing for individual and population equivalence based on the proportion of similar responses},
  year     = {1996},
  volume   = {15},
  number   = {14},
  pages    = {1489-1505},
   
}

@inproceedings{c12,
  author    = {H. Chitsaz and S. M. LaValle},
  booktitle = {2007 46th IEEE Conference on Decision and Control},
  title     = {Time-optimal paths for a Dubins airplane},
  year      = {2007},
  pages     = {2379-2384},
  address   = {New Orleans, LA, USA}
}

@article{c13,
  author   = {M. M. Tulu and R. Hou and T. Younas},
  journal  = {IEEE Access},
  title    = {Identifying Influential Nodes Based on Community Structure to Speed up the Dissemination of Information in Complex Network},
  year     = {2018},
  volume   = {6},
  pages    = {7390-7401},
   
}

@article{c14,
  author   = {N. Ye and Y. Zhang and C. M. Borror},
  journal  = {IEEE Transactions on Reliability},
  title    = {Robustness of the Markov-Chain model for Cyber-Attack detection},
  year     = {2004},
  volume   = {53},
  number   = {1},
  pages    = {116-123},
   
}

@article{c15,
  author   = {L. Strickland and M. Gombolay},
  journal  = {Journal of Aerospace Information Systems},
  title    = {Coordinating Team Tactics for Swarm-Versus-Swarm Adversarial Games},
  year     = {2024},
  volume   = {21},
  pages    = {94-113},
   
}

@article{c16,
  author   = {T. Ibuki and S. Wilson and A. D. Ames and M. Egerstedt},
  journal  = {IEEE Control Systems Letters},
  title    = {Distributed Collision-Free Motion Coordination on a Sphere: A Conic Control Barrier Function approach},
  year     = {2020},
  volume   = {4},
  number   = {4},
  pages    = {976-981},
   
}

@inproceedings{c17,
  author    = {Y. Zou and K. Chakrabarty},
  booktitle = {IEEE INFOCOM 2003. Twenty-second Annual Joint Conference of the IEEE Computer and Communications Societies},
  title     = {Sensor deployment and target localization based on virtual forces},
  year      = {2003},
  pages     = {1293-1303 vol.2},
  address   = {San Francisco, CA, USA}
}

@article{c18,
  author   = {O. Theodosiadou and D. Chatzakou and T. Tsikrika and S. Vrochidis and I. Kompatsiaris},
  journal  = {Risk Analysis},
  title    = {Real-time threat assessment based on hidden Markov models},
  year     = {2023},
  volume   = {43},
  number   = {10},
  pages    = {2069-2081},
  month    = {Oct},
   
}

@inproceedings{c19,
  author    = {S. Carpin and Y. L. Chow and M. Pavone},
  booktitle = {Proceedings of the 2016 IEEE International Conference on Robotics and Automation (ICRA)},
  title     = {Risk aversion in finite Markov decision processes using total cost criteria and average value at risk},
  year      = {2016},
  pages     = {335-342}
}

@article{c20,
  author   = {B. Hajek},
  journal  = {Advances in Applied Probability},
  title    = {Hitting-Time and Occupation-Time Bounds Implied by Drift Analysis with Applications},
  year     = {1982},
  volume   = {14},
  number   = {3},
  pages    = {502-525},
   
}

@article{c21,
  author   = {J. He and X. Yao},
  journal  = {Artificial Intelligence},
  title    = {Drift Analysis and Average Time Complexity of Evolutionary Algorithms},
  year     = {2001},
  volume   = {127},
  number   = {1},
  pages    = {57-85},
   
}

@article{c22,
  author   = {B. Doerr and D. Johannsen and C. Winzen},
  journal  = {Algorithmica},
  title    = {Multiplicative Drift Analysis},
  year     = {2012},
  volume   = {64},
  pages    = {673-697},
   
}

@article{c24,
  author   = {S. B. Connor and G. Fort},
  journal  = {Stochastic Processes and their Applications},
  title    = {State-dependent Foster-Lyapunov criteria for subgeometric convergence of Markov chains},
  year     = {2009},
  volume   = {119},
  number   = {12},
  pages    = {4176-4193},
   
}

@article{c25,
  author   = {R. Pemantle and J. S. Rosenthal},
  journal  = {Stochastic Processes and their Applications},
  title    = {Moment conditions for a sequence with negative drift to be uniformly bounded in Lr},
  year     = {1999},
  volume   = {82},
  number   = {1},
  pages    = {143-155},
   
}

@article{c26,
  author   = {T. Kötzing and M. S. Krejca},
  journal  = {Theoretical Computer Science},
  title    = {First-hitting times under drift},
  year     = {2019},
  volume   = {796},
  pages    = {51-69},
  
}

\end{document}